\documentclass[12pt]{report}
\usepackage[english,UKenglish]{babel}

	\usepackage{amsmath}
	\usepackage{amssymb}
	\usepackage{verbatim}
	\usepackage{enumerate}
	\usepackage[all]{xy}
	\usepackage{mathrsfs}
	\usepackage{amsthm}

\newcommand{\R}{{\mathbb R}}

\newcommand{\tr}{\mathrm{tr}}
\newcommand{\di}{\mathrm{div}}

\newcommand{\p}{\partial}
\newcommand{\Vol}{\textrm{Vol}}

\newcommand{\Mbar}{\overline{M}}
\newcommand{\gbar}{{\overline{g}}}
\newcommand{\kbar}{\overline{k}}

\newcommand{\pbar}{\bar \partial}
\newcommand{\phibar}{\overline{\phi}}

\newcommand{\Vbar}{\overline{V}}
\newcommand{\Kbar}{\overline{K}}
\newcommand{\Nbar}{\overline{N}}
\newcommand{\xbar}{\overline{x}}
\newcommand{\ebar}{\overline{e}}

\newcommand{\gtil}{\widetilde{g}}
\newcommand{\Deltatil}{\widetilde{\Delta}}
\newcommand{\phitil}{\widetilde{\phi}}

\newcommand{\Ltil}{\widetilde{L}}
\newcommand{\Gtil}{\widetilde{G}}

\newcommand{\sigmatil}{\widetilde{\sigma}}

\newcommand{\ra}{\rightarrow}
\newcommand{\mru}{\mathring{u}}
\newcommand{\Wtil}{\widetilde{W}}
\newcommand{\rtil}{\widetilde{r}}
\newcommand{\Jtil}{\widetilde{J}}

\newcommand{\cP}{\mathcal{P}}

\newcommand{\calL}{\mathcal{L}}
\let\Reals\R
\DeclareMathOperator{\Lap}{\Delta}
\newcommand{\loc}{\text{loc}}

\newcounter{mnotecount}[section]
\let\oldmarginpar\marginpar
\renewcommand\marginpar[1]{\-\oldmarginpar[\raggedleft\footnotesize #1]%
{\raggedright\footnotesize #1}}

\newcommand{\Number}{section} 	
\theoremstyle{plain}
	\newtheorem{thm}{Theorem}[\Number]
	\newtheorem{lem}[thm]{Lemma}
	
	\newtheorem{prop}[thm]{Proposition}
	\newtheorem{cor}[thm]{Corollary}

	\newtheorem{remark}[thm]{Remark}
	\newtheorem{defn}[thm]{Definition}
	
	\newtheorem{assumpt}[thm]{Assumption}

\title{The Einstein Constraint Equations on Asymptotically Euclidean Manifolds}

\author{James Dilts}

\date{June 15, 2015}

\usepackage[margin=1in]{geometry}

\begin{document}

\maketitle

\begin{abstract}
In this dissertation, we prove a number of results regarding the conformal method of
finding solutions to the Einstein constraint equations. These results include 
necessary and sufficient conditions for the Lichnerowicz equation to have solutions,
global supersolutions which guarantee solutions to the conformal constraint equations
for near-constant-mean-curvature (near-CMC) data as well as for far-from-CMC data,
a proof of the limit equation criterion in the near-CMC case, as well as a model
problem on the relationship between the asymptotic constants of solutions and the ADM
mass. We also prove a characterization of the Yamabe classes on asymptotically Euclidean
manifolds and
resolve the (conformally) prescribed scalar curvature problem on asymptotically
Euclidean manifolds for the case of nonpositive scalar curvatures.

Many, though not all, of the results in this dissertation have been previously
published in \cite{Dilts13b}, \cite{DIMM14}, \cite{DL14},
\cite{DM15}, and \cite{DGI15}. This article is the author's Ph.D. dissertation, except
for a few minor changes.
\end{abstract}

%%%CHAPTERS
\chapter{Introduction} \label{chap:Intro}
%%%%%%%%%%%%%%%%%%%%%%%%%%%%%%%%%%%%%%%%%%%%%%%%%%%%%%%%%%%%%%%%%%%%%%%%%%%%%%%%%%%%%%%%% 

General relativity, Albert Einstein's theory of gravity, has proven remarkably successful
in describing the universe from planetary to intergalactic scales. In this theory,
Einstein made the surprising claim that gravity is equivalent to the curvature
of space \cite{Einstein15b}. In other words, mass and energy bend and stretch space itself. 

In our solar system, the stretching is very slight; even passing over the surface of the
sun, the error is much less than one percent. However, this slight stretching has been
confirmed by numerous tests. The first physical confirmation was the orbit of Mercury.
The oval orbit of mercury precesses (rotates) by a small amount
each year. However, the observed precession is about
eight percent greater than Newtonian gravity predicts.
By taking into account the stretching of space,
Einstein \cite{Einstein15a} correctly explained the observed precession.

Some other confirmations of the accuracy of general relativity include the bending of
light in gravitational fields, the gravitational red-shift effect,
and the Shapiro time delay. General relativity has become the most accurate theory
of gravity known. It has led to remarkable technologies, such as GPS, 
and remarkable physical predictions, such as the Big Bang and black holes,
cf. \cite{Wald}.

In general relativity, the universe is described by a Lorentzian manifold, called a
spacetime. Vectors in the spacetime with
positive inner product represent space-like directions, while those with negative inner
product represent time-like directions. Those with zero inner product can be interpreted
as the directions that light can travel. As mentioned earlier, mass and energy stretch
spacetime itself. This is represented by the equations of general relativity:
\begin{equation}\label{eq:MainEinstein}
R_{\mu\nu} - \frac12 R \gamma_{\mu\nu} + \Lambda \gamma_{\mu\nu} = \kappa T_{\mu\nu}.
\end{equation} Here, $\gamma_{\mu\nu}$ is the metric, $R_{\mu\nu}$ and $R$
are the Ricci and
scalar curvatures respectively, $\Lambda$ is the cosmological constant and $\kappa$ is
a constant depending on the units chosen. In what follows, we choose units such that
$\kappa = 1$. The tensor $T_{\mu\nu}$ is the stress-energy
tensor, which combines both mass and energy into one object.

While easy to write down, the Einstein equations \eqref{eq:MainEinstein} are not easy to
solve. Minkowski space, $\R^n$ equipped with a flat Lorentzian metric,
solves them trivially. The
first non-trivial exact example is the Schwarzschild solution. This
solution is spherically symmetric in spatial (spacelike) directions, and describes
the space surrounding a star.
For many years, general relativity was a business of
finding special, symmetric solutions to the Einstein equations. Many of these are quite
important, such as the Kerr metric, which generalizes the Schwarzschild
 metric to allow
rotation and represents, it is thought, the end state of black holes, and the FLRW
 metric, a family of spatially homogeneous and isotropic solutions that are 
the basis of the Big Bang and the
standard model of cosmology. See \cite{Wald} for more information on these solutions.

In most scientific theories, one wants to be able to 
specify initial conditions, such as the location of particles and their momenta, and
then evolve the system to predict where the particles will be in the future. This is called
the initial value problem. It took about forty years before the initial value problem for
general relativity was put on a firm theoretical footing. 

In Newtonian physics, the
initial data of particles and their momenta is freely specifiable. This is not the case
in general relativity. In general relativity, initial data is (usually) given on a
Riemannian spatial submanifold of the spacetime. Using the Gauss and 
Codazzi equations, one can reduce the Einstein equations \eqref{eq:MainEinstein} on the
submanifold to the Einstein constraint equations,
\begin{subequations}\label{eq:Constraints}\begin{align}
R_g + (\tr_g K)^2 - |K|_g^2 &= T_{nn} \label{eq:ConstraintHamiltonian},\\
\nabla^i_g K_{ij} - \nabla_j(\tr_g K) &= 2 T_{in} \label{eq:ConstraintMomentum},
\end{align}\end{subequations}
 where $g$ is the induced Riemannian metric, $K$ is the second fundamental
form, $n$ is the unit normal direction, and latin indices indicate spatial directions.
Equation \eqref{eq:ConstraintHamiltonian} is
known as the Hamiltonian constraint while equation \eqref{eq:ConstraintMomentum} is
known as the momentum constraint. A standard reference on these equations is 
\cite{BI04}.

The constraint equations \eqref{eq:Constraints}
must hold on any spatial submanifold of a spacetime satisfying
the Einstein equations \eqref{eq:MainEinstein}. In 1952, Yvonne 
Choquet-Bruhat \cite{CB52} proved the
converse: given a Riemannian manifold $(M,g)$ and a symmetric 2-tensor $K$, there is a 
spacetime satisfying the Einstein equations \eqref{eq:MainEinstein} where $(M,g)$ is a
submanifold and $K$ is the second fundamental form of this submanifold. Later, 
Choquet-Bruhat and Robert Geroch \cite{CBG69} proved the existence of an appropriate 
``maximal" spacetime containing $(M,g)$, called the maximal globally hyperbolic 
development.

Due to these results, 
instead of trying to find and classify all solutions of the Einstein equations,
one may instead find and classify the solutions of the constraint equations. In addition,
the initial value problem is vital in finding solutions in complicated situations, such
as inspiraling binary black holes; exact solutions are difficult, if not impossible, to
find, but computers can approximate these solutions using the initial value formulation.

Thus one would like to understand the full set of solutions to the constraint equations
\eqref{eq:Constraints}, and, in particular, to
parameterize this set. The constraint equations are an underdetermined system of elliptic
PDEs. Roughly speaking, for an $n$ dimensional spacetime, there are $n$ functions 
determined by the constraint equations, while the rest of the data is freely specifiable.
However, it is not immediately obvious which of the quantities or components of tensors
we should attempt to specify and which we should attempt to solve for. A
useful decomposition of the data is needed. One of the most useful decompositions
is known as 
the conformal method.

\section{The Conformal Method}

The conformal method was developed by Lichnerowicz, Choquet-Bruhat, 
and York in order to
parameterize all the solutions of the constraint equations 
\eqref{eq:Constraints}. To date it has been the
most successful method in doing so.

The main idea behind the conformal method, as described in the previous section, is to
decompose the initial data into freely specifiable and determined data. Over the years,
several variations of the conformal method have been introduced, which did not appear
to be equivalent. Fortunately, David Maxwell \cite{Maxwell14} 
recently showed that all of the conformal
methods lead to the same set of solutions, and so are effectively equivalent. Indeed, 
there is a straightforward transformation of the specifiable data from any of the methods
to data from any of the others.
Because of this, we present and use the method that appears
to have the most advantages,
which Maxwell refers to as the ``conformal thin sandwich-Hamiltonian" 
formulation, or CTS-H for short.

In this method, the initial data consists not only of a Riemannian manifold $(M,\gbar)$
and a symmetric 2-tensor $\Kbar$, but also of a function $\Nbar$, called the lapse
function. When solving for the complete spacetime, the lapse function controls the 
relative length of the unit normal to the submanifold and the coordinate vector $\p_t$.
However, the $\Nbar$ found via the conformal method need not be used in finding the
spacetime; it is called the lapse function due to the derivation of the CTS formulations.

In the CTS-H formulation of Einstein's theory with matter sources,
 one first specifies a manifold $M$ and a background metric $g$.
One then chooses functions $\tau,r$, a function $N>0$ going to 1 at infinity,
a vector field $J$, and a transverse-traceless (i.e.,
divergence-free and trace-free) symmetric 2-tensor $\sigma$. We call 
$(g, \tau, N, \sigma,r,J)$
the ``seed data." One then seeks a function
$\phi>0$ and a vector field $W$ solving the conformal constraint equations:
\begin{subequations}\label{eq:ConfConst}\begin{align}
-a \Delta \phi + R \phi + \kappa \tau^2 \phi^{q-1} 
    - \left|\sigma + \frac{1}{2N} LW\right|^2 \phi^{-q-1} - r \phi^{-q/2}=0
     \label{eq:OrigLich}\\
\di \frac{1}{2N} LW = \kappa \phi^{q} d\tau + J. \label{eq:OrigVect}
\end{align}\end{subequations}
 Here, all quantities and operators are calculated relative to $g$, 
$q = \frac{2n}{n-2}$, $\kappa = \frac{n-1}{n}$, $a = \frac{4(n-1)}{n-2}$,
and $L$ is the conformal Killing operator, defined by
\begin{equation}\label{eq:CKO}
LW_{ab} = \nabla_a W_b + \nabla_b W_a - \frac2n \nabla_c W^c g_{ab}.
\end{equation} We refer to Equation \eqref{eq:OrigLich} as
the Lichnerowicz equation, while equation \eqref{eq:OrigVect} is called
the vector equation. The system is also called the 
LCBY (Lichnerowicz-Choquet-Bruhat-York) equations.

Once $(\phi,W)$ is found, the initial data solving \eqref{eq:Constraints}
is reconstructed as follows:
\begin{align} \label{eq:ReconstructedData}
\gbar_{ab} &= \phi^{q-2} g_{ab},\\
\Kbar_{ab} &= \phi^{-2}\left(\sigma_{ab} +\frac{1}{2N} LW_{ab}\right) 
          + \frac1n \tau \gbar_{ab},\\
T_{nn}     &= \phi^{-\frac32 q +1} r,\\
T_{in}     &= \phi^{-q} J.
\end{align} We make several notes on this. First, given $\phi$, York \cite{York73}
proved that
such a decomposition of $\Kbar$ exists and is unique. Indeed the decomposition is $L^2(M)$
orthogonal, i.e.,
\begin{equation}
  \int_M \langle \sigma, LW\rangle = 0,
\end{equation} though they are not in general orthogonal pointwise. Next, note that
$\tau = \tr_{\gbar} \Kbar$. Thus $\tau$ represents the mean curvature of the initial
data. Also, though we could allow $r<0$, this represents a negative energy density.
For the rest of the thesis we assume the weak energy condition, which in this case
is equivalent to saying that $r\geq 0$.

Finally, note that the metric $g$ was, in the end, only specified up to a conformal factor.
Perhaps the greatest strength of the CTS-H formulation over the other formulations of the 
conformal method is that it is conformally covariant. Specifically, we have the following
proposition, as proven in \cite[Prop 6.4]{Maxwell14}. 

\begin{prop}\label{prop:ConformalCovariance}
Let $(g, \tau, N, \sigma,r,J)$ be CTS-H seed data, and let $\psi$ be a smooth positive
function. Then
$(\phi,W)$ solve the conformal constraint equations \eqref{eq:ConfConst}
 for the data $(g,\tau, N,\sigma,r,J)$ if and only if 
$(\psi^{-1} \phi, W)$ solve the conformal constraint equations for the data 
\begin{equation}
(\psi^{q-2}g, \tau, \psi^q N, \psi^{-2}\sigma, \psi^{-\frac32 q +1} r, \psi^{-q} J).
\end{equation}
Both yield the same solution $(\gbar, \Kbar)$ of the 
constraint equations.
\end{prop}

As a consequence of this, when we attempt to find solutions to the conformal constraint
equations, we can, without loss of generality, do all calculations with respect to any
convenient representative $\gtil \in [g]$. Usually we use this freedom to choose a
representative with convenient scalar curvature. In the more traditional conformal 
method, the corresponding conformal constraint equations (essentially 
\eqref{eq:ConfConst} with $N\equiv \frac12$)
 are not conformally covariant. 

The Lichnerowicz equation \eqref{eq:OrigLich} is a semilinear elliptic PDE. If
$\sigma + \frac{1}{2N} LW, r \equiv 0$, the Lichnerowicz equation becomes the
(conformally) prescribed
scalar curvature equation. In particular, if $\phi$ solves \eqref{eq:OrigLich} with
$\sigma + \frac{1}{2N} LW, r \equiv 0$, then the scalar curvature of $\phi^{q-2}g$ is 
$-\kappa \tau^2$. The prescribed scalar curvature problem is closely related to
the solvability of the Lichnerowicz equation (and the conformal constraint equations
overall) as we see below.

The vector equation \eqref{eq:OrigVect} is a linear elliptic PDE. In the absence
of conformal Killing fields, the
vector equation is completely understood. However, in the presence of conformal Killing
fields, it is not well understood. A conformal Killing field $V$ is one such that 
$LV \equiv 0$, and represents a symmetry of some conformally related metric.
On a compact manifold, in the vacuum case (i.e., $(r,J) \equiv 0$), 
\begin{equation} \label{eq:KillingFieldsBad}
\begin{aligned}
0 &= \int_M \frac{-1}{4N} \langle LW, LV\rangle\\
  &= \int_M \left\langle V ,\di\frac{1}{2N} LW\right\rangle \\
  &= \int \kappa \phi^{q} \langle V, d\tau\rangle
\end{aligned}
\end{equation} by integration by parts. (The adjoint of $L$ is $-2 \di$.) Thus
$\phi^q d\tau$ must be $L^2$ orthogonal to all conformal Killing fields for there
to be a solution $W$ to the vector equation \eqref{eq:OrigVect}. In the case 
$d\tau \equiv 0$, i.e., the constant mean curvature (CMC) case, this is not a problem.
However, in general, since $\phi$ is also unknown, this is a serious complication. The
``drift formulation" by Maxwell \cite{Maxwell15}, described below in Section \ref{sec:Drifts},
is an extension of the CTS-H formulation that, among other things, attempts to resolve
this problem. In this thesis, we assume that the metric does not allow
any conformal Killing fields. Fortunately, it is well known \cite{BCS05}
that generic metrics do not admit any conformal Killing fields.

Earlier, we described the CTS-H conformal method as splitting the initial data into
freely specifiable and determined data. This is not precisely the case. For instance,
on a compact manifold, if $\tau$ is a constant, then
the solution to the vector equation is $W\equiv 0$. If, in addition, $R>0$ and
 $\sigma,r \equiv 0$, the maximum principle implies
that the Lichnerowicz equation \eqref{eq:OrigLich}
has no positive solution. Thus the seed data is not
freely specifiable. The goal, then, becomes to determine which seed data sets
lead to (a hopefully unique) solution of the conformal constraint equations.
If fully understood,
this leads to a parameterization of the solutions to the constraint equations
\eqref{eq:Constraints}. The case that is most
fully understood is the case where $M$ is a compact manifold without boundary.

\section{The Compact Case}

The simplest case is the constant mean curvature case, i.e., when $d\tau \equiv 0$.
In this case, the conformal constraint equations \eqref{eq:ConfConst}
decouple. If $J$ is $L^2$ orthogonal to any conformal Killing fields,
the vector equation \eqref{eq:OrigVect} has a solution, and $W$ is not dependent on $\phi$.
Thus we can reduce the conformal constraint equations to a single equation:
\begin{equation}\label{eq:ReducedLichn}
-a\Delta \phi + R\phi + \kappa \tau^2 \phi^{q-1} - \beta^2 \phi^{-q-1}  
    -r \phi^{-q/2}=0
\end{equation} where $\beta = \left| \sigma + \frac1{2N} LW\right|$.

As before, the sign of the scalar curvature
$R$ affects whether or not equation \eqref{eq:ReducedLichn}
has any solutions. This leads us to the Yamabe problem. The Yamabe problem asks whether
a metric can be conformally transformed to one with constant scalar curvature. The answer
is yes (cf. \cite{LP87}), with the sign of the target scalar curvature being 
prescribed by a conformal invariant called the Yamabe invariant. The Yamabe invariant of
a metric, $Y(g)$, is defined by
\begin{equation}\label{eq:DefYamabeInvariant}
Y(g) := \inf_{u \in C^\infty(M), u\not\equiv 0} 
   \frac{\int_M a|\nabla u|^2 + Ru^2}{\|u\|_q^2}.
\end{equation} If $Y(g)>0$, we say $g$ is Yamabe positive, and similar for Yamabe null
and negative. The resolution of the Yamabe problem says that $g$ is Yamabe positive
if and only if $g$ can be conformally transformed to a metric with constant positive
scalar curvature, and similar statements hold for Yamabe null and negative metrics.
Since the CTS-H method is conformally covariant (cf. Proposition 
\ref{prop:ConformalCovariance}), we can assume the scalar curvature $R$ is constant of
the appropriate sign.

For compact manifolds, the CMC case (with $r \equiv 0$) was completed by Jim Isenberg
in \cite{Isenberg95}. The case $r\geq 0$ is proven essentially the same way, and so we
include it here also. The solvability of equation \eqref{eq:ReducedLichn} is detailed in
Table \ref{table:CMCSolvability}.

\begin{table}[h]
\caption{Solvability of the CMC Conformal Constraint Equations}
\centering
\begin{tabular}{|c|c|c|c|c|}
\hline
       & $\tau=0, \beta,r \equiv 0$ & $\tau=0, \beta,r \not\equiv 0$ &
         $\tau\neq 0, \beta,r\equiv 0$ & $\tau\neq 0, \beta,r\not\equiv 0$  \\ \hline
$Y(g)>0$ & No   & Yes & No  & Yes \\ \hline
$Y(g)=0$ & Yes* & No  & No  & Yes \\ \hline
$Y(g)<0$ & No   & No  & Yes & Yes \\ \hline
\end{tabular}
\label{table:CMCSolvability}
\end{table}

In all cases the solution $\phi$ to equation \eqref{eq:ReducedLichn} is unique, except
in the case $Y(g) = 0$, $\tau=0$ and $\beta,r\equiv 0$ (marked with a $*$), 
in which case there is a one 
parameter homothety family of solutions, namely $\phi \in \R^+$. This gives a complete
parameterization of the CMC solutions of the constraint equations
\eqref{eq:Constraints}, cf. 
\cite{Isenberg87}. Similar results have been found for other topologies
and asymptotic conditions. 

If every spacetime could be evolved from CMC initial data, Isenberg's work
would be enough to parameterize the solutions of the Einstein equations.
Unfortunately, not all spacetimes can be obtained this way, as proved 
in \cite{CIP05}.
It is not known whether or not generic spacetimes can be obtained from CMC initial data.
Thus for a complete parameterization of solutions to the Einstein equations, we must
consider the conformal constraint equations with generic mean curvature $\tau$. 

We first consider the Lichnerowicz equation \eqref{eq:OrigLich}. Unsurprisingly, the 
solvability of the Lichnerowicz equation mirrors the solvability of the CMC conformal
constraint problem as tabulated in Table \ref{table:CMCSolvability}, but with one caveat.
On a compact manifold, if $Y(g)<0$, the Lichnerowicz equation has a solution 
if and only if $g$ can be
conformally transformed to a metric with scalar curvature $-\kappa \tau^2$. For
$\tau^2 >0$, this is always true. For $\tau$ with zeroes, the solvability is discussed in
\cite{Rauzy95}, \cite{DM15} and Chapter \ref{chap:Yamabe} below. This problem is completely
understood.

Excepting that caveat, one might expect the solvability of the generic 
conformal constraint equations to mirror
that of the CMC case, as shown in Table \ref{table:Solvability}. 

\begin{table}[h]
\caption{Hypothesized Solvability of the Conformal Constraint Equations}
\centering
\begin{tabular}{|c|c|c|}
\hline
       & $\tau\not\equiv 0, \beta,r\equiv 0$ & $\tau\not\equiv 0, \beta,r\not\equiv 0$  \\ \hline
$Y(g)>0$ & No  & Yes \\ \hline
$Y(g)=0$ & No  & Yes \\ \hline
$Y(g)<0$ & Yes & Yes \\ \hline
\end{tabular}
\label{table:Solvability}
\end{table}

For nearly CMC data, this solvability is realized, at least in the case where
there are no conformal Killing fields. The near-CMC conditions typically
come in two flavors. If $\tau$ is a constant for which the
CMC conformal constraint equations have a solution,
then the inverse function theorem can be used to show that any nearby $\tau$ (in
$W^{1,p}_{\delta-1}$) also leads to a solution. In the second case, the condition
is that  $\|d\tau\|_{p}$ is sufficiently small compared to $\inf \tau$. 
Using these types of conditions, the Yamabe negative near-CMC case was 
settled in 1996 \cite{IM96}, the nonexistence cases for Yamabe nonnegative 
metrics in 2004 \cite{IOM04}, and the remaining cases in 2008 \cite{ACI08}. (These 
results prove results for manifolds with scalar curvature of a strict sign; Maxwell's
conformal covariance of the CTS-H formulation \cite{Maxwell14} is needed to make them
apply to the entire Yamabe classes.) 
All of these results rely on there being no conformal Killing fields, for the
reasons discussed above. 

The only generic result known for the arbitrary mean curvature case was proven by Holst,
Nagy, and Tsogtgerel \cite{HNT09}, then improved by Maxwell \cite{Maxwell09}.
This result essentially
says that on a Yamabe positive compact manifold,
given an arbitrary $\tau$, if $\sigma,r,$ and $J$ are small enough, then the conformal
constraint equations \eqref{eq:ConfConst} have a (not necessarily
unique) solution. More recently, however, Nguyen \cite{Nguyen14} showed that all such solutions are
merely rescalings of perturbations off of the maximal ($\tau\equiv 0$) case. We discuss
this in Corollary \ref{cor:FarIsNear}.
Thus the only generic far-from-CMC result known is, essentially, a near-CMC result.

Another attempted method to find solutions to the conformal constraint equations is the
``limit equation" criterion. First explored by Dahl, Gicquaud, and Humbert \cite{DGH11},
this method says that either the conformal constraint equations or the limit equation
\begin{equation}\label{eq:LimitEquation}
\di \frac{1}{2N} LW = \alpha_0 \sqrt{\kappa}\left|LW\right| \frac{d\tau}{2N\tau}
\end{equation} (for some $\alpha_0\in (0,1]$) have a (nontrivial) solution.
As was suspected, Nguyen recently showed \cite{Nguyen14} that both can in fact have solutions. 

The limit equation was originally found via a subcriticality argument. If the exponent
of $\phi$ in the vector equation \eqref{eq:OrigVect} is reduced by epsilon, the 
coupling of the conformal constraint equations is weak enough so that solutions are
relatively simple to find. As $\epsilon \to 0$, if these subcritical solutions are
bounded, they must converge to a solution to the conformal constraint equations. If
they are instead unbounded, it can be shown that they converge to a solution of the 
limit equation \eqref{eq:LimitEquation}.

The limit equation criterion is that if the limit equation has no solutions, then
the conformal constraint equations must have a solution. While this method may be used
for the far-from-CMC case, so far, it has only been used to find solutions in the 
near-CMC case, as in \cite{DGH11}.

In order to better explore the far-from-CMC regime, Maxwell studied a model problem
with high symmetry \cite{Maxwell11}. He studied seed data on $T^n$ with the flat metric,
where the data depended on only one coordinate (i.e., with $U^{n-1}$ symmetry).
For some particular data, he showed
that if $\tau$ was sufficiently far-from-CMC in some sense, then there were no
solutions to the conformal constraint equations (in the symmetry class of the data).
Given such a $\tau$, however, if the transverse traceless part of the data were
sufficiently small, then there were at least two solutions. 

This is in contrast to
the CMC and near-CMC case, where solutions are unique. Also, the
nonexistence for far-from-CMC data is in contradiction with the hypothesized solvability
described in Table \ref{table:Solvability}. However, this model problem may be a special
case for several reasons. First, the background metric is flat, which is known to be a
very special case, even in the CMC theory. Second, the background metric has conformal
Killing fields, which is known to be non-generic. Third, the mean curvature function
$\tau$ has jump discontinuities. Finally, the non-existence/non-uniqueness only occurs when
$\tau$ changes signs. However, this could also represent new phenomena, or, perhaps, 
limitations of the conformal method.

Maxwell \cite{Maxwell14b} later studied a related problem. On flat $T^n$, and arbitrary $\tau$,
again with $U^{n-1}$ symmetry, 
he found seed data that led to either flat Kasner or static-toroidal solutions of the
constraint equations. It was shown that there was in fact a one parameter family of
solutions to the conformal constraint equations if and only if
\begin{equation}\label{eq:TauStarDef}
\tau^* := \frac{\int_{S^1} N \tau dx}{\int_{S^1} N dx} = 0,
\end{equation} where the integrals are respect to the flat metric. 
While $\tau^*$ appears to be determined by $N$ and $\tau$, the reality
is more complicated.

Recall that the CTS-H method is conformally covariant, as described in Proposition 
\ref{prop:ConformalCovariance}. However, equation \eqref{eq:TauStarDef} is not conformally
covariant. Thus, if we started with arbitrary seed data on $T^n$ that happened to lead
to one of these solutions, we would need to calculate $\tau^*$ with respect to the solution
metric and not the background metric. In general, then, there is no way of determining
whether or not a set of seed data leads to a one parameter family until after the solution
has already been found. This presents serious problems for the goal of parameterizing all
solutions to the constraint equations 
\eqref{eq:Constraints}, since these one parameter
families are essentially impossible to detect.

\section{The Drift Formulation}\label{sec:Drifts}

In an attempt to avoid the pitfalls for parameterizing solutions to the constraint
equations described in the last section, Maxwell introduced the drift formulation of the
conformal method, originally in \cite{Maxwell14b}, and expanded in \cite{Maxwell15}. In the standard
CTS-H method, the mean curvature $\tau$ is specified in the seed data, and is unchanged
by the conformal factor found by solving the Lichnerowicz equation \eqref{eq:OrigLich}.
This, however, makes calculating $\tau^*$ impossible without first finding the solution
to the conformal constraint equations \eqref{eq:ConfConst}.

Maxwell's idea was to specify the constant $\tau^*$ directly, and then define $\tau$ by
adding $\tau^*$ to a conformally varying term, given by something he calls
a drift, for reasons explained in \cite{Maxwell15}. Since $\tau^*$ is specified directly,
the one parameter families of solutions described in the previous section occur only
when $\tau^*$ is specified to be zero. The drift 
formulation also has the advantage of making it possible to find solutions even in the
presence of conformal Killing fields.

Though Maxwell introduces several possible ways to construct such $\tau$, we 
discuss only one. In this formulation, which he calls CTS-H with volumetric drift,
the drift is given by a vector field $V$ which is specified up to a conformal Killing
field $Q$. Given seed data $(g, \tau^*, V, N, \sigma,r,J)$, one tries to find a solution
$(\phi, W, Q)$ to  
\begin{subequations}\label{eq:DriftEqns}\begin{align}
-a\Delta \phi + R \phi 
  + \kappa \left(\tau^* + \frac{\phi^{-2q}}{N} \di(\phi^q(V +Q))\right)^2 \phi^{q-1} 
    - \left|\sigma + \frac{1}{2N} LW\right|^2 \phi^{-q-1} - r \phi^{-q/2}=0
     \label{eq:DriftLich}\\
\di \frac{1}{2N} LW 
  = \kappa \phi^{q} d\left(\tau^* + \frac{\phi^{-2q}}{N} \di(\phi^q(V +Q))\right)
     + J. \label{eq:DriftVect},
\end{align}\end{subequations}
 which is the same as \eqref{eq:ConfConst}, except that
we replaced $\tau$ with
\begin{equation}\label{eq:TauDriftDef}
\tau := \tau^* + \frac{\phi^{-2q}}{N} \di(\phi^q(V +Q)).
\end{equation} The data is then reassembled as before, except
\begin{equation}
\Kbar_{ab} = \phi^{-2}\left(\sigma_{ab} +\frac{1}{2N} LW_{ab}\right) 
          + \frac1n \left(\tau^* + \frac{\phi^{-q}}{N} \overline{\di} V\right) \gbar_{ab},
\end{equation}  where $\overline{\di}$ is the divergence with respect to $\gbar$.
Note that for $V\equiv 0$, this method reduces to the CMC CTS-H method.

The drift method has several advantages over the CTS-H formulation.
First, the one parameter
families found in \cite{Maxwell14b} occur if and only if $\tau^* = 0$. Thus at least 
that obstruction to parameterization is overcome. Also, Maxwell proved 
\cite[Thm 10.1]{Maxwell15} that the vector equation \eqref{eq:DriftVect} has a solution for
some $Q$, even if $g$ has conformal Killing fields. Thus it becomes possible to solve
the constraint equations in the presence of conformal Killing fields.

However, the drift formulation equations \eqref{eq:DriftEqns} are
much more complicated analytically. For example, in the original vector equation
\eqref{eq:OrigVect},
one can find an upper bound on $LW$ based on an upper bound for $\phi$. In the drift
vector equation, a similar upper bound naively requires bounds on $\|\phi\|_{C^2}$. 
Perturbation methods, such as those used to produce solutions in the CTS-H formulation,
are expected to extend to the drift setting \cite{PC15}.
Since all known generic results for the CTS-H method are near-CMC results,
this would show that the drift formulation is at least as useful as the CTS-H method. The
drift formulation is a promising approach to finding the parameterization of the 
constraint equations.

\section{This Dissertation}

In this dissertation, we discuss the conformal constraint equations, in particular focusing
on the asymptotically Euclidean (AE) case. In Chapter \ref{chap:AEIntro}, we 
introduce AE manifolds and the appropriate Banach spaces for analysis, and then discuss
elliptic operator theory on AE manifolds. Because the vector equation \eqref{eq:OrigVect}
is relatively simple, we discuss its solvability in this chapter.
This chapter serves as a common introduction to all the subsequent chapters.

In Chapter \ref{chap:IFF}, we discuss the solvability of the Lichnerowicz equation
\eqref{eq:OrigLich}. In particular, we show that the Lichnerowicz equation is
solvable if an only if the metric can be conformally transformed to one with 
scalar curvature $-\kappa \tau^2$. We 
then leverage this result to obtain a circumstance where the conformal constraint
equations do not admit a solution, and also show an example of the blowup of solutions. The
results in this chapter will appear in \cite{DGI15}.

In Chapter \ref{chap:Yamabe} we discuss when the prescribed scalar curvature problem
from the previous chapter has a solution. We give a necessary and sufficient
condition for the problem to have a solution; namely, that the zero set of the prescribed
scalar curvature has positive Yamabe invariant,
as defined in this chapter. Because of this problem's close relation to AE Yamabe classes, we
also give a characterization of the AE Yamabe classes. This chapter is taken from
\cite{DM15}.

In Chapter \ref{chap:FarFromCMC}, we prove the existence of solutions for the
conformal constraint
equations for arbitrary mean curvature, assuming the tensor $\sigma$ and the matter
 terms $r$ and $J$ are sufficiently
small. We also show existence in the near-CMC case. This part is taken from
\cite{DIMM14}. We also discuss a new solvability criterion, related to Nguyen's
``local supersolution" from \cite{Nguyen14}.

In Chapter \ref{chap:LimitEquation},
we discuss the limit equation criterion for AE manifolds.
Unfortunately, we only show that the solution of the limit equation is nontrivial in the
near-CMC case. We show that arbitrarily near-CMC data, in the sense required,
does not ever occur. This is 
unpublished joint work with Romain Gicquaud and Jim Isenberg.

In Chapter \ref{chap:Mass}, we discuss the relation of the ADM mass to the
asymptotics of the solution of the conformal constraint equations 
\eqref{eq:ConfConst}. We then present a model problem for this
relation, which shows that the ADM mass is not monotonically dependent on the asymptotics
of the solution. 

\chapter{Asymptotically Euclidean Manifolds} \label{chap:AEIntro}
%%%%%%%%%%%%%%%%%%%%%%%%%%%%%%%%%%%%%%%%%%%%%%%%%%%%%%%%%%%%%%%%%%%%%%%%%%%%%%%%%%%%%%%%% 

Perhaps the simplest solution to the Einstein constraint equations \eqref{eq:MainEinstein}
is Euclidean space,
with the second fundamental form $K$ and stress-energy tensor $T$ vanishing.
Physically this represents space with no matter and no tidal forces.
Heuristically, far from any mass and energy, space should become more and more
like Euclidean space. Far from any star, gravity becomes very weak. 
Mathematically, this kind of initial data is represented by asymptotically Euclidean
(AE) manifolds.

A manifold $(M^n,g)$ is called asymptotically Euclidean (AE) if there exists a compact 
set $K\subset M$ such that $M\setminus K$ is a (finite) collection of components $E_i$,
each diffeomorphic to the exterior of a ball in Euclidean space, $\R^n \setminus B_R(0)$,
and on each end, $g$ is asymptotic to the Euclidean metric $g_{Euc}$.
The $E_i$ are called the ends of $M$. 

In order to be precise, we must first define appropriate weighted Sobolev and H\"older
norms. First, fix a Euclidean coordinate system on each end, i.e., a distinguished
diffeomorphism from $E_i$ to $\R^n \setminus B_R(0)$. Let $\rho\geq 1$ be a
smooth function which agrees with the radial coordinate on each end. We say a function
$f \in W^{s,p}_\delta(M)$ if 
\begin{equation} \label{eq:SobolevNorm}
\sum_{|j|\leq s} \|\rho^{-\delta - \frac{n}{p} + |j|} \nabla^j f \|_{L^p} < \infty,
\end{equation} where $j$ is a multi-index, and $\nabla^j$ is calculated with respect
to a frame agreeing with the Euclidean frame on each end.
We denote this quantity by $\|f\|_{W^{s,p}_\delta(M)}$, without the $(M)$ if the
domain is understood. If $s=0$, we denote the space as $L^{p}_{\delta}(M)$ and the norm
as $\|f\|_{p,\delta}$. To extend this space and norm to tensors, we require the same
regularity and decay for each component of the tensor with respect to 
the Euclidean frame
in the background Euclidean metric $g_{Euc}$. Note that our convention on $\delta$
 is chosen such
that $f\in W^{s,p}_\delta(M)$ implies that, using the little-o notation, $f$ is
$o(\rho^{\delta})$; other conventions exist in the literature.

For $\alpha \in [0,1]$, we say a function $f\in C^{s,\alpha}_\delta(M)$ if
\begin{equation}\label{eq:HolderNorm}
\sup_{B, |j|\leq s}\left\{|\nabla^j f| \rho^{|j|-\delta}, 
   [\nabla^s f]_{B;\alpha}\rho^{s-\delta}\right\} < \infty,
\end{equation} where the supremum is over all balls $B\subset M$ of unit radius, and 
$[\cdot]_{B;\alpha}$ is the H\"older seminorm on that ball. We denote $C^{s,0}_\delta(M)$ by
$C^s_\delta(M)$ for simplicity. We extend this space to tensors similarly.

We then say that $g$ is a $W^{s,p}_\delta$ AE manifold if $\delta <0$ and
\begin{equation}
g|_{E_i} - g_{Euc} \in W^{s,p}_\delta
\end{equation} on each end $E_i$. We similarly define $C^{s,\alpha}_\delta$ AE manifolds.

We state some basic properties of these spaces in the following two propositions,
the first of which is taken from \cite[Lemma 1]{Maxwell05b}:
\begin{prop}[Properties of Weighted Sobolev Spaces]
\label{prop:SobolevEmbeddings} The following properties hold for the weighted Sobolev
spaces defined by \eqref{eq:SobolevNorm}:
\begin{enumerate}
\item If $p\geq q$ and $\delta'<\delta$ then
$L^p_{\delta'}\subset L^q_{\delta}$ and the inclusion is
continuous.

\item For $s\ge1$ and $\delta'<\delta$ the inclusion
$W^{s,p}_{\delta'}\subset W^{s-1,p}_{\delta}$ is compact.

\item If $s<n/p$ then $W^{s,p}_{\delta}\subset L^r_{\delta}$
where $r=np/(n-sp)$.  If $s=n/p$  then
$W^{s,p}_{\delta}\subset L^r_{\delta}$ for all $r\ge p$.
If $s>n/p$ then $W^{s,p}_{\delta}\subset C^{0}_{\delta}$.
These inclusions are continuous, and the last is compact.

\item If $m\le\min(j,s)$, $p\le q$, $\epsilon>0$, and $m<j+s-n/q$, then
multiplication is a continuous bilinear map from
$W^{j,q}_{\delta_1}\times W^{s,p}_{\delta_2}$ to
$W^{m,p}_{\delta_1+\delta_2+\epsilon}$ for any $\epsilon>0$. In particular,
if $s>n/p$ and $\delta<0$, then $W^{s,p}_{\delta}$ is an algebra.
\end{enumerate}
\end{prop}

\begin{prop}[Properties of Weighted H\"older Spaces]
\label{prop:HolderEmbeddings} The following properties hold for the weighted Holder
spaces defined by \eqref{eq:HolderNorm}:
\begin{enumerate}
\item If $s+\alpha \geq s'+\alpha'$, $\alpha\neq1$, and $\delta\leq \delta'$
       then the inclusion $C^{s,\alpha}_\delta \subset C^{s',\alpha'}_{\delta'}$ is
       continuous.

\item If $s+\alpha> s'+\alpha'$ and $\delta<\delta'$ then the inclusion
      $C^{s,\alpha}_\delta \subset C^{s',\alpha'}_{\delta'}$ is compact.

\item Assume $s+\alpha\leq s'+\alpha'$. Then multiplication is a continuous bilinear map
      from $C^{s,\alpha}_\delta\times C^{s',\alpha'}_{\delta'}$ to 
      $C^{s,\alpha}_{\delta+\delta'}$. In particular, if $\delta\leq0$, then 
      $C^{s,\alpha}_{\delta}$ is an algebra.
\end{enumerate}
\end{prop}

The standard Poincar\'e and Sobolev inequalities on $\R^n$, with appropriately chosen
weights, also hold on AE manifolds, as shown in \cite[Lem 2.1]{DM15}.

\begin{lem}\label{lem:poincare}
There exist constants $c_1,c_2$ such that
\begin{equation}\label{eq:poincare}
\|\nabla u\|_{p, -n/p} \ge c_1 \|u\|_{p,1-n/p}
\end{equation}
\begin{equation}\label{eq:sobolev}
\|\nabla u\|_2 \ge c_2 \|u\|_{q}
\end{equation}
for all $u\in u\in W^{1,2}_{\delta^*}(M)$ and $p\in[1,n)$.
\end{lem}

We refer the reader to \cite{Bartnik86}
for further properties of weighted Sobolev spaces.

\section{Elliptic Operators}

Elliptic operator theory on AE manifolds is well established, going back at least to
\cite{McOwen79}. For more references, see also 
\cite{Maxwell05b}, \cite{CBIY00}, and appendix C in \cite{CMP12}. While more general
results are available, we focus our attention on the Laplacian and vector Laplacian.
The following result is adapted from \cite{DIMM14}.

\begin{prop}
\label{prop:AELinearExistence}
Suppose $(M^n, g)$ is a $W^{s,p}_\gamma$ AE manifold, with $s\geq 2$, $s>n/p$,
and $\gamma<0$. Suppose $V \in W^{s-2,p}_{\gamma-2}$. Let $\cP$ be either the operator
$-a\Delta + V$ or the operator $\di \frac{1}{2N} L$, where $L$ is the conformal Killing
operator \eqref{eq:CKO}. Then for $\delta\in(2-n,0)$ the operator
\begin{equation}\label{eq:POperator1}
\cP: W^{s,p}_\delta \to W^{s-2,p}_{\delta-2}
\end{equation} is Fredholm of index zero, and
\begin{equation}\label{eq:PEstimate1}
\|u \|_{W^{s,p}_{\delta}} \leq C\left(\|P u\|_{W^{s-2,p}_{\delta-2}}+\|u\|_{L^p_{\delta'}}\right)
\end{equation} holds for some $C>0$, any $\delta'$ and all
$u \in W^{s,p}_\delta$. The map
\eqref{eq:POperator1} is an isomorphism if and only if $\cP$ has trivial null space
in $W^{s,p}_\delta$. If \eqref{eq:POperator1} is an isomorphism, then the estimate
\eqref{eq:PEstimate1} can be strengthened to 
\begin{equation}\label{eq:PEstimate2}
\|u \|_{W^{s,p}_{\delta}} \le C\| \cP u\|_{W^{s-2,p}_{\delta-2}}.
\end{equation}

Similarly, if $(M,g)$ is a $C^{s,\alpha}_\gamma$ AE manifold, then
\begin{equation}
\cP: C^{s,\alpha}_\delta \longrightarrow C^{s-2,\alpha}_{\delta-2}
\end{equation}
is Fredholm of index zero, with a corresponding a priori estimate. If $\cP$ has trivial
nullspace in $C^{s,\alpha}_\delta$, then there exists a constant $C > 0$ such that
\begin{equation}\label{eq:PEstimate3}
\|u\|_{C^{s,\alpha}_\delta} \leq C \| \cP u\|_{C^{s-2,\alpha}_{\delta-2}}.
\end{equation}
\end{prop}

Because Proposition \ref{prop:AELinearExistence} requires $s\geq 2$, 
and for simplicity, we will assume
$s =2$ for the rest of this paper, unless mentioned otherwise. 
In other words, $g$ is either a 
$W^{2,p}_\gamma$ or $C^{2,\alpha}_\gamma$ AE manifold. In any case,
if $g\in W^{s,p}$ for $s>2$ and $s>n/p$, Sobolev embedding 
\ref{prop:SobolevEmbeddings}, implies that $g\in W^{2,p'}$ with $p'>n/2$.

We now prove two maximum principles, taken from \cite{Maxwell05b}.

\begin{prop}[A Maximum Principle for AE Manifolds]
\label{prop:MaxPrinciple}
Suppose $(M,g)$ and $V$ are as in Proposition \ref{prop:AELinearExistence},
and suppose $V\geq 0$. Suppose $u \in W^{2,p}_\delta$ for some 
$\delta<0$. If 
\begin{equation}
-a\Delta u + V u \geq 0,
\end{equation}
then $u\geq 0$.
\end{prop}
\begin{proof}
Let 
\begin{equation}
v = (u + \epsilon)^- := \min\{0,u+\epsilon\}
\end{equation} for some $\epsilon>0$. Since $u \to 0$ on each end,
we see that $v$ is compactly
supported. By Sobolev embedding, $v\in W^{1,2}$ as well. Since
\begin{equation}
\int_M a|\nabla v|^2 = \int_M -av \Delta u \leq \int_M -Vu v \leq  0,
\end{equation} we know $u\geq -\epsilon$. Letting $\epsilon \to 0$,
we find $u \geq 0$.
\end{proof}

\begin{prop}[A Strong Maximum Principle for AE Manifolds]\label{prop:StrongMaxPrinciple}
Suppose $(M,g)$ and $V$ are as in Proposition \ref{prop:AELinearExistence}. Suppose
$u \in W^{2,p}_{loc}$ is nonnegative and satisfies
\begin{equation}
-a\Delta u + Vu \geq 0.
\end{equation} If $u(x) = 0$ somewhere, then $u \equiv 0$.
\end{prop}
\begin{proof}
Suppose $u(x) = 0$. The
weak Harnack inequality from \cite{Trudinger73} applies to $u$; i.e., for
some radius $R$ sufficiently small and some exponent $q$ sufficiently large, there exists
$C>0$ such that
\begin{equation}
\|u\|_{L^q(B_{2R}(x))} \leq C \inf_{B_R(x)} u = 0.
\end{equation} Thus $u$ vanishes on a neighborhood of $x$, and a connectivity argument
shows that $u \equiv 0$.
\end{proof}

In order to discuss when the operator $\cP = -a\Delta +V$ is an isomorphism,
we need to first 
discuss the Yamabe invariant. The Yamabe invariant on AE manifolds is defined similarly
to how it is defined on compact manifolds (cf. Equation \eqref{eq:DefYamabeInvariant}),
except that the test functions must have compact support. Precisely,
\begin{equation}\label{eq:DefYamabeInvariantAE}
Y(g) := \inf_{u \in C^\infty_0(M), u\not\equiv 0} 
   \frac{\int_M a|\nabla u|^2 + Ru^2}{\|u\|_q^2},
\end{equation} where $C^\infty_0(M)$ represents smooth functions with compact support.
As before, we say $g$ is Yamabe positive if $Y(g)>0$. The Yamabe classes on AE
manifolds appear to behave very differently than the Yamabe classes 
on compact manifolds, but they are in
fact equivalent to each other in some sense. This idea is discussed
further in Chapter \ref{chap:Yamabe}.
We now can prove the following isomorphism theorem.

\begin{prop}\label{prop:Isomorphism}
The operator $-a\Delta +V$ from Equation \eqref{eq:POperator1} is an isomorphism either
if $V\geq 0$ or if $V=R_g$ and $g$ is Yamabe positive. 

No $C^{2,\alpha}_\gamma$ AE manifold allows a conformal Killing field in 
$C^{2,\alpha}_\delta$. Also, if $p>n$, no $W^{2,p}_\gamma$ AE manifold allows a 
conformal Killing field in $W^{2,p}_\delta$.
Thus the operator $\di \frac{1}{2N} L$ is always an isomorphism on $C^{2,\alpha}_\delta$,
and is an isomorphism on $W^{2,p}_\delta$ if $p>n$ or if the metric admits no conformal
Killing fields.
\end{prop}
\begin{proof}
Suppose $V\geq 0$ and $(-a\Delta+V)u = 0$. By the maximum principle 
\ref{prop:MaxPrinciple}, $u\geq0$
and $-u\geq 0$. Thus $u\equiv 0$, and so $-a\Delta +V$ is an isomorphism. 

Suppose $V= R_g$ and $g$ is Yamabe positive, but that $-a\Delta+R$ is not
an isomorphism. Then there exists a nontrivial solution 
$u\in W^{2,p}_\delta \subset C^{0,\alpha}_\delta$ solving $(-a\Delta+R)u = 0$. 
By Sobolev embedding 
\ref{prop:SobolevEmbeddings}, $u \in L^q$. Integration by parts implies that
$\int a|\nabla u|^2 + R u^2 = 0$. Estimating $u$ by smooth functions with compact 
support, we find that $Y(g) =0$ (see the definition \eqref{eq:DefYamabeInvariantAE}),
which is a contradiction.

The facts about conformal Killing fields are found in \cite{Maxwell05b}. That the kernel
of $\di \frac{1}{2N} L$ is the set of conformal Killing fields follows from the 
calculation \eqref{eq:KillingFieldsBad}.
\end{proof}

On compact manifolds, the constants are harmonic functions; i.e., they satisfy 
$\Delta u = 0$. On AE manifolds, the constants are harmonic functions, but in addition,
 if for each end $E_i$, we
specify a constant $u_i$, there is a unique harmonic function $u$ such that 
$u-u_i \in W^{2,p}_\delta(E_i)$ (cf. \cite[Lem 4.1]{DIMM14}). 
We introduce the following notation.

\begin{defn}\label{defn:AsymptoticFunction}
For any set of constants $u_i$, the ``asymptotic function" $\mru$ is the unique harmonic
function such
that $\mru \to u_i$ on $E_i$. Such a function has the same regularity as the metric;
i.e., if $(M,g)$ is a $W^{2,p}_\gamma$ AE manifold, $\mru \in W^{2,p}$ and 
$\mru - u_i \in W^{2,p}_\delta(E_i)$ for any $\delta \in (2-n, 0)$. The existence of
such a function is guaranteed by \cite[Lem 4.1]{DIMM14}. When we refer to
$\mru$ we do not mention the constants $u_i$. 
\end{defn}

\begin{cor}
\label{lem:HarmonicBoundedness}
The function $\mru$ satisfies $\min_i u_i \leq \mru \leq \max_i u_i$ with
equality if and only if $\min_i u_i = \max_i u_i$.
\end{cor}
\begin{proof}
The maximum principle \ref{prop:MaxPrinciple} implies that
$\min_i u_i \leq \mru \leq \max_i u_i$. If $\min_i u_i = \mru$ somewhere,
the strong maximum principle \ref{prop:StrongMaxPrinciple} implies that
$\min_i u_i = \sup u = \max_i u_i$.
If $\min_i u_i = \max_i u_i$, all the $u_i$ are the same, and so
$\mru \equiv u_i$ is the desired harmonic function.
\end{proof}

When searching for solutions
to the Lichnerowicz equation \eqref{eq:OrigLich}, there is no reason to restrict ourselves
to $\phi$ such that $\phi\to 1$ on each end, since $\phi$ approaching any other constant
simply scales the Euclidean coordinates on that end. Indeed, as we
see in Chapter \ref{chap:FarFromCMC}, assuming $\phi$ approaches some other
constant can assist in finding solutions. Thus, we generally assume $\phi \to \mru$
on each end.

While the operator $\di \frac{1}{2N} L$ appears in a linear equation
\eqref{eq:OrigVect}, the Lichnerowicz equation \eqref{eq:OrigLich} is semilinear,
and so Proposition \ref{prop:Isomorphism} is not sufficient to find solutions of
this equation. A 
useful tool for finding solutions to semilinear equations is the method of sub and
supersolutions.

Consider the nonlinear problem
\begin{equation}\label{eq:Semilinear}
-a\Delta u = f(x,u)
\end{equation}
for a function $f(x,y):M\times \R\to \R$ which takes the form 
$f(x,y) = \sum_{i=1}^j a_i(x) y^{b_i}$ for specified functions $a_i$ and constants
$b_i$, where we use the convention that $y^{b_i} \equiv 1$ if $b_i =0$.
We also assume that $a_i(x) \in L^{p}_{\delta-2}$ for some $\delta<0$ (or, similarly,
that $a_i(x) \in C^{0,\alpha}_{\delta-2}$). Note that, depending on the value(s) of $b_i$,
$y^{b_i}$ is smooth on $(0,\infty)$, $[0,\infty)$, or $(-\infty,\infty)$. 
We say a function
$f$ is `` regular" if it satisfies these properties, and the largest interval for which
all the $y^{b_i}$ are smooth is $f$'s ``interval of regularity" $I$. Note that the
Lichnerowicz
equation \eqref{eq:OrigLich} takes this form, as long as we require sufficient regularity
of the seed data. Recall that $u_-$ is called a subsolution
of \eqref{eq:Semilinear} if $-a\Delta u_- \leq f(x,u_-)$, and similarly (with $\geq$ 
replacing $\leq$) for a supersolution $u_+$.

\begin{thm}[Sub and Supersolution Theorem for AE Manifolds]
\label{thm:SubSupersolutionTheorem}
Let $(M,g)$ be a $W^{2,p}_\gamma$ AE manifold with $p>n/2$ and
$\gamma<0$. Suppose $f(x,y)$ is regular (as defined above) for some $\delta \in (2-n,0)$.
Suppose that there are sub and supersolutions $u_\pm\in L^\infty$ such that
$u_-\leq u_+$ and $\inf u_- \in I$. Suppose $\mru$ is such that,
sufficiently far
out on each end, $u_- \leq \mru \leq u_+$.
Then Equation \eqref{eq:Semilinear} admits a solution $u$ such that
$u_- \leq u \leq u_+$ and $u- \mru \in W^{2,p}_{\delta}$.

A similar theorem holds for $C^{2,\alpha}_\gamma$ AE manifolds if $f$ is 
$C^{0,\alpha}_\delta$ regular and $u_\pm \in C^{0,\alpha}$.
The solution then satisfies $u-\mru \in C^{2,\alpha}_{\delta}$.
\end{thm}
\begin{remark}
This is essentially Theorem 1 in Appendix B.2. in \cite{CBIY00}, but with lower regularity
requirements, and generalized asymptotics. We mirror their proof. Note that
if some $b_i <0$, the theorem requires $u_->\epsilon>0$ for some $\epsilon>0$.
\end{remark}
\begin{proof}
We only prove the Sobolev case. The H\"older case is proven similarly.

We construct a solution by induction, starting from $\phi_-$. Let $k$ be a positive
function on $M$ such that $k \in L^p_{\delta}$ and
\begin{equation}\label{eq:kDef}
k(x) + \sup_{y \in \textrm{Range}(u_\pm)} f_y(x,y) \geq 0.
\end{equation} Such a $k$
exists by our assumptions on $u_\pm$ and $f$. 

Let $v_1\in W^{2,p}_{\delta}$ be the unique solution to
\begin{equation}
-a\Delta v_1 + kv_1 = f(x,u_-) + k(u_- -\mru)
\end{equation} and let $u_1 = v_1 + \mru$. The solution $v_1$ exists by Proposition
\ref{prop:Isomorphism}. 

Using the equality and inequality satisfied by $v_1$ and $u_-$ respectively, we find that
\begin{equation}
-a\Delta(u_1- u_-) +k(u_1-u_-) \geq 0.
\end{equation} By the maximum principle \ref{prop:MaxPrinciple}, $u_1 \geq u_-$.
Similarly,
\begin{align}
-a\Delta(u_+ - u_1) + k(u_+-u_1) &\geq f(x,u_+) - f(x,u_-) + k(u_+-u_-)\\
&= (u_+-u_-) \left(k + \int_0^1 f_y(x,u_-+t(u_+-u_-)) dt\right)\\
&\geq 0,
\end{align} where the last line holds by our assumption on $k$, Equation \eqref{eq:kDef}.
Again by the maximum principle \ref{prop:MaxPrinciple}, $u_1\leq u_+$.

We then let $u_i = v_i + \mru$, where $v_i \in W^{2,p}_\delta$ solves
\begin{equation}
-a\Delta v_i + k v_i = f(x,u_{i-1}) - k v_{i-1}.
\end{equation} Again using the maximum principle, we can show that $u_i$ is an
increasing sequence; i.e.,
\begin{equation}
u_- \leq u_1 \leq u_2 \leq \cdots \leq u_{i-1} \leq u_i \leq \cdots \leq u_+.
\end{equation} Since the $u_i$ constitute a bounded increasing sequence, the $u_i$ converge
to some function $u$ with $u_- \leq u \leq u_+$. We claim that $u$ is a 
solution of Equation \eqref{eq:Semilinear}.

From Proposition \ref{prop:AELinearExistence}, we have
\begin{equation}
\|v_{i+1}\|_{W^{2,p}_\delta} \leq C \|f(x,u_i) - k v_i\|_{L^p_\delta}.
\end{equation} The right hand side is uniformly bounded by our assumptions on
$k$ and $f$, and since $v_i$ and $u_i$ are bounded. Thus $v_i$ is uniformly
bounded in $W^{2,p}_\delta$. 

The compact embedding of $W^{2,p}_\delta$ into 
$C^{0,\alpha}_{\delta'}$ for any $\delta'> \delta$ and some $\alpha>0$
from Proposition \ref{prop:SobolevEmbeddings}
implies that $u_i \to u$ in $C^{0,\alpha}_{\delta'}$, and that 
$u-\mru \in W^{2,p}_\delta$.
This convergence implies that $f(x,u_{i-1}) - kv_{i-1}$ converges in $L^{p}_\delta$,
and so, since $-a\Delta +k$ is an isomorphism, $u_i$ must converge to $u$ in 
$W^{2,p}_\delta$. Thus $-a\Delta u = f(x,u)$, as 
desired.
\end{proof} 

\chapter{Solvability of the Lichnerowicz Equation} \label{chap:IFF}
%%%%%%%%%%%%%%%%%%%%%%%%%%%%%%%%%%%%%%%%%%%%%%%%%%%%%%%%%%%%%%%%%%%%%%%%%%%%%%%%%%%%%%%%% 

The results in this chapter will appear in \cite{DGI15}.

The Lichnerowicz equation \eqref{eq:OrigLich}
 is a semilinear elliptic equation. Because of the mixed sign of the
exponents, it is of a type not generally studied. However, with appropriate sign
restrictions on the coefficients, we can fully understand this equation. Recall that
in the compact case, the solvability of the Lichnerowicz equation is given by
Table \ref{table:Solvability}, with one caveat. Namely, if $g$ is Yamabe negative,
the Lichnerowicz equation is solvable if and only if $g$ can be
conformally transformed to a metric with scalar curvature $-\kappa \tau^2$. The main
result of this chapter is that, regardless of Yamabe class, the Lichnerowicz equation
on AE manifolds is solvable if and only if $g$ can be conformally transformed to a
metric with scalar curvature $-\kappa \tau^2$. 

First we must discuss what kind of data we are looking for when we discuss asymptotically
Euclidean initial data. Clearly we want an AE manifold, $(M,\gbar)$. The usual regularity
we desire for $\gbar$ is $W^{2,p}_\delta$. However, we also need $\Kbar$ to decay 
at infinity. Heuristically, we want this so that our spacelike slice $(M,\gbar)$ is not
curled up inside the spacetime near infinity. We thus require 
$\Kbar \in W^{1,p}_{\delta-1}$. In order to guarantee this regularity,
we require that our seed data satisfies
\begin{equation} \label{eq:DataRegularity}
(g - g_{Euc}, \tau, N-1, \sigma,r,J) \in W^{2,p}_\delta \times W^{1,p}_{\delta-1} \times
   W^{2,p}_\delta
   \times L^{2p}_{\delta-1} \times L^p_{\delta-2} \times L^p_{\delta-2},
\end{equation} with $p>n/2$ and $\delta \in (2-n,0)$, and similarly for 
$C^{2,\alpha}_\delta$ seed data. We then seek a solution
$(\phi, W)$ with $\phi-\mru \in W^{2,p}_\delta$ and $W \in W^{2,p}_\delta$. The 
reconstructed initial data $(\gbar, \Kbar, T_{nn}, T_{in})$ from 
\eqref{eq:ReconstructedData} then has the desired regularity.

We must make one additional restriction on the sign of the seed data. In particular, we
must require $r\geq 0$. This is known as the weak energy condition, and is simply
requiring that the matter density can never be negative. 
Equivalently, it says that gravity is always an attractive force. This is a physically
reasonable assumption, though not strictly necessary for general relativity. For
instance, solutions of the Einstein equations with stable, traversable wormholes require
matter with negative energy density.

We can now prove the main result of this chapter.

\begin{thm} [Curvature Criterion for AE Solutions to the Lichnerowicz Equation]
\label{thm:LichIff}
Suppose that $(M, g)$ is a $W^{2,p}_\delta$ AE manifold with $p>n/2$ and 
$\delta \in (2-n,0)$. Assume that $r$, $\left|\sigma+\frac{1}{2N} LW\right|^2$
and $\tau^2$ are all contained
in $L^{p}_{\delta-2}$, and that $r\geq 0$. Then the Lichnerowicz equation 
\eqref{eq:OrigLich} has a positive solution $\phi$ with $\phi-\mru \in W^{2,p}_\delta$ 
if and only if there exists a positive conformal factor $\psi$ with
$\psi-\mru'\in W^{2,p}_\delta$ such that $\gbar = \psi^{q-2} g$ has scalar curvature
$-\kappa \tau^2$. The $\mru$ and $\mru'$ are two positive 
asymptotic functions, as defined in
Definition \ref{defn:AsymptoticFunction}. A similar result holds for 
$C^{2,\alpha}_\delta$ regularity.
\end{thm}

\begin{proof}
$(\Rightarrow)\,$ Suppose there is such a solution $\phi$ to the Lichnerowicz equation.
It is well known that the desired $\psi$ is a solution to
\begin{equation}
\label{eq:conformalFactor}
-a\Delta \psi + R \psi + \kappa\tau^2 \psi^{q-1}= 0.
\end{equation}

Equation \eqref{eq:conformalFactor} clearly satisfies the conditions of the sub and
supersolution theorem
\ref{thm:SubSupersolutionTheorem} as a consequence of the regularity we have presumed.
Note that the scalar curvature $R$ must be in $L^{p}_{\delta-2}$. 
For $\beta\geq 1$, $\beta\phi$ is a supersolution for \eqref{eq:conformalFactor}. If
$\beta > \sup \mru'/\mru$, $\beta \phi$ satisfies the conditions of Theorem
\ref{thm:SubSupersolutionTheorem}. For the
subsolution, we take $\psi_- \equiv 0$. This is certainly regular. Also, we note that
since the exponents in \eqref{eq:conformalFactor} are positive, $0$ lies in the interval
of regularity for $f(x,y)$. Together,
these conditions and Theorem \ref{thm:SubSupersolutionTheorem} guarantee the existence of
a solution $\psi\geq 0$ of \eqref{eq:conformalFactor} with the
properties we desire, except that it may be zero somewhere.

However, we can easily argue that $\psi$ cannot be 0 anywhere. Suppose it were zero at
some point. Since $\psi\in W^{2,p}_{loc}$, the strong maximum principle
\ref{prop:StrongMaxPrinciple}
implies that $\psi \equiv 0$. But $\psi \to \mru$ at infinity, a contradiction. Thus 
$\psi>0$, proving the implication.

$(\Leftarrow)\,$ Suppose there is such a conformal factor $\psi$. Note that $\psi$ must
then satisfy Equation \eqref{eq:conformalFactor}. For $\beta\leq 1$, $\beta\psi$ is a
subsolution for the Lichnerowicz equation. If $\beta< \inf \mru/\mru'$, then $\beta\psi$
satisfies the conditions of Theorem \ref{thm:SubSupersolutionTheorem}. 

To help find the supersolution, we use the conformal covariance of the Lichnerowicz 
equation \ref{prop:ConformalCovariance} and the conformal factor $\psi$ to assume
that the scalar curvature is $-\kappa \tau^2$. 

Proposition \eqref{prop:Isomorphism} shows that 
there exist solutions $v_\epsilon$ to the linear problem
\begin{equation}
-a\Delta v_\epsilon + \epsilon\kappa \tau^2 v_\epsilon = 
  \epsilon\left(\left|\sigma+ \frac{1}{2N}LW\right|^2 + r\right)
\end{equation} such that $v_\epsilon-1\in W^{2,p}_\delta$ for each $\epsilon \in [0,1]$.
Note that $v_0 \equiv 1$, and that the solution map is continuous in $\epsilon$. We 
claim that $v_\epsilon >0$ for all $\epsilon \in [0,1]$. By continuity, the set of 
$\epsilon$ for which $v_\epsilon>0$ is open. Suppose some $\epsilon$ were on the boundary
of the set for which $v_\epsilon>0$. By continuity, $v_\epsilon \geq 0$, and 
$v_\epsilon =0$ somewhere. By the strong maximum principle
\ref{prop:StrongMaxPrinciple}, $v_\epsilon \equiv 0$. However, this
contradicts that $v_\epsilon \to 1$ at infinity. Thus the set of $\epsilon$ for which
$v_\epsilon>0$ is open. Since this set is also nonempty, it is all of $[0,1]$.

We claim $\beta v:= \beta v_1$ is a supersolution to the Lichnerowicz equation 
(with $R = -\kappa \tau^2$) for large $\beta$. Indeed, if we plug $\beta v$ into the Lichnerowicz
equation, we get
\begin{multline}
-\kappa \tau^2\beta v-\kappa\tau^2 \beta v +\left|\sigma+ \frac{1}{2N}LW\right|^2 \beta +r\beta \\
  + \kappa \tau^2 (\beta v)^{q-1}-\left|\sigma+ \frac{1}{2N}LW\right|^2(\beta v)^{-q-1}-r (\beta v)^{-q/2} \\
  = \kappa \tau^2\left[(\beta v)^{q-1}-2\beta v\right] 
   + \left|\sigma+ \frac{1}{2N}LW\right|^2\left[\beta-(\beta v)^{-q-1}\right] 
   + r\left[\beta-(\beta v)^{-q/2}\right] \geq 0
\end{multline}
for sufficiently large $\beta$. If $\beta> \mru/\mru'$, the sub and supersolution
theorem
\ref{thm:SubSupersolutionTheorem} provides the desired solution to
the Lichnerowicz equation.
\end{proof}

In light of this theorem, we make the following definition.

\begin{defn}\label{def:Admissible}
The seed data $(g, \tau, N, \sigma,r,J)$ is said to be ``admissible" if there is a 
conformal factor transforming $g$ to a metric with scalar curvature $-\kappa \tau^2$.
\end{defn} 

It is important to understand for which AE metrics we can make
this conformal transformation. This question is resolved in Chapter \ref{chap:Yamabe}.
We note here, however, that it is well known (cf. \cite{Maxwell05b}) that $g$ is Yamabe
positive if and only if $g$ can be conformally transformed to a metric with identically
vanishing scalar curvature. We use this fact later in this chapter.

If the seed data is admissible, the Lichnerowicz equation has a solution asymptotic 
to any desired asymptotic function. This solution is unique.

\begin{thm}[Uniqueness of Solutions to the Lichnerowicz Equation]
\label{thm:LichUniqueness}
Suppose that $(M, g)$ is a $W^{2,p}_\delta$ AE manifold with $p>n/2$ and 
$\delta \in (2-n,0)$. Assume that $r$, $\left|\sigma+\frac{1}{2N} LW\right|^2$
and $\tau^2$ are all contained in $L^{p}_{\delta-2}$, and that $r\geq 0$.
If $\phi_1,\phi_2$ both solve the Lichnerowicz equation, and are such that
$\phi_1-\phi_2 \in W^{2,p}_\delta$, then $\phi_1= \phi_2$.
\end{thm}
\begin{proof}
The following proof is taken from \cite[Thm 8.3]{CBIP06}.

We use $\phi_i$ as a conformal factor and use the conformal covariance of the
Lichnerowicz equation (cf. Proposition \ref{prop:ConformalCovariance}) to get
\begin{equation}
R_{\phi_i^{q-2} g} = - \kappa\tau^2 + |\sigma +LW|^2 \phi_i^{-2q} + r \phi_i^{-\frac32 q +1}.
\end{equation} Therefore, using $u := \phi_2/\phi_1$, we obtain
\begin{multline}
-a\Delta_{\phi^{q-2}_1 g} u + \left(\kappa\tau^2-\left|\sigma+\frac1{2N}LW\right|^2\phi_1^{-2q} + r \phi_1^{-\frac32 q +1}\right)u \\
   = \left(\kappa\tau^2-\left|\sigma+\frac1{2N}LW\right|^2 \phi_2^{-2q} + r \phi_2^{-\frac32 q +1}\right) u^{q-1}.
\end{multline} This equation may be written as
\begin{equation}
-a\Delta_{\phi_1^{q-2} g} (u-1) + \xi(\phi_1,\phi_2) (u-1) = 0
\end{equation} where $\xi(\phi_1, \phi_2)$ is a positive expression
in terms of $\phi_i$ and the seed data. Also, as long as the $\phi_i$ are continuous and
positive, $\xi \in L^p_{\delta-2}$. The operator $-a\Delta+\xi$ is thus an isomorphism 
(cf. Proposition \ref{prop:Isomorphism}). 
Thus $u-1\equiv 0$ and so $\phi_1 \equiv\phi_2$.
\end{proof}

%%%%%%%%%%%%%%%%%%%%%%%%%%%%%%%%%%%%%%%%%%%%%%%%%%%%%%%%%%%%%%%%%
\section{Non-existence of Solutions to the Constraints}
\label{sec:Nonexistence}
%%%%%%%%%%%%%%%%%%%%%%%%%%%%%%%%%%%%%%%%%%%%%%%%%%%%%%%%%%%%%%%%%

Theorem \ref{thm:LichIff} allows us to
find seed data which cannot lead to solutions of
the conformal constraint equations \eqref{eq:ConfConst}.

\begin{thm}
Suppose $g$ is Yamabe non-positive. If $\tau \equiv 0$, then the conformal constraint
equations allow no solutions. The class of such metrics is non-empty.
\end{thm}
\begin{proof}
By Theorem \ref{thm:LichIff}, the Lichnerowicz equation is solvable if and only if
$g$ can be conformally transformed to a metric with scalar curvature 
$-\kappa \tau^2 \equiv 0$.
However, since $g$ is not Yamabe positive,
there is no such conformally related metric (cf. \cite{Maxwell05b}).
Thus there can be no solution to the conformal constraint equations.

Friedrich in \cite{Friedrich11} showed the existence of an AE
Yamabe null manifold. In Proposition \ref{prop:yAE=yCpct}, we show that if a Yamabe
null or negative compact manifold is 
decompactified in a particular way, the related AE metric is also Yamabe null
or negative. 
\end{proof}

In general, if one cannot conformally transform the scalar curvature
to $-\kappa \tau^2$, there cannot be a solution to the conformal constraint equations.
In Chapter \ref{chap:Yamabe}, we give a characterization of when this is possible.

In addition, we can show that if $\tau_i$ approaches some $\tau$ that does not allow 
solutions, then any solutions $\phi_i, W_i$ of the constraint equations corresponding
to seed data with $\tau_i$ must blow up
as $i\to \infty$. We focus on the case $\tau \equiv 0$, but the techniques work
in more general cases. We first prove a lemma.

\begin{lem}\label{lem:BiggerTauSmallerPhi}
Suppose $\tau_1^2\geq \tau_2^2$, with $\tau_i \in W^{1,p}_{\delta-1}$. Suppose
$\phi$ solves the Lichnerowicz equation with $\tau_2$ and any $\sigma$ and $LW$. Suppose
a conformal factor $\psi$ transforming the scalar curvature to $-\kappa \tau_1^2$
exists, and that $\psi-\psi \in W^{2,p}_\delta$. Then $\phi\geq \psi$.
\end{lem}
\begin{proof}
Let $\phi/\psi = \phitil$. By conformal covariance \ref{prop:ConformalCovariance},
$\phitil$ solves
\begin{equation}
-a\Deltatil \phitil -\kappa \tau_1^2 \phitil + \kappa \tau_2^2 \phitil^{q-1}
   - \psi^{-2}\left|\sigma + \frac1{2N} \Ltil W\right|_{\psi^{q-2} g}^2 \phitil^{-q-1} 
    - \psi^{-\frac32 q +1} r \phitil^{-q/2}=0,
\end{equation} where $\Deltatil$ and $\Ltil$ are operators with respect to $\psi^{q-2}g$.

Suppose, by way of contradiction, that $\phi<\psi$ somewhere. Thus $\phitil<1$ somewhere.
Since $\phitil\to 1$ at infinity, it must have a global minimum at some point $p\in M$.
On some small ball $B(p)$ around $p$, $\phitil<1$, and so 
\begin{equation}
-\kappa \tau_1^2 \phitil + \kappa \tau_2^2 \phitil^{q-1} \leq 0
\end{equation} on $B(p)$. Clearly, then, $-\Deltatil\phitil \geq 0$ on $B(p)$.

Let 
\begin{equation}
v= (\phitil - \inf_{\p B(p)} \phitil)^- := \min\{0,\phitil - \inf_{\p B(p)} \phitil\} \leq 0.
\end{equation} Since $\phitil(p)$ is a global minimum, $v=0$ on $\p B(p)$. 
Thus
\begin{equation}
\int_{B(p)} a|\nabla v|^2 = \int_{B(p)} -a \Deltatil \phitil v \leq 0,
\end{equation} and so $v \equiv 0$ on $B(p)$. Thus 
$\inf \phitil \geq \inf_{\p B(p)} \phitil$, and so $\phitil$ is constant on $B(p)$. 
By continuity, we can similarly argue that
 $\phitil<1$ everywhere. This contradicts that $\phitil \to 1$
at infinity. Thus $\phitil\geq 1$, and so $\phi\geq \psi$.
\end{proof}

Using a prescribed scalar curvature result from Chapter \ref{chap:Yamabe},
we can show that sequences of solutions with $\tau_i \to 0$ blow up.

\begin{thm}\label{thm:Blowup}
Suppose we have seed data $(g, \tau_i, N, \sigma,r,J)$ as in \eqref{eq:DataRegularity}.
Suppose $g$ is Yamabe non-positive, and has no conformal Killing fields.
Suppose $(\phi_i, W_i)$ are solutions to the conformal constraint equations 
\eqref{eq:ConfConst} for the data with $\tau_i$. If 
$\tau_i \to 0$ in $C^{0}_{\delta-1}\cap W^{1,p}_{\delta-1}$,
then $\sup \phi_i \to \infty$. 
\end{thm}
\begin{proof}
Since $\tau_i \to 0$ in $C^0_{\delta-1}$, $\kappa \tau_i^2 \leq C \rho^{2\delta-2}$ for
some $C>0$. Since the $\tau_i$ are admissible, Lemma \ref{lem:LowerScalarCurvature}
shows that there exists a conformal factor $\psi$ transforming $g$ to a metric
with scalar curvature $-C\rho^{2\delta-2}$. Then, by Lemma \ref{lem:BiggerTauSmallerPhi},
$\phi_i\geq \psi>0$ for all $i$. Thus the $\phi_i$ are uniformly
bounded below.

Suppose the $\phi_i$ are bounded above. Then 
$\phi_i^q d\tau_i \to 0$ in $L^{p}_{\delta-2}$. By the continuity of the vector equation,
and since there are no conformal Killing fields on $g$ (cf. Proposition
\ref{prop:Isomorphism}), $W_i\to 0$ in $W^{2,p}_\delta$. 

Since the $\phi_i$ are bounded above and below, 
\begin{equation}
F(\phi_i) := R\phi_i + \tau_i^2 \phi_i^{q-1} 
  - \left|\sigma +\frac1{2N} LW_i\right|^2 \phi_i^{-q-1} - r \phi_i^{-q/2}
\end{equation} is bounded in
$L^p_{\delta-2}$. Since $\phi_i$ solves the Lichnerowicz equation \eqref{eq:OrigLich},
estimate \eqref{eq:PEstimate1}
shows that the $\phi_i$ are uniformly bounded in $W^{2,p}_\delta$. By compact embedding,
$\phi_i$ converge strongly in $L^\infty$ to some $\phi_\infty$. Thus $F(\phi_i)$ 
converges strongly in $L^p_{\delta-2}$ to $F(\phi_\infty)$. This in turn shows that 
$\phi_i$ converges in $W^{2,p}_\delta$ to $\phi_\infty$.

However, since $\tau_i\to0$, $\tau_i^2 \phi_i^{q-1} \to 0$. Also, $LW_i \to LW_\infty$,
where $LW_\infty$ is the solution of 
\begin{equation}
\di \frac{1}{2N} LW_\infty = J.
\end{equation} Thus
\begin{equation}
\left|\sigma+\frac1{2N}LW_i\right|^2\phi_i^{-q-1} \to 
   \left|\sigma+\frac1{2N}LW_\infty\right|^2\phi_\infty^{-q-1}.
\end{equation} Since $\phi_i$ converge in $W^{2,p}_\delta$,
$\phi_\infty$ solves 
\begin{equation}
-a\Delta \phi_\infty + R\phi_\infty 
  -\left|\sigma+\frac1{2N} LW_\infty\right|^2 \phi_\infty^{-q-1}-r \phi_\infty^{-q/2}=0.
\end{equation}
which is impossible because, in this equation, $\tau\equiv 0$, which is not 
admissible.
Thus $\phi_i$ cannot be bounded above.
\end{proof}

\begin{thm}
Suppose we have seed data $(g, \tau_i, N, \sigma,r,J)$ as in \eqref{eq:DataRegularity}.
Suppose $g$ is Yamabe non-positive, and has no conformal Killing fields.
Suppose $(\phi_i, W_i)$ are solutions to the conformal constraint equations 
\eqref{eq:ConfConst} for the data with $\tau_i$. If 
$\tau_i \to 0$ in $C^{0}_{\delta-1}\cap W^{1,p}_{\delta-1}$,
and $\tau_i \geq \tau_{i+1}$, then for any
choice of radial function $\rho \geq 1$ and for any $p>n$, one of the following occurs:
\begin{itemize}
\item for all $\eta \in (2-n,0)$, $\|\tau_i^2 \phi_i^{q-1}\|_{L^p_{\eta-2}}$
         is unbounded.

\item for all $\eta\in \R$, $\|\phi_i\|_{L^p_{\eta}}$ is unbounded.
\end{itemize}
\end{thm}
\begin{proof}
Since the $\tau_i$ are admissible, let $\psi_i$ be the conformal factors
transforming $g$ to a metric with scalar curvature
$-\kappa \tau_i^2$. Suppose, by way of
contradiction, that both
$\|\tau_i^2 \psi_i^{q-1}\|_{L^p_{\eta-2}}$ and $\|\psi_i\|_{L^p_{\eta'}}$ are bounded,
for some choices of
$p, \eta, \eta'$ and radial function $\rho$. 
By the estimate \eqref{eq:PEstimate1},
\begin{equation}
\|\psi_i\|_{W^{2,p}_\eta} \leq C\|\tau_i^2 \psi_i^{q-1}\|_{L^{p}_{\eta-2}}
                + C \|\psi_i\|_{L^p_{\eta'}},
\end{equation} which is bounded by assumption.
Since $\psi_i$ is uniformly bounded in $W^{2,p}_\eta$,
by Sobolev embedding \ref{prop:SobolevEmbeddings}
a subsequence, which we also denote by $\psi_i$, converges in $C^0$. Mirroring the proof of 
Theorem \ref{thm:Blowup}, we can find a limit $\psi_\infty$ which
again contradicts that $g$ is Yamabe nonpositive.
Thus either $\|\tau_i^2 \psi_i^{q-1}\|_{L^p_{\eta-2}}$ or $\|\psi_i\|_{L^p_{\eta'}}$
is unbounded.

By Lemma \ref{lem:BiggerTauSmallerPhi}, $\phi_i\geq \psi_i$, and so $\phi_i$ must be
unbounded in the same way as $\psi_i$.
\end{proof} 

%\chapter{Yamabe Classification and Prescribed Scalar Curvature
%in the Asymptotically {E}uclidean Setting} \label{chap:Yamabe}
\chapter{Prescribed Scalar Curvature and Yamabe Classes} \label{chap:Yamabe}
%%%%%%%%%%%%%%%%%%%%%%%%%%%%%%%%%%%%%%%%%%%%%%%%%%%%%%%%%%%%%%%%%%%%%%%%%%%%%%%%%%%%%%%%% 

This work is (lightly) adapted from \cite{DM15}. This paper is quoted with David Maxwell's
permission.

One formulation of the prescribed scalar curvature problem asks: given Riemannian
manifold
$(M^n,g)$ and some function $R'$, is there a conformally related metric $g'$ with scalar
curvature $R'$? If we define $g' = \phi^{N-2} g$ for $N:= \frac{2n}{n-2}$,
\footnote{In this chapter, and this chapter alone, we use the notation $N:= \frac{2n}{n-2}$
instead of $q$. Since the lapse $N$ does not appear in the chapter, there should be no
confusion about the two $N$'s.} this is
equivalent to finding a positive solution of
\begin{equation}\label{eq:prescribed-sc}
  -a\Delta \phi + R \phi = R' \phi^{N-1}.
\end{equation}

On a compact manifold the Yamabe invariant of the conformal
class of $g$ poses an obstacle to the solution of \eqref{eq:prescribed-sc}.
For example, in the case where $M$ is connected and $R'$ is constant, problem \eqref{eq:prescribed-sc}
is known as the Yamabe problem, and it admits a solution if and only
if the sign of the Yamabe invariant agrees with the sign of $R'$
\cite{Yamabe60}\cite{Trudinger68}\cite{Aubin76}\cite{Schoen84}.
More generally,
if $R'$ has constant sign, we can conformally transform to a metric with
scalar curvature $R'$ only if the sign of the Yamabe invariant agrees with the
sign of the scalar curvature.  Hence it is natural to divide conformal classes
into three types, Yamabe positive, negative, and null, depending on the sign
of the Yamabe invariant.

We are interested in
solving equation \eqref{eq:prescribed-sc} on a
class of complete Riemannian manifolds that, loosely speaking, have a geometry
approximating Euclidean space at infinity.  These asymptotically Euclidean
(AE) manifolds also possess a Yamabe invariant, but the relationship between
the Yamabe invariant and problem \eqref{eq:prescribed-sc} was not, up until this work,
well understood in the AE setting, except for some results
concerning Yamabe positive metrics.
We have the following consequences of \cite{Maxwell05b} Proposition 3.
\begin{enumerate}
\item An AE metric can be conformally transformed to an AE
metric with zero scalar curvature if and only if it is Yamabe positive.
As a consequence, since the scalar curvature of an AE metric decays to zero at infinity,
only Yamabe positive AE metrics can be conformally transformed to
have constant scalar curvature.
\item Yamabe positive AE metrics have conformally related AE metrics with
everywhere positive scalar curvature, and
conformally related AE metrics with everywhere negative scalar curvature.
\item If an AE metric admits a conformally
related metric with non-negative scalar curvature, then it is Yamabe positive.
\end{enumerate}
Note that it was at one time believed that transformation to zero scalar
curvature is possible if and only if the manifold is Yamabe non-negative
\cite{BC81}.
The proof of this contention in \cite{BC81} contains an error, and the statement and proof
were corrected in \cite{Maxwell05b}. See also \cite{Friedrich11}, which
shows that there exist Yamabe-null AE manifolds and hence the
hypotheses of \cite{BC81} and \cite{Maxwell05b} are genuinely different.

As a consequence of the above three facts, the situation on an AE manifold is somewhat
different from the compact setting.  In particular,
although positive scalar curvature is a hallmark
of Yamabe positive metrics, negative scalar curvature does
not characterize Yamabe-negative metrics.  Indeed, reporting joint work with 
David Maxwell, we show in this chapter
that given an AE metric $g$, and a strictly negative
function $R'$ that decays to zero suitably at infinity, the conformal
class of $g$ includes a metric with scalar curvature equal to $R'$ regardless
of the sign of the Yamabe invariant.  So every strictly negative scalar curvature
is attainable for every conformal class, but zero
scalar curvature is attainable only for Yamabe positive metrics.  Thus we
are lead to investigate the role of the Yamabe class
in the boundary case of prescribed non-positive scalar curvature.

Rauzy treated the analogous problem on
smooth compact Riemannian manifolds in \cite{Rauzy95},
which contains the following statement.
Suppose $R'\le 0$ and $R'\not\equiv 0$.  Observe that if
$R'$ is the scalar curvature of a metric conformally related to some $g$,
then $g$ must be Yamabe-negative, and without
loss of generality we assume that $g$ has constant negative scalar curvature $R$.
Then there is a metric in the conformal
class of $g$ with scalar curvature $R'$ if and only if
\begin{equation}\label{eq:rauzycond}
a \lambda_{R'} > -R
\end{equation}
where $a$ is the constant from equation \eqref{eq:prescribed-sc} and where
\begin{equation}
\lambda_{R'} = \inf \left\{ \frac{\int |\nabla u|^2 }{\int u^2}: u\in W^{1,2}, u\ge 0, u\not \equiv 0, \int R' u= 0  \right\}.
\end{equation}
Rauzy's condition \eqref{eq:rauzycond}
is not immediately applicable on asymptotically Euclidean manifolds, in part
because of the initial transformation to constant negative scalar curvature.
However, recalling that $R$ is constant
we can write $a\lambda_{R'}+R$ as the infimum of
\begin{equation}
\frac{\int a|\nabla u|^2 + R u^2 }{\int u^2}
\end{equation}
over functions $u$ supported in the region where $R'=0$.  So, in effect, inequality
\eqref{eq:rauzycond} expresses the positivity of the first eigenvalue of
the conformal Laplacian of the constant scalar curvature metric $g$
on the region $\{R'=0\}$.  The connection between the first eigenvalue
of the conformal Laplacian and prescribed scalar curvature problems
is well known, but its use is more technical on
non-compact manifolds where true eigenfunctions need not exist. For example,
\cite{FCS80} shows that a metric on a noncompact manifold
can be conformally transformed to a scalar flat one if and only if the first eigenvalue
of the conformal Laplacian is positive on every bounded domain.

In this chapter, following \cite{DM15}, we extend these ideas in a number of ways
to solve the prescribed non-positive scalar curvature problem on
asymptotically Euclidean manifolds, and we obtain
a related characterization of the Yamabe class of
an AE metric. In particular, we show the following.
\begin{itemize}
\item Every measurable subset $V\subseteq M$
can be assigned a number $y(V)$ that generalizes
the Yamabe invariant of a manifold.  The invariant
depends on the conformal class of the AE metric,
but is independent of the conformal representative.
\item We can assign every measurable subset $V\subseteq M$
a number $\lambda_{\delta}(V)$ that generalizes the
first eigenvalue of the conformal Laplacian.  These
numbers are not conformal invariants, and are not
even canonically defined as they depend on a choice
of parameters (a number $\delta$ and a choice
of weight function at infinity).  Nevertheless
the sign of $\lambda_\delta(V)$ agrees with the sign
of $y(V)$, regardless of the choice of these parameters.
\item Given an AE metric $g$ and a candidate scalar curvature $R'\le 0$,
there is a metric in the conformal class of $g$
with scalar curvature equal to $R'$ if and only if
$\{R'=0\}$ is Yamabe positive; i.e., $y(\{R'=0\})>0$.
\item A metric is Yamabe positive if and only if
for every function $R'\le 0$ there is a conformally related metric with
scalar curvature equal to $R'$.
\item A metric is Yamabe null if and only if
for every function $R'\le 0$, except for $R'\equiv 0$,
there is a conformally related metric with scalar curvature equal to $R'$.
\item A metric is Yamabe negative if and only if
there is a function $R'\le 0$, $R'\not\equiv 0$,
such that no conformally related metric has scalar curvature equal
to $R'$.  We also
present some results concerning which scalar curvatures
have Yamabe positive zero sets.
\item Additionally, a metric is Yamabe positive/negative/null
if and only if it admits a conformal compactification
to a metric with the same Yamabe type.
\end{itemize}

These results carry over to compact manifolds, where
we obtain some technical improvements.  First,
Rauzy's condition \eqref{eq:rauzycond} is
equivalent to our condition $y(\{R'=0\})>0$
(or equivalently $\lambda_\delta(\{R'=0\})>0$).  But the condition
$y(\{R'<0\})>0$
can be checked without reference to a particular background metric.
Moreover, we work with fairly general metrics ($W^{2,p}_\loc$ with $p>n/2$),
and candidate scalar curvatures in $L^p(M)$.
Finally, there is an error in Rauzy's proof, closely
related to the gap in Yamabe's original attempt at
the Yamabe problem, that we correct in our presentation.
\footnote{We thank Rafe Mazzeo for useful conversations concerning this 
correction.}

The prescribed scalar curvature problem on AE manifolds for $R'\ge0$, or 
for functions $R'$ which change sign, remains open.
Of course if $R'\ge 0$ the problem
can only be solved if the manifold is Yamabe positive, but it is not
known the extent to which this condition is sufficient.  For prescribed scalar curvatures
that change sign, little is known for any Yamabe class.  Nevertheless,
the case $R'\le 0$ that we treat here has an interesting application to general relativity; see below.
For comparison, we note that the prescribed scalar curvature problem on a compact manifold
is also not yet fully solved.  On a Yamabe-positive manifold it is necessary
that $R'>0$ somewhere, and on a Yamabe-null manifold it is necessary
that either $R'\equiv 0$, or $R'>0$ somewhere and $\int R'<0$ when
computed with respect to the scalar flat conformal representative.  See
\cite{ES86}
which shows that these conditions
are sufficient in some cases.  See also \cite{Bourguignon:1987ge}
for obstructions posed by conformal Killing fields.

\section{Asymptotically Euclidean Manifolds}\label{sec:AEManifolds}

We mention a few extensions of what has been discussed above in Chapter \ref{chap:AEIntro}
which will be applicable in this chapter.
We will work exclusively with $W^{2,p}_{\alpha}$ AE metrics with $p>n/2$,
and we henceforth assume
\begin{equation}
p>n/2 \quad\text{and}\quad \alpha<0.
\end{equation}
A $W^{2,p}_\alpha$
metric is H\"older continuous and has curvatures in $L^p_{\alpha-2}$.

The Laplacian and conformal Laplacian of a $W^{2,p}_\alpha$ metric are well-defined
as maps from $W^{2,q}_\delta$ to $L^q_{\delta-2}$ for $q\in (1,p]$, they are Fredholm
with index 0 if $\delta\in (2-n,0)$, and indeed the Laplacian is an isomorphism
in this range; see, e.g., \cite{Bartnik86} Proposition 2.2.  Note that
\cite{Bartnik86} works on a manifold diffeomorphic to $R^n$, but the results
we cite from \cite{Bartnik86} extend
to manifolds with general topology and any finite number of ends.

Many of the results in this chapter hold for both asymptotically Euclidean
and compact manifolds, and indeed we can often
treat a $W^{2,p}$ metric on a compact manifold
as a $W^{2,p}_\alpha$ metric on an asymptotically Euclidean manifold with zero ends,
in which case the weight function $\rho$ is irrelevant and could be set to 1
if desired.  For the sake of brevity, throughout Section \ref{sec:FirstEigenvalue}
we interpret a compact manifold as an AE manifold with zero ends.
In the remaining sections there are differences between the two cases and
we assume that AE manifolds have at least one end.

The weight parameter
\begin{equation}\label{eq:deltastar}
\delta^* = \frac{2-n}{2}
\end{equation}
plays a prominent role in this chapter, and it reflects the minimum decay
needed to ensure $\int|\nabla u|^2$ is finite.  At this decay
rate, $L^N_{\delta^*}=L^N$ and we have the inequalities
that generalize the Poincar\'e and Sobolev inequalities,
Lemma \ref{lem:poincare}.

Lemma \ref{lem:poincare} evidently fails on compact manifolds, as can be
seen by taking $u$ to be a constant.  For our proofs that treat the compact and
non-compact case simultaneously it will be helpful to have
a suitable inequality that works in both settings.
Observe that for any $\delta>0$ there exists $c_2$ such that
\begin{equation}
\label{eq:sobolev-cpct}
\|u\|_{2,\delta}+\|\nabla u\|_2 \ge c_2 \|u\|_{N}.
\end{equation}
This follows from the standard Sobolev inequality on compact manifolds and
follows trivially from inequality \eqref{eq:sobolev} on non-compact manifolds.
Recall, again, that in this chapter, $N:=\frac{2n}{n-2}$.

\section{The Yamabe Invariant of a Measurable Set} \label{sec:FirstEigenvalue}

Throughout this section, let $(M,g)$ be a $W^{2,p}_\alpha$ AE manifold
with $p>n/2$ and $\alpha<0$, with the convention that a compact manifold is
an AE manifold with zero ends.
For $u\in C_c^\infty(M)$ (i.e., smooth functions of compact support),
 $u\not\equiv 0$, the Yamabe quotient of $u$ is
\begin{equation}
Q^y_g(u) = \frac{\int a|\nabla u|^2 + R u^2}{\|u\|_{N}^2}
\end{equation}
and the Yamabe invariant of $g$ is the infimum of $Q^y_g$ taken over
$C_c^\infty(M)$. Here and in other notations we drop the decoration
$g$ when the metric is understood. Our principal goal in this section is to
define a similar conformal invariant for arbitrary
measurable subsets of $M$ and to analyze its properties.

It will be convenient to work with a complete function space, and we claim
that the domain of $Q^y$ can be extended to
$W^{1,2}_{\delta^*}\setminus\{0\}$ where $\delta^*$ is defined in equation
\eqref{eq:deltastar}.
To see this, first note from the embedding properties of
weighted Sobolev spaces that $W^{1,2}_{\delta^*}$ embeds continuously in $L^N=L^N_{\delta^*}$
and that $u\mapsto \nabla u$ is continuous from $W^{1,2}_{\delta^*}$ to $L^2$;
indeed $\delta^*$ is the minimum decay needed to ensure these conditions.
To treat the scalar curvature term in $Q^y$, we have the following.
\begin{lem}\label{lem:RIntegralBound}
The map
\begin{equation}\label{eq:loworder}
u\mapsto \int R u^2
\end{equation}
is weakly continuous on $W^{1,2}_{\delta^*}$.  Moreover, for any $\delta>\delta^*$
and $\epsilon>0$, there is constant $C>0$ such that
\begin{equation}\label{eq:RuSquaredEstimate}
\left| \int R u^2 \right| \le \epsilon \|\nabla u\|_2^2 + C \|u\|_{2,\delta}^2.
\end{equation}
\end{lem}
\begin{proof}
Recall that $R\in L^p_{\alpha-2}$ where $p>n/2$ and $\alpha<0$.
So there is an $s\in (0,1)$ such that
\begin{equation}
\frac 1p= s\frac 2n.
\end{equation}
Set $\sigma=\delta^*-\alpha/2$.  Since $s<1$ and $\sigma>\delta^*$,
$W^{1,2}_{\delta^*}$ embeds compactly in $W^{s,2}_{\sigma}$,
where the interpolation (Sobolev) space $W^{s,2}_{\sigma}$
is described in
\cite{triebel-weighted-one}\cite{triebel-weighted-two}.
Moreover, $W^{s,2}_{\sigma}$ embeds continuously in $L^q_\sigma$
where
\begin{equation}
\frac{1}{q} = \frac{1}{2} - \frac{s}{n} =  \frac{1}{2}\left( 1-\frac{1}{p}\right).
\end{equation}
Since
\begin{equation}
\frac{1}{p} + \frac{2}{q} = 1
\end{equation}
and since
\begin{equation}
\alpha-2 + 2\sigma = 2\delta^* -2 = -n,
\end{equation}
H\"older's inequality implies the map \eqref{eq:loworder}
is continuous on $L^q_\sigma$, and from the previously mentioned compact
embedding the map \eqref{eq:loworder} is therefore weakly continuous on $W^{1,2}_{\delta^*}$.
Moreover, H\"older's inequality implies there is a constant $C$ such that
\begin{equation}\label{eq:loworder1}
\left| \int Ru^2\right| \le C \|u\|_{W^{s,2}_\sigma}^2.
\end{equation}
From interpolation \cite{triebel-weighted-two} we have
\begin{equation}\label{eq:loworder2}
\|u\|_{W^{s,2}_\sigma} \le C \|u\|_{W^{1,2}_{\delta^*}}^s \|u\|_{2,\delta}^{1-s}
\end{equation}
where $\delta$ satisfies
\begin{equation}
s\delta^* + (1-s)\delta = \sigma.
\end{equation}
Since $\sigma  = \delta^* - \alpha/2$, we find
\begin{equation}
\delta = \delta^* - \frac{\alpha/2}{1-s},
\end{equation}
and since $\alpha<0$ and $s\in(0,1)$, $\delta>\delta^*$.  Indeed, by raising $\alpha$ close to zero,
or lowering $p$ close to $n/2$ (which raises $s$ up to 1), we can obtain
any particular $\delta>\delta^*$. We conclude from
inequalities \eqref{eq:loworder1}, \eqref{eq:loworder2}
and the arithmetic-geometric mean inequality that
\begin{equation}
\left| \int R u^2 \right| \le \epsilon \|\nabla u\|_{W^{1,2}_{\delta^*}}^2 + C \|u\|_{2,\delta}^2.
\end{equation}
This establishes inequality \eqref{eq:RuSquaredEstimate} on a compact manifold, and
we obtain \eqref{eq:RuSquaredEstimate} in the non-compact case by applying
the Poincar\'e inequality \eqref{eq:poincare}.
\end{proof}
\begin{cor}\label{cor:uppersc}
The map
\begin{equation}
u\mapsto \int a|\nabla u|^2 + Ru^2
\end{equation}
is weakly upper semicontinuous on $W^{1,2}_{\delta^*}$.
\end{cor}
\begin{proof}
This follows from the weak upper semicontinuity of $u\mapsto \int|\nabla u|^2$
along with Lemma \ref{lem:RIntegralBound}.
\end{proof}

\begin{defn}
Let $V\subseteq M$ be a measurable set.  The \textit{test functions supported in $V$}
are
\begin{equation}
A(V) := \left\{ u\in W^{1,2}_{\delta^*}(M): u\not\equiv 0, u|_{V^c}=0\right\},
\end{equation}
where $V^c$ is the complement of $V$.
\end{defn}

\begin{defn}
Let $V\subseteq M$ be measurable. The \textit{Yamabe invariant} of $V$ is
\begin{equation}\label{eq:yamabe-inv}
y_g(V) = \inf_{u\in A(V)} Q^y(u).
\end{equation}
If $V$ has measure zero, and hence $A(V)$ is empty, we use the convention
$y_g(V)=\infty$.
\end{defn}

Since $C^\infty_c(M)$ is dense in $W^{1,2}_{\delta^*}(M)$, for $V=M$, this
agrees with the usual definition of the Yamabe invariant.

In principle, the infimum in the definition of the Yamabe invariant
could be $-\infty$.  The following estimate, which will be useful
later in the paper as well, shows that this is not possible.

\begin{lem}\label{lem:Qplusdelta}
Let $\delta\in\Reals$.
There exist positive constants $C_1$ and $C_2$ such that for all $u\in W^{1,2}_{\delta^*}$,
\begin{equation}\label{eq:basicbound}
\|u\|_{W^{1,2}_{\delta^*}} \le C_1\left[\int a|\nabla u|^2 + R u^2\right]
+ C_2\|u\|_{2,\delta}^2.
\end{equation}
\end{lem}
\begin{proof}
It is enough to establish inequality \eqref{eq:basicbound} assuming $\delta>\delta^*$.
From Lemma \ref{lem:RIntegralBound}, there is a constant $C$ such that
\begin{equation}
\left|\int R u^2 \right| \le \frac{a}{2} \int |\nabla u|^2 + C \|u\|_{2,\delta}^2
\end{equation}
and hence
\begin{equation}
\int a|\nabla u|^2 + R u^2 \ge \frac{a}{2}\int |\nabla u|^2 - C\|u\|_{2,\delta}^2.
\end{equation}
Consequently
\begin{equation}
\int |\nabla u|^2 \le \frac{2}{a}\left[\int a|\nabla u|^2 + R u^2\right] + \frac{2C}{a} \|u\|_{2,\delta}^2.
\end{equation}
Inequality \eqref{eq:basicbound} now follows trivially in the compact case,
and follows from the Poincar\'e inequality \eqref{eq:poincare}
in the non-compact case.
\end{proof}

\begin{lem}\label{lem:YamabeBoundedBelow}
For every measurable set $V$, $y(V)>-\infty$.
\end{lem}
\begin{proof}
Let $u_k$ be some minimizing sequence for $Q^y$ normalized
so that $\|u_k\|_{N} = 1$. Lemma \ref{lem:Qplusdelta} and the continuous embedding
$L^N\hookrightarrow L^2_{\delta}$ implies that $u_k$ is uniformly
bounded in $W^{1,2}_{\delta^*}$. Estimate \eqref{eq:RuSquaredEstimate} then
implies that $Q(u_k)$ is uniformly bounded below.
\end{proof}

As one might expect, $y(V)$ is a conformal invariant.
\begin{lem}\label{ConformalInvariance}
Suppose $g' = \phi^{N-2}g$ is a conformally related metric with
$\phi-1\in W^{2,p}_{\alpha}$. Then
\begin{equation}
y_{g'}(V) = y_g(V).
\end{equation}
\end{lem}
\begin{proof}
The conformal transformation laws
\begin{equation}
\begin{aligned}
dV_{g'} &= \phi^{N} dV_{g}\\
R_{g'} &= \phi^{1-N}( -a\Delta_{g} \phi + R_g \phi)
\end{aligned}
\end{equation}
together with an integration by parts imply
\begin{equation}
\int_M |\nabla u|_{g'}^2 + R_{g'}u^2\; dV_{g'} =
\int_M |\nabla (\phi u)|_{g}^2 + R_{g}(\phi u)^2\; dV_{g}
\end{equation}
for all $u\in W^{1,2}_{\delta^*}(M)$.
Since $\| \cdot \|_{g',N} = \|\phi \cdot \|_{g,N}$, it follows that
\begin{equation}\label{eq:ConformalEquality}
Q^y_{g'}(u) = Q^y_{g}(\phi u)
\end{equation}
for all $u\in W^{1,2}_{\delta^*}(M)$ as well.  Since
$A(V)$ is invariant under multiplication by $\phi$,
$y_{g'}(V) = y_g(V)$.
\end{proof}

We will primarily be interested in the sign of the Yamabe invariant.
\begin{defn}
A measurable set $V\subseteq M$ is called \textit{Yamabe positive},
\textit{negative}, or \textit{null} depending on the
sign of $y_g(V)$.
\end{defn}

The Yamabe invariant involves the critical Sobolev exponent $N$
and hence can be technically difficult to work with.
On a compact manifold, however,
the sign of the Yamabe invariant can be determined from the sign
of the first eigenvalue of the conformal Laplacian. These eigenvalues
enjoy superior analytical properties (for instance, it is simpler to
show that the related eigenfunctions exist), and we now describe how
to extend this approach to measurable subsets of compact or asymptotically
Euclidean manifolds.

For $\delta>\delta^*$ we define the Rayleigh quotients
\begin{equation}
Q_{g,\delta}(u) = \frac{\int a|\nabla u|^2 + R u^2}{\|u\|_{2,\delta}^2}.
\end{equation}
Our previous arguments for the Yamabe quotient
imply that $Q_{g,\delta}$ is well-defined for any $u\in W^{1,2}_{\delta^*}\setminus\{0\}$,
and indeed $Q_{g,\delta}$ is continuous on this set.

\begin{defn} The first \textit{$\delta$-weighted eigenvalue of the conformal Laplacian} is
\begin{equation}\label{eq:EigenvalueConfLaplacian}
\lambda_{g,\delta}(V) = \inf_{u\in A(V)} Q_{g,\delta}(u).
\end{equation}
By convention, if $V$ has measure zero then $\lambda_{g,\delta}(V)=\infty$.
We will write $Q_\delta$ and $\lambda_\delta$ when the metric is understood.
\end{defn}

The value of $\lambda_\delta(V)$ is not particularly meaningful; it depends
on the choice of weight function $\rho$ and it is not a conformal invariant.
Nevertheless, its sign is a conformal invariant independent of the choice of $\rho$.

\begin{prop}\label{cor:MeasTFAE}
For any measurable set $V\subseteq M$,
the following are equivalent:
\begin{enumerate}
\item $y(V)>0$.
\item $\lambda_\delta(V)>0$ for all $\delta>\delta^*$.
\item $\lambda_\delta(V)>0$ for some $\delta>\delta^*$.
\end{enumerate}
\end{prop}
\begin{proof}  We assume that $V$ has positive measure since
the equivalence is trivial otherwise.
The implication $1 \Rightarrow 2$ follows from the inequality
$\|u\|_{2,\delta} \leq C \|u\|_{N}$ applied to $Q^y$. The implication
$2\Rightarrow 3$ is trivial. So it remains to show that $3 \Rightarrow 1$.

Let $V$ be a measurable set with $\lambda_\delta(V)>0$ for some
 $\delta>\delta^*$. Suppose to produce a contradiction
that $y(V)\leq 0$. Then there is a sequence $u_k \in A(V)$, normalized so that
$\int a|\nabla u_k|^2 + \|u_k\|_{2,\delta}^2 = 1$, such that $Q^y(u_k) \leq 1/k$. Then
\begin{align}\label{eq:SignEquivalence1}
\lambda_\delta(V) \|u_k\|_{2,\delta}^2\leq \int a|\nabla u_k|^2 + R u_k^2
      \leq \frac1k \|u_k\|_N^2 \leq \frac{c}k \left[\int a |\nabla u_k|^2
      +\|u_k\|_{2,\delta}^2\right] \leq \frac{c}{k}
\end{align} by the Sobolev inequality \eqref{eq:sobolev-cpct}.
In particular, $\|u_k\|_{2,\delta}^2 \to 0$. Using inequality
\eqref{eq:SignEquivalence1}, we also find that
\begin{equation}
\int R u_k^2 \leq \frac{c}k - \int a |\nabla u|^2 \to -1.
\end{equation} However, by Lemma \ref{lem:RIntegralBound}, there exists $C>0$ such that
\begin{equation}
\left|\int Ru_k^2\right| \leq \frac a2 \|\nabla u_k\|^2_2 + C \|u_k\|_{2,\delta}^2
      \to \frac12,
\end{equation} which is a contradiction.
\end{proof}

\begin{cor}\label{cor:SignEquivalence}
For a measurable set $V\subseteq M$,
the signs of $y(V)$ and $\lambda_\delta(V)$ are the same
for any $\delta>\delta^*$.
\end{cor}
\begin{proof}
Proposition \ref{cor:MeasTFAE} shows that $y(V)$ is positive if and only if
$\lambda_\delta(V)$ is also. Choosing an appropriate test function shows that $y(V)$ is
negative if and only if $\lambda_\delta(V)$ is also. Together, these imply that $y(V)$ is
zero if and only if $\lambda_\delta(V)$ is.
\end{proof}
The decay rate $\delta^*$ is critical for Corollary \ref{cor:SignEquivalence}.
For $\delta<\delta^*$, $W^{1,2}_{\delta^*}$ is not contained in
$L^2_\delta$ and hence our definition of $\lambda_\delta$
does not extend to this range.  One could minimize $Q_\delta$ over smooth functions
instead to define $\lambda_\delta$, but using rescaled bump
functions on large balls as test functions, it can be shown that
$\lambda_\delta(\R^n) =0$ for $\delta<\delta^*$, despite the fact that Lemma \ref{lem:poincare}
implies $y(\R^n)>0$.   Note that we have not addressed equality in the threshold
case $\delta=\delta^*$.

We now turn to continuity properties of $\lambda_\delta$.
Monotonicity is obvious from the definition.
\begin{lem}\label{lem:monotone} Let $\delta>\delta^*$. If $V_1$ and $V_2$ are measurable sets with $V_1\subseteq V_2$,
then $\lambda_\delta(V_1)\ge \lambda_\delta(V_2)$.
\end{lem}
Note that Lemma \ref{lem:monotone} holds even for $V_1=\emptyset$, and that
this relies on our definition $\lambda_\delta(\emptyset)=y(\emptyset)=\infty$.
To obtain more refined properties of $\lambda_\delta$, we start by
showing that minimizers of the Rayleigh quotients
exist and are generalized eigenfunctions.
\begin{prop}\label{prop:eigenfunctions}
Let $V$ be a measurable set with positive measure and let $\delta>\delta^*$.
There exists a non-negative $u \in A(V)$ that minimizes
$Q_\delta$ over $A(V)$.
Moreover, on any open set contained in $V$,
\begin{equation}\label{eq:eigenvalue-eq}
-a\Lap u + R u = \lambda_\delta(V) \rho^{2(\delta^*-\delta)} u.
\end{equation}
\end{prop}
\begin{proof}
Let $u_k$ be a minimizing sequence in $A(V)$; this uses the hypothesis
that $V$ has positive measure.  Without loss of generality we may assume that
each $\|u_k\|_{2,\delta}=1$.  Since
\begin{equation}
a\int_M |\nabla u_k|^2 + R u_k^2 = Q_\delta(u_k),
\end{equation}
and since $u_k$ is a minimizing sequence, Lemma \ref{lem:Qplusdelta}
implies $\{u_k\}$ is bounded in
$W^{1,2}_{\delta^*}(M)$ and hence converges weakly in $W^{1,2}_{\delta^*}(M)$ and
strongly in $L^2_{\delta}(M)$ to a limit $u\in W^{1,2}_{\delta^*}(M)$ with
$\|u\|_{2,\delta}=1$.
Since each $u_k=0$ on $V^c$,
from the strong $L^2_{\delta}$ convergence we see $u=0$ on $V^c$,
and since $u\not \equiv 0$ we conclude that $u\in A(V)$.
Weak upper semicontinuity (Corollary \ref{cor:uppersc}) implies that
$u$ minimizes $Q_\delta$ over the test functions $A(V)$.  Noting
that $|u|$ is also a minimizer, we may assume $u\ge 0$.

Suppose $V$ contains an open set $\Omega$.  Then any $\phi\in C^\infty_c(\Omega)$
with $\phi\not\equiv 0$ belongs to $A(V)$, and we can differentiate
$Q_\delta(u+t\phi)$ at $t=0$ to find that $u$ is a weak solution in $\Omega$ of equation
\eqref{eq:eigenvalue-eq}.
\end{proof}

\begin{lem}[Continuity from above]\label{lem:continuityfromabove}
Let $V\subseteq M$ be a measurable set.
If $\{V_k\}$ is a decreasing sequence of measurable sets with $\cap V_k = V$, then
\begin{equation}\label{eq:CtsFromAbove}
\lim_{k\ra\infty} \lambda_\delta(V_k) = \lambda_\delta(V).
\end{equation}
\end{lem}
\begin{proof} From the elementary monotonicity of $\lambda_\delta$,
$\Lambda = \lim_{k\ra\infty} \lambda_\delta(V_k)$ exists and
\begin{equation}
\lambda_\delta(V_k) \le \Lambda \le \lambda_\delta(V)
\end{equation}
for each $k$. So it is enough to show that
\begin{equation}\label{eq:CtsFromAboveUpper}
\Lambda \ge \lambda_\delta(V).
\end{equation}
We may assume that $\Lambda$ is finite, for inequality \eqref{eq:CtsFromAboveUpper}
is trivial otherwise.  As a consequence, each $V_k$ is nonempty and
Proposition \ref{prop:eigenfunctions} provides minimizers
$u_k$ of $Q_\delta$ over $A(V_k)$ satisfying
$\|u_k\|_{2,\delta}=1$.
For each $k$, since $\|u_k\|_{2,\delta}=1$,
\begin{equation}\label{eq:CtsFromAboveUpperBound}
\int a|\nabla u_k|^2 + Ru^2_k \le \Lambda.
\end{equation}
From inequality \eqref{eq:CtsFromAboveUpperBound}
and the boundedness of the sequence in $L^{2}_{\delta}(M)$,
Lemma \ref{lem:Qplusdelta} implies that the sequence is bounded in $W^{1,2}_{\delta^*}(M)$.
A subsequence converges weakly in $W^{1,2}_{\delta^*}(M)$ and strongly in $L^{2}_{\delta}(M)$
to a limit $v$ with $\|v\|_{2,\delta}=1$. From weak upper semicontinuity
(Corollary \ref{cor:uppersc}) we conclude that $Q_\delta(v)\le \Lambda$ as well.
Moreover, $v\in A(V)$ since $v=0$ on $V_k^c$.  So $\lambda_\delta(v) \le \Lambda$.
\end{proof}
Note that Lemma \ref{lem:continuityfromabove} is false for the Yamabe invariant.
For example, one can take a sequence of balls in $\Reals^n$ that shrink down to
the empty set.  It is easy to see that the Yamabe invariant is scale invariant
and hence is a finite constant along the sequence.  Yet the Yamabe invariant of
the empty set is infinite.  In contrast, if $V_n\searrow \emptyset$,
Lemma \ref{lem:continuityfromabove} implies
$\lambda_\delta(V_n)\ra\infty$, and in particular
at some point along the sequence $\lambda_\delta(V_n)>0$.  The following
result, which is an extension of \cite{Rauzy95} Lemma 2 to the AE setting, shows
that in fact $\lambda_\delta(V)$ is positive so long as
a certain weighted volume is sufficiently small.

\begin{lem}[Small sets are Yamabe positive]
\label{lem:UniformOuterApproximation}
For any $\mu > n$, there exists $C>0$ such that if
$\Vol_\mu(V) := \int_{V} \rho^{-\mu} <C$, $V$ is Yamabe positive.
\end{lem}
\begin{proof}
Suppose that $u\in A(V)$. Define $\delta$ by $(-2\delta -n) \frac{n}2 = -\mu$.
Note that $\mu>n$ implies that $\delta >\delta^*$. Then, by H\"older's inequality,
\begin{equation}\label{eq:uniform1}
\|u\|_{2,\delta}^2 = \int u^2 \rho^{-2\delta -n} \leq \left(\int u^N \right)^{2/N}
\left( \int_{V} \rho^{(-2\delta-n)\frac{n}{2}}\right)^{2/n}
= \|u\|_N^2 \Vol_\mu(V)^{2/n}.
\end{equation}
By the Sobolev inequality \eqref{eq:sobolev-cpct}, there exists $C_1$ such that
\begin{equation}\label{eq:uniform2}
\|u\|_N^2 \leq C_1 \left[\int a |\nabla u|^2 + \|u\|_{2,\delta}^2\right].
\end{equation} We also note that Lemma \ref{lem:RIntegralBound}
implies there exists $C_2$ such that
\begin{equation}\label{eq:uniform3}
-C_2 \|u\|_{2,\delta}^2 \leq \frac12 \int a |\nabla u|^2 + \int R u^2.
\end{equation}

Let $\eta$ be defined by $\eta \Vol_\mu(V)^{2/n}C_1 = \frac12$.
Using inequalities \eqref{eq:uniform1}-\eqref{eq:uniform3}, we calculate
\begin{equation}\label{eq:UniformOuter1}
\begin{aligned}
(\eta - C_2) \|u\|_{2,\delta}^2
&\leq \eta\|u\|_N^2 \Vol_\mu(V)^{2/n} + \int R u^2
 + \frac{1}{2} \int a|\nabla u|^2\\
&\leq \eta\Vol_\mu(V)^{2/n}C_1 \left[\int a |\nabla u|^2 +\|u\|_{2,\delta}^2\right]+ \int R u^2
 + \frac12 \int a|\nabla u|^2\\
&= \int \left(a|\nabla u|^2 + Ru^2\right) + \frac{1}{2}\|u\|_{2,\delta}^2.
\end{aligned}
\end{equation} Dividing through by $\|u\|_{2,\delta}^2$, inequality
\eqref{eq:UniformOuter1} reduces to
\begin{equation}
\eta- C_2 -\frac{1}{2}\leq Q_\delta(u).
\end{equation} As $\Vol_\mu(V) \to 0$, $\eta \to \infty$. Thus there is a $C>0$ such
that if $\Vol_\mu(V) <C$, then $Q_\delta(u)$ has a uniform positive lower bound
for all $u \in A(V)$. Thus $\lambda_\delta(V) >0$, and so $V$ is Yamabe positive
by Corollary \ref{cor:SignEquivalence}.
\end{proof}

In Section \ref{sec:YamabeClassification} below we discuss the relationship between
the Yamabe invariant of an AE manifold and its compactification.
After compactification, for $\mu=2n$, the condition
$\Vol_{\mu}(V)<C$ corresponds to the condition that the
usual volume of the compactified set is sufficiently small. This is exactly
Rauzy's condition, and the other choices of $\mu$ provide a mild generalization
of his result.

\begin{lem}[Strict monotonicity at connected, open sets] \label{InnerApproxNull}
Let $\delta>\delta^*$ and let $\Omega$ be a connected open set.
For any measurable set $E$ in $\Omega$  with positive measure,
\begin{equation}\label{eq:strictmonotone}
\lambda_\delta(\Omega\setminus E) > \lambda_\delta(\Omega).
\end{equation}
\end{lem}
\begin{proof}
Let $V=\Omega\setminus E$. We may assume $V$ has positive measure, for
inequality \eqref{eq:strictmonotone} is trivial otherwise.

Suppose to the contrary that $\lambda_\delta(V) = \lambda_\delta(\Omega)$.
Since $V$ has positive measure, Proposition \ref{prop:eigenfunctions} provides
a function $u\in A(V)$ with $Q_\delta(u) = \lambda_\delta(V)$.
Hence $u$ also is a minimizer of $Q_\delta$ over $A(\Omega)$, and
Proposition \ref{prop:eigenfunctions} implies that $u$ weakly solves
\begin{equation}
-a\Delta u + \left[R-\lambda_\delta \rho^{2(\delta^*-\delta)}\right]u =0
\end{equation}
on $\Omega$. Local regularity implies that $u\in W^{2,p}_{\mathrm{loc}}(\Omega)$,
and we may assume after adjusting $u$ on a set of zero measure
that $u$ is continuous. Since $E$ has positive measure, we can still
conclude that $u$ vanishes at some point in $\Omega$.
Following the argument of Lemma 4 from \cite{Maxwell05b},
we may apply the weak Harnack inequality of \cite{trudinger-measurable}
to conclude that $u$ vanishes everywhere on the connected
set $\Omega$, and hence on all of $M$.  Since $u\in A(\Omega)$,
this is a contradiction.
\end{proof}
The connectivity hypothesis in Lemma \ref{InnerApproxNull} is necessary to
obtain strict monotonicity. For example, two disjoint unit balls in $R^n$
have the same first eigenvalue as a single unit ball.  On the other hand,
the assumption that $\Omega$ is open is not optimal, and relaxing this condition
would require a suitable replacement for the weak Harnack inequality.

Although we have not established continuity from below for $\lambda_\delta$,
it holds in certain cases.  The following is a prototypical result
that suffices for our purposes.
\begin{lem}[Continuity from below; prototype] \label{InnerApproxNeg}
Suppose $V$ is measurable.  Let $x_0\in M$ and let $B_r(x_0)$
be the ball of radius $r$ about $x_0$.  Then for any $\delta>\delta^*$
\begin{equation}\label{eq:CtsFromBelow}
\lim_{r\ra 0} \lambda_\delta( V\setminus B_r) = \lambda_\delta( V ).
\end{equation}
\end{lem}
\begin{proof}
Let $u$ be a function in $A(V)$ that minimizes $Q_\delta$.
Let $\chi_r$ be a radial bump function that equals 0 on $B_{r}(x_0)$,
equals 1 outside
$B_{2r}(x_0)$, and has its gradient bounded by $2/r$.  Defining $u_r = \chi_r u$ we
claim that $u_r\ra u$ in $W^{1,2}_{\delta^*}(M)$.  Assuming this for the moment,
we conclude from the continuity of $Q_\delta$ that
\begin{equation}
\lambda_\delta(V) \le \lambda_\delta(V\setminus B_r) \le Q_\delta( u_r )
\ra Q_\delta(u) = \lambda_\delta(V)
\end{equation}
and hence we obtain equality \eqref{eq:CtsFromBelow}.

To show that $u_r\ra u$ in $W^{1,2}_{\delta^*}$, since
$u_r\ra u$ in $L^2_{\delta^*}$, it is enough to show that
$\int |\nabla(u-u_r)|^2\rightarrow 0$.  However,
\begin{equation}\label{eq:ur-to-u}
\int |\nabla(u-u_r)|^2 \le  2 \int (1-\chi_r)^2 |\nabla u|^2 + u^2 | \nabla(1-\chi_r)|^2.
\end{equation}
The first term on the right-hand side of inequality \eqref{eq:ur-to-u} evidently
converges to zero. For the second, we note from H\"older's inequality that
\begin{equation}
\int_{B_{2r}} u^2 \le \left[\int_{B_{2r}} u^N\right]^{\frac{2}{N}}
\left[\int_{B_{2r}} 1 \right]^{\frac{2}{n}} \le Cr^2 \left[\int_{B_{2r}} u^N\right]^{\frac{2}{N}}.
\end{equation}
Since $u\in L^N_{\mathrm{loc}}$, $\int_{B_{2r}} u^N\to 0$ as $r\ra 0$.  Since
$\nabla(1-\chi_r)$ is bounded by $c/r$, we conclude that the second term of
the right-hand side of inequality \eqref{eq:ur-to-u} also converges to zero.

\end{proof}

\section{Prescribed Non-Positive Scalar Curvature} \label{sec:PrescribedProblem}

In this section, we prove the following necessary and sufficient condition for
an AE Riemannian manifold with at least one end to be conformally related to one
which has scalar curvature equal to a specified nonpositive function.

\begin{thm}\label{thm:PrescribedScalarCurvature}
Let $(M^n,g)$ be a $W^{2,p}_\alpha$ AE manifold with $p>n/2$ and $\alpha \in (2-n,0)$.
Suppose $R'\in L^p_{\alpha-2}$ is non-positive.  Then the following are equivalent:
\begin{enumerate}
\item There exists
a positive function $\phi$ with
$\phi-1\in W^{2,p}_\alpha$
 such that the scalar curvature
of $g'=\phi^{N-2}g$ is $R'$.
\item $\{R'=0\}$ is Yamabe positive.
\end{enumerate}
\end{thm}

For compact Yamabe negative manifolds we have the following analogous
result. Since Rauzy's condition
\eqref{eq:rauzycond} is equivalent to the set $\{R'=0\}$ being Yamabe
positive, this theorem is a generalization to lower regularity
and a correction of the proof of part of Theorem 1 in Rauzy's work \cite{Rauzy95}.

\begin{thm}\label{thm:PrescribedCompact}
Let $(M^n,g)$ be a $W^{2,p}$ compact Yamabe negative manifold with $p>n/2$.
Suppose $R'\in L^p$ is non-positive.  Then the following are equivalent:
\begin{enumerate}
\item There exists
a positive function $\phi$ with
$\phi\in W^{2,p}$ such that the scalar curvature
of $g'=\phi^{N-2}g$ is $R'$.
\item $\{R'=0\}$ is Yamabe positive.
\end{enumerate}
\end{thm}

For the most part, the proof of Theorem \ref{thm:PrescribedCompact} can be obtained
from the proof of Theorem \ref{thm:PrescribedScalarCurvature}
by treating a compact manifold as an asymptotically Euclidean manifold
with zero ends.  So we
focus on Theorem \ref{thm:PrescribedScalarCurvature} and then
present the few additional arguments needed to prove Theorem \ref{thm:PrescribedCompact}
at the end of the section.

Turning to Theorem \ref{thm:PrescribedScalarCurvature},
the  proof that 1) implies 2) is short, so we
delay it and concentrate on the direction 2) implies 1).
Suppose that $\{R'=0\}$ is Yamabe positive.
We show that we can make the desired conformal change
using a sequence of results proved over the remainder of this section.
It suffices to work under the following simplifying hypotheses.
\begin{enumerate}
\item We may assume that the prescribed scalar curvature $R'$ is bounded
since Lemma \ref{lem:LowerScalarCurvature}, which we prove next, shows that we can the lower
scalar curvature after first solving the problem for a scalar curvature that is
truncated below.
\item We may assume $\{R'=0\}$ contains a neighborhood of infinity,
since continuity from above (Lemma \ref{lem:continuityfromabove})
shows that we can truncate $R'$ in a ``small'' neighborhood
of infinity such that its zero set remains Yamabe positive, and
we can subsequently lower the scalar curvature after solving the modified problem.
\item We may assume that the initial scalar curvature satisfies $R=0$ in
a neighborhood of infinity, since Lemma \ref{lem:ZeroNearInfinity}, which we prove
below, shows that we can initially conformally transform to such a scalar curvature,
and since the hypotheses of Theorem \ref{thm:PrescribedScalarCurvature} are conformally
invariant.
\end{enumerate}

\begin{lem} \label{lem:LowerScalarCurvature}
Suppose $(M, g)$ is a $W^{2,p}_\alpha$ AE manifold with $p>n/2$ and $\alpha \in (2-n,0)$.
Suppose $R'\in L^{p}_{\alpha-2}$. If $R_g \geq R'$, then there exists a positive $\phi$
with $\phi-1 \in W^{2,p}_\alpha$ such that $g' = \phi^{N-2} g$ has scalar curvature $R'$.
\end{lem}
\begin{proof}
We seek a solution to $-a\Delta \phi + R_g\phi = R'\phi^{q-1}$. Note that
$0$ is a subsolution and, since $R_g\geq R'$, $1$ is a supersolution. By \cite{Maxwell05b}
Proposition 2,
there exists a solution $\phi$ with $0\leq \phi\leq 1$ and $\phi-1\in W^{2,p}_\alpha$.
Since $\phi\ge 0$ solves $-a\Delta\phi+ (R-R'\phi^{q-2})\phi= 0$, and since $\phi\to 1$
at infinity, the weak Harnack inequality \cite{trudinger-measurable} implies that $\phi$
is positive.
\end{proof}

\begin{lem}\label{lem:ZeroNearInfinity}
Suppose $(M, g)$ is a $W^{2,p}_\alpha$ AE manifold with $p>n/2$ and $\alpha \in (2-n,0)$.
There exists $\phi>0$ with $\phi-1 \in W^{2,p}_\alpha$
such that the metric $g' =\phi^{N-2}g$ has zero scalar curvature on some neighborhood
of infinity.
\end{lem}
\begin{proof}
We prove this result for a manifold with one end; the extension to
several ends can be done by repeated application of our argument.
Let $E_r$ be the region outside the coordinate ball of radius $r$ in end
coordinates. By Lemma
\ref{lem:UniformOuterApproximation}, $y(E_r)>0$ for $r$ large enough.
Following Proposition 3 in \cite{Maxwell05b} we claim that
\begin{equation} \label{eq:DirLap}
-a\Delta + \eta R : \{u \in W^{2,p}_\alpha(E_R): u|_{\p E_r} = 0\} \to L^p_{\alpha-2}(E_R)
\end{equation} is an isomorphism for all $\eta \in [0,1]$. Because we assume
homogenous boundary conditions, the argument in Propositions 1.6 through 1.14
in \cite{Bartnik86} showing that $-a\Delta+\eta R$ is Fredholm of index zero
requires no changes except imposing the boundary condition. Suppose, then,
to produce a contradiction, that there exists a nontrivial $u$ in the kernel.
An argument parallel to Lemma 3 in \cite{Maxwell05b} implies that
$u\in W^{2,p}_{\alpha'}$ for any $\alpha' \in (2-n,0)$.  In particular,
the extension of $u$ by zero to $M$ belongs to $W^{1,2}_{\delta^*}(M)$ and hence
also to $A(E_r)$. Integration by parts implies $Q^y(u) = 0$, which contradicts
the fact that $E_r$ is Yamabe positive. Thus $-a\Delta+\eta R$ is an isomorphism.

Let $u_\eta$ be the nontrivial solution in $\{u \in W^{2,p}_\alpha(E_r): u|_{\p E_r} =0\}$ of
\begin{equation}
  -a\Delta u_\eta + \eta R u_\eta = -\eta R.
\end{equation} Then $\phi_\eta:= u_\eta +1$ solves
\begin{equation}
-a\Delta \phi_\eta + \eta R \phi_\eta = 0
\end{equation} on $E_r$. Let $I = \{\eta \in [0,1]: \phi_\eta>0\}$. Since
$\phi_0 \equiv 1$, $I$ is nonempty. The set of solutions $u_\eta$ such that $u_\eta>-1$
is open in $W^{2,p}_\alpha \subset C^0_\alpha$. Thus, by the continuity of the map
$\eta\mapsto u_\eta$, $I$ is open. Suppose $\eta_0 \in \overline{I}$. If
$\phi_{\eta_0} =0$ somewhere, the weak Harnack inequality \cite{trudinger-measurable}
implies that $\phi_{\eta_0} \equiv 0$, which
contradicts the fact that $\phi_{\eta_0} \to 1$ at infinity. Thus $\phi_{\eta_0} >0$
on $E_r$, and so $I$ is closed. Thus $I = [0,1]$, and $\phi_1>0$.
We set $\phi$ to be an arbitrary positive $W^{2,p}_\alpha$ extension of
$\phi_1|_{E_{r}}$; $\phi$ satisfies the properties claimed in this lemma.
%The proof of
%Proposition 3 of \cite{Maxwell05b} then implies the existence of the desired $\phi$
%on $E_r$, with $\phi =1$ on $\p E_r$. Since $\p E_r$ is smooth, $\phi$ can be extended
%to a positive globally $W^{2,p}_\alpha$ function.
\end{proof}

Consider the family of functionals
\begin{equation}\label{eq:DefOfF}
F_q(u) = \int a |\nabla u|^2 + \int R (u+1)^2 - \frac{2}{q} \int R' \left|u+1\right|^q
\end{equation}
for $q\in [2,N)$.

Broadly, the strategy of the proof of Theorem \ref{thm:PrescribedScalarCurvature}
is to construct minimizers $u_q$ of the subcritical functionals,
and then establish sufficient control to show that $(1+u_q)$ converges in
the limit $q\ra N$ to the desired conformal factor.
The following uniform coercivity estimate, which we prove following a variation
of techniques found in \cite{Rauzy95}, is the key step in showing the existence
of subcritical minimizers.

\begin{prop}[Coercivity of $F_q$]\label{prop:Coercivity} Suppose that
 $\{R'=0\}$ is Yamabe positive, that $\delta>\delta^*$, and that $q_0\in(2,N)$.
For every $B\in\R$  there is a $K>0$ such that for all $q\in [q_0,N)$ and all
$u\in W^{1,2}_{\delta^*}$ with $u\ge -1$, if $\|u\|_{2,\delta}>K$ then $F_q(u)>B$.
\end{prop}
\begin{proof}
For $\eta>0$ let
\begin{equation}
A_{\eta} = \left\{ u\in W^{1,2}_{\delta^*}, u\ge -1:
   \int |R'| |u|^2 \le \eta \|u\|_{2,\delta}^2 \int |R'|\right\}.
\end{equation}
Morally, $u\in A_{\eta}$ if it is concentrated on the zero set
\begin{equation}
Z = \{R'=0\},
\end{equation}
with greater concentration as $\eta\ra 0$.

Fix a constant $\calL\in(0,\lambda_\delta(Z))$.
We first claim that there is an $\eta_0<1$
such that if $u\in A_{\eta_0}$, then
\begin{equation}\label{eq:pseudo-yamabe-positive}
\int a |\nabla u|^2 + Ru^2 \ge \calL \|u\|_{2,\delta}^2.
\end{equation}
Suppose to the contrary
that this is false, and let $\eta_k$ be a sequence converging to $0$.
We can then construct a sequence $v_{k}$ with each $v_k\in A_{\eta_k}$ such that
$\|v_{k}\|_{2,\delta}=1$ and
\begin{equation}
\int a |\nabla v_{k}|^2 + Rv_{k}^2< \calL.
\end{equation}
Note that $\calL$ is finite even if $\lambda_\delta(Z)=\infty$.  So
from the boundedness of the sequence $v_k$ in $L^2_\delta$ and
Lemma \ref{lem:Qplusdelta}, the sequence $v_k$ is bounded in $W^{1,2}_{\delta^*}$,
and a subsequence (which we reduce to) converges weakly in $W^{1,2}_{\delta^*}$
and strongly in $L^2_{\delta}$ to a limit $v$ with $\|v\|_{2,\delta}=1$.
Now
\begin{align}
0 \le \int |R'| v_k^2 &\le \eta_k \int |R'|\rightarrow 0.
\end{align}
Since $|R'| v_k^2 \rightarrow |R'| v^2$ in $L^1$
we conclude that $v=0$ outside of $Z$.  From weak upper semicontinuity
(Corollary \ref{cor:uppersc}) we conclude that
\begin{equation}\label{eq:v_was_small}
\int a|\nabla v|^2 + Rv^2 \le \calL
\end{equation}
as well. However, since $v$ is supported in $Z$
\begin{equation}
\int a|\nabla v|^2 + Rv^2 \ge \lambda_\delta(Z)\|v\|_{2,\delta}^2 = \lambda_\delta(Z) > \calL,
\end{equation}
which is a contradiction, and establishes inequality \eqref{eq:pseudo-yamabe-positive}.

Let $B\in\R$
and suppose $q\in(q_0,N)$, $u\in W^{1,2}_{\delta^*}$ and $u\ge -1$. We wish to show that
there is a $K$ independent of $q$ so that if $\|u\|_{2,\delta}>K$ then $F_q(u)>B$.
It is enough to find a choice of $K$ under two cases depending on
whether $u\in A_{\eta_0}$ or not. If $u$ is concentrated on $Z$, the coercivity
follows from the fact that $Z$ is Yamabe positive (as used to obtain
inequality \eqref{eq:pseudo-yamabe-positive}),
and if $u$ is not concentrated on $Z$ then the coercivity follows from the
fact that $R'<0$ away from $Z$.

Suppose that $u\not\in A_{\eta_0}$, so
\begin{equation}\label{eq:u_not_in_A}
\int |R'| |u|^2 > \eta_0 \|u\|_{2,\delta}^2 \int |R'|.
\end{equation}
We calculate
\begin{equation}\label{eq:F_qstep1a}
\begin{aligned}
F_q(u) &= \int a|\nabla u|^2 + \int R(u+1)^2 + \frac{2}{q} \int |R'||u+1|^q\\
&\ge \int a|\nabla u|^2 -2 \int |R|(u^2+1) + \frac{2}{q} \int |R'|(|u|^q-1)\\
&\ge \int \frac{a}{2}|\nabla u|^2 - C\|u\|_{2,\delta}^2 -2 \int |R| + \frac{2}{q} \int |R'| (|u|^q-1)\\
&\ge \int \frac{a}{2}|\nabla u|^2 - C\|u\|_{2,\delta}^2
      -2\int \left(|R| + \frac{1}{q}|R'|\right) + \frac{2}{q} \int |R'| |u|^q.
\end{aligned}
\end{equation}
Here we have applied Lemma \ref{lem:RIntegralBound}, and
have used the fact that $(u+1)^q\ge |u|^q-1$ for $u\ge -1$. Inequality
\eqref{eq:u_not_in_A} and H\"older's inequality imply
\begin{equation}
\begin{aligned}
\eta_0\|u\|_{2,\delta}^2 \int |R'| &<  \int |R'| |u|^2
\le \left(\int |R'| |u|^q\right)^{\frac 2 q }
\left(\int |R'| \right)^{1-\frac 2 q }
\end{aligned}
\end{equation}
and hence
\begin{equation}\label{eq:F_qstep1b}
(\eta_0)^\frac{q}{2}\|u\|_{2,\delta}^q \int |R'| \le \int |R'| |u|^q.
\end{equation}
Using the fact that $\eta_0<1$ and $q<N$, inequalities \eqref{eq:F_qstep1a} and
\eqref{eq:F_qstep1b} imply at last that
\begin{equation}
F_q(u) \geq
\int \frac{a}{2}|\nabla u|^2 - C\|u\|_{2,\delta}^2 -2\int \left(|R| + \frac{1}{q}|R'|\right)
 + \frac{2}{q} (\eta_0)^{\frac{N}{2}} \|u\|_{2,\delta}^q \int |R'|.
\end{equation}
We note that $\int |R'|>0$, for otherwise condition \eqref{eq:u_not_in_A}
is impossible,
and hence the coefficient on $\|u\|_{2,\delta}^q$ is positive. Since $q>2$,
there is a $K$ such that if  $\|u\|_{2,\delta}>K$, then
$F_q(u)\ge B$. Note that since $C$ is independent of $q\geq q_0$, so is the choice of $K$.

Now suppose $u\in A_{\eta_0}$, so inequality \eqref{eq:pseudo-yamabe-positive} holds.
Then for any $\epsilon>0$,
\begin{equation}
\begin{aligned}\label{eq:F_q_near_zero}
F_q(u) &\ge \int a|\nabla u|^2 + \int R (u+1)^2\\
&= \int a|\nabla u|^2 + Ru^2 + \int R\left[ (u+1)^2 - u^2\right]\\
&\ge \int a|\nabla u|^2 + Ru^2 - \int |R|\left[ \epsilon u^2 +1+\frac{1}{\epsilon}\right]\\
&\ge (1-\epsilon) \left[\int a|\nabla u|^2 + Ru^2\right] + \epsilon \int (a|\nabla u|^2-2|R|u^2)
- \left(1+\frac{1}{\epsilon}\right) \int |R|\\
&\ge (1-\epsilon) \calL \|u\|_{2,\delta}^2
+ \epsilon\left( \int \frac{a}{2}|\nabla u|^2-C\|u\|_{2,\delta}^2\right)
-\left(1+\frac{1}{\epsilon}\right) \int |R|\\
&\ge \left[ (1-\epsilon) \calL-\epsilon C\right]
\|u\|_{2,\delta}^2 + \epsilon \int \frac{a}{2}|\nabla u|^2
-\left(1+\frac{1}{\epsilon}\right) \int |R|.
\end{aligned}
\end{equation}
Here we have applied Lemma \ref{lem:RIntegralBound},
inequality \eqref{eq:pseudo-yamabe-positive},
and the fact that $(u+1)^2-u^2 \le \epsilon u^2 +1 + (1/\epsilon)$
for all $u\geq -1$ and all $\epsilon>0$.  We can pick $\epsilon$
sufficiently small so that the coefficient of $\|u\|_{2,\delta}$
in the final expression of inequality \eqref{eq:F_q_near_zero} is at
least $\calL/2$.  Hence there is a $K$ such that
if $\|u\|_{2,\delta}\ge K$, then $F_q(u)\ge B$.
Since $C$ is independent of $q\geq q_0$, so is $\epsilon$ and the choice of $K$.
\end{proof}

\begin{lem}\label{lem:ContinuityOfF}
For $q<N$ the operator $F_q$ is weakly upper semicontinuous on
$W^{1,2}_{\delta^*}$.
\end{lem}
\begin{proof}
Lemma \ref{lem:RIntegralBound} together with the weak continuity of continuous linear
maps implies that
\begin{equation}
u\mapsto \int a|\nabla u|^2 + R(u+1)^2
\end{equation}
is weakly upper semicontinuous on $W^{1,2}_{\delta^*}$.  Hence it suffices
to show that
\begin{equation}\label{eq:weakmap2}
u\mapsto \int R'|u+1|^{q-1}
\end{equation}
is weakly continuous on $W^{1,2}_{\delta^*}$.  But fixing $\delta>\delta^*$
we know that the embedding $W^{1,2}_{\delta^*}\hookrightarrow L^{q}_{\delta}$
is compact and that the map \eqref{eq:weakmap2} is continuous on $L^q_{\delta}$.
\end{proof}

We now obtain existence of subcritical minimizers from the coercivity of $F_q$,
along with uniform estimates in $W^{1,2}_{\delta^*}$ for the minimizers.

\begin{lem}\label{lem:subcriticalExistence}
For any $q_0 \in (2,N)$, for each $q\in [q_0,N)$, there exists $u_q>-1$,
bounded in $W^{1,2}_{\delta^*}$ and independent of $q$, which minimizes $F_q$
and is a weak solution of
\begin{equation}\label{eq:almost-scalar-curvature}
-a\Delta (u_q+1) + R(u_q+1) = R'(u_q+1)^{q-1}.
\end{equation}
Moreover, $u_q \in W^{2,p}_{\sigma}$ for every $\sigma\in(2-n,0)$.
\end{lem}
\begin{proof}
Let $B= \int R + \int |R'|$, let $\delta>\delta^*$, and let $q_0\in (2,N)$. Observe that
\begin{equation}
F_q(0) \le B
\end{equation}
for all $q\in (q_0,N)$. Let $K$ be the constant associated with $B$, $\delta$ and $q_0$
obtained from Proposition \ref{prop:Coercivity}.
Fix $q\in(q_0,N)$ and let $u_k$ be a minimizing sequence in
$W^{1,2}_{\delta^*}$ for $F_q$.
Without loss of generality, we can assume each $u_k\ge -1$ since
$F_q(u_k)=F_q(\max(u_k,-2-u_k))$. We
can assume that each $F_q(u_k)\le F_q(0) \le B$ and hence
Proposition \ref{prop:Coercivity} implies that each
$\|u_k\|_{2,\delta}\le K$. Since
\begin{equation}
\int a|\nabla u_k|^2+R(1+u_k)^2 \le F_q(u_k) < B
\end{equation}
as well, Lemma \ref{lem:Qplusdelta} implies that there is a $C>0$ such that each
$\|u_k\|_{W^{1,2}_{\delta^*}}\le C$.
Note that $C$ depends on $K$ and $B$, which are
independent of $q\geq q_0$.
A subsequence (which we reduce to) converges weakly in $W^{1,2}_{\delta^*}$ and
strongly in $L^q_{\delta}$ to a limit $u_q\ge -1$. Lemma \ref{lem:ContinuityOfF}
shows that $F_q$ is weakly upper semicontinuous, so $u_q$ is a minimizer.
Moreover, $\|u_q\|_{W^{1,2}_{\delta^*}}\le C$ as well.

Since $u_q$ is a minimizer, we find that $(1+u_q)$ is a weak solution of
\begin{equation}\label{eq:minlinop}
\left[-a\Delta  + R -R'(1+u_q)^{q-2}\right] (1+u_q) = 0.
\end{equation}
Since $R'\in L^\infty_{\text{loc}}$ and since $u_q\in L^{N}_{\text{loc}}$,
an easy computation shows that $R'(1+u_q)^{q-2}\in L^{r}_{\text{loc}}$ for some $r>n/2$.
Since $R\in L^{p}_{\text{loc}}$ and $g\in W^{2,p}_{\text{loc}}$ with $p>n/2$,
we find that the coefficients of the differential operator in brackets
in equation \eqref{eq:minlinop} satisfy the hypotheses of the weak Harnack
inequality of \cite{trudinger-measurable}. Hence, since $1+u_q\ge 0$ and
since the manifold is connected,
either $1+u_q>0$ everywhere or $u_q\equiv-1$. But $u_q$ decays at infinity, and so
we conclude that $1+u_q$ is everywhere positive.

We now bootstrap the regularity of $u_q$, which we know initially belongs to $L^N_{\delta^*}$.
Fix $\sigma\in(2-n,0)$. Suppose it is known
that for some $r\ge N$ that $u_q\in L^r_\loc$. From equation \eqref{eq:minlinop},
$u_q$ solves
\begin{equation}\label{eq:bootstrap-grr}
-a\Delta u_q = R'(1+u_q)^{q-1}-R(1+u_q).
\end{equation}
Recall that $R'\in L^\infty_\loc$ and $R\in L^p_\loc$ and both
have compact support.  Then $R'(1+u_q)^{q-1}$ belongs to $L^{t_1}_\sigma$
with
\begin{equation}
\frac{1}{t_1} = \frac{q-1}{r} \le \frac{1}{r}+\frac{q-2}{N} < \frac{N-1}N,
\end{equation}
and $R(1+u_q)$ belongs to  $L^{t_2}_\sigma$ with
\begin{equation}
\frac{1}{t_2} = \frac{1}{r} + \frac{1}{p}.
\end{equation}
Let $t=\min(t_1,t_2)$ and note that $t<p$ since $t_2<p$. From
\cite{Bartnik86} Proposition 1.6 we see that $u_q$ is a strong
solution of \eqref{eq:bootstrap-grr} and from \cite{Bartnik86} Proposition 2.2,
which implies $\Delta:W^{2,t}_{\sigma}\rightarrow L^t_\sigma$
is an isomorphism for $1<t\le p$, we conclude that
$u_q \in W^{2,t}_{\sigma}$. From Sobolev embedding we obtain
$u_q\in L^{r'}_\sigma$ where
\begin{equation}
\frac 1 {r'} = \frac 1 t-\frac{2}{n},
\end{equation}
so long as $1/t> n/2$, at which point the bootstrap changes as discussed below.
Now
\begin{equation}\label{eq:boothalf1}\begin{aligned}
\frac{1}{t_1}-\frac{2}{n} &\le  \frac{1}{r} + \frac{q-2}{N} - \frac{2}{n} \\
&= \frac{1}{r} +\frac{q}{N} - \left[ \frac{2}{N}+\frac{2}{n}\right]\\
&= \frac{1}{r} +\left[\frac{q}{N}-1\right].
\end{aligned}\end{equation}
Also,
\begin{equation}
\frac{1}{t_2} -\frac{2}{n} = \frac{1}{r} + \left[\frac{1}{p} -\frac{2}n\right].\label{eq:boothalf2}
\end{equation}
Let $\epsilon = \min(1-q/N, 2/n-1/p )$ and note that $\epsilon$
is positive and independent of $r$. Inequalities \eqref{eq:boothalf1}
and \eqref{eq:boothalf2} imply
\begin{equation}
\frac{1}{r'} \le \frac{1}{r} - \epsilon
\end{equation}
Hence, after a finite number of iterations (depending on the size of $\epsilon$, and
hence on how close $q$ is to $N$) we can reduce $1/r$ by multiples of $\epsilon$
until $1/r\le \epsilon$.  At this point the bootstrap changes, and in at most
two more iterations we can conclude that $u_q\in L^\infty_\sigma$ and also
$u_q\in W^{2,p}_\sigma$.

%  Thus
% $R'(1+u_q)^{q-1}\in L^{N/(q-1)}$ and $R(1+u_q) \in L^{pN/(N+p)}$. Since
% $\Delta:W^{2,r}_{\sigma}\to L^r_{\sigma}$ is an isomorphism for any
% $\sigma\in(2-n,0)$, elliptic regularity
% gives that $u_q\in W^{2,r}_\sigma$, where $r=\min\{N/(q-1), pN/(N+p),p\}$.
% \dmnote{}
% By Sobolev embedding, $u_q$ belongs to $L^{r'}$ for some $r'>N$.
% Thus we can bootstrap the regularity until $r=p$.
% Thus $u_q \in W^{2,p}_\sigma$ for any $\sigma \in(2-n,0)$.
\end{proof}

The uniform $W^{1,2}_{\delta^*}$ bounds of Lemma \ref{lem:subcriticalExistence}
are enough to obtain the existence of a solution $u$ in $W^{2,N/(N-1)}_\sigma$
of equation \eqref{eq:almost-scalar-curvature}
with $q=N$. At the end of Section IV.6 of \cite{Rauzy95} it is
claimed that on a compact manifold in the smooth setting that
elliptic regularity now implies
$u$ is smooth.  But in fact this is not quite enough regularity to start a bootstrap:
$W^{2,N/(N-1)}_\sigma$ embeds continuously in $L^N_\sigma$, which implies no more
regularity than was known initially.  To start a bootstrap and ensure
the continuity of $u$ we need
the following improved estimate, which follows a modification of the strategy
of \cite{LP87} Proposition 4.4.

\begin{lem}\label{UniformSubcriticalBound}
For each compact set $K$, the minimizers $u_q$ are uniformly bounded in $L^M(K)$
for some $M>N$.
\end{lem}
\begin{proof}
Let $\chi$ be a smooth positive function with compact support that equals 1 in a neighborhood
of $K$.  Let $v = \chi^2 (1+u_q)^{1+2\sigma}$ where $u_q$ is a subcritical minimizer
and where $\sigma$ is a small constant to be chosen later.
Note that since $u_q \in L^{\infty}_{\text{loc}}\cap W^{1,2}_{\text{loc}}$,
$v\in W^{1,2}_{\delta^*}$.  Setting  $w =(1+u_q)^{1+\sigma}$, a short
computation shows that
\begin{equation}\label{eq:subcritbound1}
\int\chi^2 |\nabla w|^2 = -2\frac{1+\sigma}{1+2\sigma} \int \left<\chi\nabla w,w\nabla \chi\right>
+ \frac{(1+\sigma)^2}{1+2\sigma} \int \left<\nabla u_q, \nabla v\right>.
\end{equation}
Applying Young's inequality to the first term on the right-hand side of equation
\eqref{eq:subcritbound1} and merging a resulting piece into the left-hand side we conclude
there is a constant $C_1$ such that
\begin{equation}\label{eq:subcritbound2}
\| \chi \nabla w\|^2_2 \le C_1 \|w\nabla\chi\|_2^2 + 2\frac{(1+\sigma)^2}{1+2\sigma} \int
\left<\nabla u_q, \nabla v\right>.
\end{equation}
Since $u_q$ is a subcritical minimizer,
\begin{equation}\label{eq:subcritbound3}
\begin{aligned}
a\int \left<\nabla u_q, \nabla v\right> &= \int R' (1+u_q)^{q-2} \chi^2 w^2 -
\int R \chi^2 w^2 \\
&\le \left|\int R \chi^2 w^2\right| \\
&\le \epsilon\|\nabla(\chi w)\|_{2}^2  + C_\epsilon \|\chi w\|_{2}^2.
\end{aligned}
\end{equation}
We applied Lemma \ref{lem:RIntegralBound} in the last line and used
the fact that for functions with support contained in a fixed compact set,
weighted and unweighted norms are equivalent.  Note also that the fact that $R'\leq 0$
everywhere is used in going from line 1 to line 2 in \eqref{eq:subcritbound3}.
Noting that there is a constant $C_2$ such that
\begin{equation}\label{eq:subcritbound4}
\| \nabla(\chi w)\|^2_2 \le C_2( \|\chi \nabla w\|_2^2 + \|w\nabla\chi\|_2^2 ),
\end{equation}
we can combine inequalities \eqref{eq:subcritbound2}, \eqref{eq:subcritbound3}, and
\eqref{eq:subcritbound4} to conclude that, if $\epsilon$ is sufficiently small
to absorb the term from inequality \eqref{eq:subcritbound3} into the left-hand side,
then there is a constant $C_3$ such that
\begin{equation}
\| \nabla(\chi w)\|^2_2 \le C_3 \left[ \|w\nabla\chi\|_2^2 + \|w \chi\|_2^2 \right].
\end{equation}
Finally, from the Sobolev inequality \eqref{eq:sobolev-cpct},
there is a constant $C_4$ such that
\begin{equation}
\| \chi w \|_N^2 \le C_4 \left[ \|w\nabla\chi\|_2^2 + \|w \chi\|_2^2 \right]
\end{equation}
as well.
Now $u_q$ is bounded uniformly in $L^N$ on the support $K'$ of $\chi$,
and hence we can take $\sigma$ sufficiently small so that $w$ is bounded
independent of $q$ in $L^2(K')$ as well. Thus $(1+u_q)$ is bounded uniformly
in $L^M(K)$ for $M=N(1+\sigma)$.
\end{proof}

\begin{cor} \label{cor:UniformBound} Let $p$ be the exponent such that $g$ is
a $W^{2,p}_{\alpha}$ AE manifold and let $\sigma\in (2-n,0)$.
The subcritical minimizers $u_q$ are bounded in $W^{2,p}_{\sigma}$ as $q\rightarrow N$.
\end{cor}
\begin{proof}
Consider a subcritical minimizer $u_q$, which is a weak solution of
\begin{equation}
-a\Lap u_q = -R(1+u_q) +R'(1+u_q)^{q-1}.
\end{equation}
Let $K$ be a compact set containing the support of $R$ and $R'$,
and let $M>N$ be an exponent such that we have uniform bounds on $u_q$
in $L^M(K)$.  We wish to bootstrap this to better regularity for $u_q$.

Since the bootstrap for the two terms is different, we concentrate
first on the interesting term,  $R'(1+u_q)^{q-1}$, and suppose for
the moment that the other term is absent. Let us write
\begin{equation}
\frac 1 M = \frac{1}{N} - \epsilon
\end{equation}
for some $\epsilon>0$.  Now
\begin{equation}
|R'(1+u_q)^{q-1}| \le |R'| (1 + |1+u_q|^{N-1}).
\end{equation}
Since $R'$ is bounded, the term $R'|1+u_q|^{N-1}$
belongs to $L^{s}(K)$ with
\begin{equation}
\begin{aligned}
\frac{1}{s} &= \frac{1}{M}(N-1)\\
&= \left(\frac{1}{N}-\epsilon\right)(N-1) \\
&= \frac{2}{n} + \frac{1}{N} - \epsilon (N-1).
\end{aligned}
\end{equation}
Since $R'$ is zero outside of $K$
we conclude $R'(1+u_q)^{q-1}\in L^{s}_{\sigma}$.
Note that the norm of $R'(1+u_q)^{q-1}$ in $L^{s}_{\sigma}$
depends on  the norm of $u_q$ in $L^M(K)$ but is otherwise independent
of $q$.  Since the functions $u_q$ are uniformly bounded in $L^M(K)$,
we obtain control of $R'(1+u_q)^{q-1}$ in $L^{s}_{\sigma}$ independent of $q$.

If $s\le p$ then $s\in (1,p]$ and
we cite \cite{Bartnik86} Proposition 2.2 to conclude
$u_q\in W^{2,s}_\sigma$
and therefore $u_q\in L^{M'}(K)$ with
\begin{equation}
\frac{1}{M'} = \frac{1}{s}-\frac{2}{n} = \frac{1}{N} - \epsilon (N-1).
\end{equation}
Similarly, after $k$ iterations of this process we would find $u_q$ belongs
to $W^{2,s}_\sigma$ with
\begin{equation}
\frac{1}{s} = \frac{2}{n} + \frac{1}{N} - \epsilon (N-1)^k
\end{equation}
unless $s > p$, at which point the bootstrap terminates at
$u_q\in W^{2,p}_{\sigma}$
with norm depending on $\|u_q\|_{L^M(K)}$ (which is independent of $q)$
and on the number of iterations needed to reach $s\le p$.
Note that since $N>2$, we will reach the condition
$s\ge p$ in a finite number of steps independent of $q$.

Now consider the bootstrap for the term $-R(1+u_q)$ alone.  Write
\begin{equation}
\frac 1 p = \frac 2 n - \epsilon'
\end{equation}
for some $\epsilon'>0$. The term $-R(1+u_q)$ then belongs to $L^{t}(K)$ with
\begin{equation}
\frac{1}{t} = \frac{1}{p}+\frac{1}{M} = \frac{2}{n}-\epsilon' +\frac{1}{M}.
\end{equation}
Note that $1<t<p$ and hence \cite{Bartnik86} Proposition 2.2 implies
$u_q\in W^{2,t}_{\sigma}$.  Note that the norm of $u_q$ in
$W^{2,t}_{\sigma}$ depends on the norm of $u_q$ in $L^M(K)$ but is otherwise
independent of $q$.  Consequently
$u_q$ is controlled in $L^{M'}(K)$ independent of $q$ where
\begin{equation}
\frac{1}{M'} = \frac{1}{t}-\frac{2}{n} = \frac{1}{M} - \epsilon'.
\end{equation}
After $k$ iterations we would find instead
\begin{equation}
\frac{1}{M'} = \frac{1}{M} - k\epsilon'
\end{equation}
and the bootstrap stops in finitely many steps independent of $q$
if $k\epsilon'>1/M$, at which point
we find that $u_q\in W^{2,p}_{\sigma}$, with norm
independent of $q$.
There is an exceptional case if $k\epsilon'=1/M$,
but it can be avoided by an initial perturbation of $M$.

The bootstrap in the full case follows from combining these arguments.
\end{proof}

\begin{proof}[Proof of Theorem \ref{thm:PrescribedScalarCurvature}]
\textbf{(2. implies 1.)}\quad The $u_q$ are uniformly bounded in $W^{2,p}_{\sigma}$
by Corollary \ref{cor:UniformBound} for any $\sigma\in(2-n,0)$.
Thus they converge to some $u$ strongly in
$W^{1,2}_{\delta^*}$ and uniformly on compact sets.
In particular, since the $u_q$ weakly solve \eqref{eq:almost-scalar-curvature},
$\phi:= u+1$ weakly solves
\begin{equation}
-a\Delta\phi + R\phi = R'\phi^{N-1}.
\end{equation}

Since each $u_q\geq -1$, $\phi\geq 0$, and since $\phi \to 1$
at infinity, $\phi\not\equiv 0$.  Hence the weak Harnack
inequality \cite{trudinger-measurable} implies $\phi>0$.

Since $\sigma\in(2-n,0)$ is arbitrary, $\phi-1\in W^{2,p}_{\alpha}$ in particular.
Note that the rapid decay $\sigma\approx 2-n$ uses the fact that $R=0$ near infinity.
The lesser decay rate $\alpha$ in the statement of the theorem
stems from the fact that we may have used a conformal factor in $W^{2,p}_\alpha$
to initially set $R=0$ near infinity or to lower the scalar curvature after
changing it to $R'$.

\textbf{(1. implies 2.)}\quad Let $Z=\{R'=0\}$. The case
where $Z$ has zero measure is trivial, for
then $y(Z) = \infty>0$.  Hence we assume
$Z$ has positive measure and suppose there exists a conformally related metric $g'$ with
scalar curvature $R'$. Let $\delta>\delta^*$ be fixed and let $u$ be a
minimizer of $Q_{g',\delta}$ over $A(Z)$ as provided
by Proposition \ref{prop:eigenfunctions}. Note that
\begin{equation}
\int {R'} u^2 dV_{g'} = 0
\end{equation}
since $R'=0$ on $Z$ and $u=0$ on $Z^c$.  Hence
\begin{equation}
\lambda_{g',\delta}(Z) = Q_{g',\delta}(u)
= a\frac{\int |\nabla u|^2_{g'}dV_{g'}}{\|u\|_{g',2,\delta}}.
\end{equation}
In particular, $\lambda_{g',\delta}(Z)\ge 0$, and $\lambda_{g',\delta}(Z)=0$
only if $u$ is constant. But $Z$ has positive measure, and therefore
$A(Z)$ does not contain any constants.  Hence $\lambda_{g',\delta}(Z)>0$,
and Proposition \ref{cor:MeasTFAE}  implies that $Z$ is Yamabe positive.
\end{proof}

%   On a
% compact manifold, we know that $u=0$ on $Z^c$, and hence since $Z^c$
% has positive measure, again $u$ is zero everywhere.
This completes the proof of Theorem \ref{thm:PrescribedScalarCurvature}.
Turning to the compact case (Theorem \ref{thm:PrescribedCompact})
recall that we started the AE argument with the following simplifying hypotheses:
\begin{enumerate}
\item The prescribed scalar curvature $R'$ is bounded.
\item The prescribed scalar curvature $R'$ has compact support.
\item The initial scalar curvature $R$ has compact support.
\end{enumerate}
The last two of these are trivial if $M$ is compact, and
the first is justified by  Lemma \ref{lem:LowerScalarCurvatureCompact} below,
which shows that we can lower scalar curvature after first solving the problem
for a scalar curvature that is truncated below. In the compact case we require
an additional condition which
will be used in Lemma \ref{lem:UniformLowerBoundForCompact}.
\begin{enumerate}
\setcounter{enumi}{3}
\item We may assume that the initial scalar curvature $R$ is continuous and negative.
Indeed, from Proposition \ref{prop:eigenfunctions} there is a positive function
$\phi$ solving $-a\Delta \phi + R\phi = \lambda_\delta(M) \phi$ on $M$.
Note that $\lambda_\delta(M)<0$ since $g$ is Yamabe negative.  Using
$\phi$ as the conformal factor we obtain a scalar curvature $\lambda_\delta(M) \phi^{2-N}$.
The hypotheses of Theorem 4.2 are conformally invariant and hence unaffected by this change.
\end{enumerate}
\begin{lem} \label{lem:LowerScalarCurvatureCompact}
Suppose $(M, g)$ is a $W^{2,p}$ compact Yamabe negative manifold.
Suppose $R'\in L^{p}$. If $0\geq R \geq R'$, then there exists a positive $\phi$ with
 $\phi \in W^{2,p}$ such that $g' = \phi^{N-2} g$ has scalar curvature $R'$.
\end{lem}
\begin{proof}
We wish to solve
\begin{equation}\label{eq:sclower}
-a \Delta \phi + R\phi = R'\phi^{N-1}.
\end{equation}
Note that $\phi_+=1$ is a supersolution of equation \eqref{eq:sclower}.
To find a subsolution first observe that $R\not \equiv 0$ since the manifold is Yamabe negative.
So, since $-R\ge 0$ and $-R\not\equiv 0$,
for each $\epsilon>0$ there exists a unique $\phi_\epsilon\in W^{2,p}$ solving
\begin{equation}\label{eq:sceps}
-a\Delta \phi_\epsilon - R\phi_\epsilon  = -R + \epsilon R'.
\end{equation}
If $\epsilon=0$ the solution is $1$, and since $W^{2,p}$ embeds continuously
in $C^0$ we can fix $\epsilon>0$ such that $\phi_{\epsilon}>1/2$
everywhere. We claim that $\phi_-:=\eta \phi_\epsilon$ is a subsolution
if $\eta>0$ is sufficiently small.  Indeed,
\begin{equation}\begin{aligned}
-a\Delta \phi_- + R\phi_- &= \eta\left[ R(2\phi_\epsilon-1) \right] + \eta\epsilon R'\\
& \le \eta\epsilon R'.
\end{aligned}\end{equation}
So $\phi_-$ is a subsolution so long as
\begin{equation}\label{eq:subsol2}
\eta\epsilon R' \le R'\phi_-^{N-1}.
\end{equation}
A quick computation shows that inequality \eqref{eq:subsol2} holds
if $\eta$ is small enough so that $\eta^{2-N}\ge \phi_\epsilon^{N-1}/\epsilon$
everywhere. We can also take $\eta$ small enough so that $\phi_-\le \phi_+=1$, and hence
there exists a solution $\phi\in W^{2,p}$ with $\phi\ge\phi_->0$
of equation \eqref{eq:sclower}  (\cite{Maxwell05b} Proposition 2).
\end{proof}

The remainder of the proof of Theorem \ref{thm:PrescribedCompact}
closely follows the proof of Theorem \ref{thm:PrescribedScalarCurvature}
by treating a
compact manifold as an asymptotically Euclidean manifold with zero ends.
In particular, the cited results of Section \ref{sec:FirstEigenvalue}
 apply equally in both cases,
and differences arise only when the following facts are cited.
\begin{itemize}
\item A constant function in $W^{1,2}_{\delta^*}$ is identically zero.
\item The Laplacian is an isomorphism from $W^{2,p}_\sigma$ to $L^p_\sigma$
for $\sigma\in (2-n,0)$.
\end{itemize}
Twice, we use the property that constants in $W^{1,2}_{\delta^*}$ vanish:
once in Lemma \ref{lem:subcriticalExistence} in showing $1+u_q\not\equiv 0$,
and once in the final proof of Theorem \ref{thm:PrescribedScalarCurvature}
showing that in the limit $1+u\not\equiv 0$ as well.  The following
lemma provides the alternative argument needed to ensure that these
functions do not vanish identically in the compact case.

\begin{lem}\label{lem:UniformLowerBoundForCompact}
Suppose $(M,g)$ is compact and that $R_g$
is continuous and negative. Fix $q_0\in (2,N)$. Then $\|1+u_q\|_2\geq C$ for some $C$
independent of $q\in (q_0,N)$. Moreover, the limit $1+u$ is not identically zero.
\end{lem}
\begin{proof}
Note that for any constant $k$,
\begin{equation}
F_q(k) = (1+k)^2\int R - \frac{2}{q} (1+k)^{q}\int R'.
\end{equation}
Since $\int R<0$, for any $k\neq -1$ close enough to $-1$, $F_q(k)<0$.
Indeed, there are constants $k_0>-1$ and $c>0$
such that $F_q(k_0)<-c$ for all $q\in (q_0,N)$.  But then
\begin{equation}
\int R (1+u_q)^2 \le F_q(u_q) \le F_q(k_0) \le -c
\end{equation}
since $u_q$ minimizes $F_q$. Since $R$ is continuous,
and thus bounded below, $\|1+u_q\|_2\geq C$ for some $C$ independent of $q\in(q_0,N)$.
Since $u_q \to u$ in $L^2$, we also have $\|1+u\|_2 \geq C$, and so $1+u$ is not identically
zero.
\end{proof}

We use the fact that $\|\Delta u\|_{p,\sigma}$ controls $\|u\|_{W^{2,p}_{\sigma}}$
twice as well, once in the bootstrap of Lemma \ref{lem:subcriticalExistence}
and once in the bootstrap of Lemma \ref{cor:UniformBound}.  However,
on a compact manifold, $\|u\|_{W^{2,p}}$ is controlled by the sum
of $\|\Delta u\|_{p}$ and $\|u\|_2$, and the coercivity estimate
from Proposition \eqref{prop:Coercivity} ensures that $\|u_q\|_2$
is uniformly bounded as $q\ra N$. This provides the needed extra
control for the bootstraps and
completes the proof of Theorem \ref{thm:PrescribedCompact}.

\section{Yamabe Classification}\label{sec:YamabeClassification}

In this section we provide two characterizations of the Yamabe class of an
asymptotically Euclidean manifold, one in terms of the prescribed scalar curvature
problem and one in terms of the Yamabe type of the manifold's compactification.
Note that throughout this section AE manifolds have at least one end.

\begin{thm}\label{YamabeClassification}
Suppose $(M,g)$ is a $W^{2,p}_\alpha$ AE manifold with $p>n/2$ and $\alpha \in (2-n,0)$.
Let $\mathcal{R}_{\le 0}$ be the set of non-positive elements of $L^p_{\alpha-2}$.
\begin{enumerate}
\item $M$ is Yamabe positive if and only if
the set of non-positive scalar curvatures of metrics conformally equivalent
to $g$ is $\mathcal{R}_{\le 0}$.
\item $M$ is Yamabe null if and only if
the set of non-positive scalar curvatures of metrics conformally equivalent
to $g$ is $\mathcal{R}_{\le 0}\setminus\{0\}$.
\item $M$ is Yamabe negative if and only if
the set of non-positive scalar curvatures of metrics conformally equivalent
to $g$ is a strict subset of $\mathcal{R}_{\le 0}\setminus\{0\}$.
\end{enumerate}
\end{thm}
\begin{proof}  It suffices to prove the forward implications.

\textbf{1)}\; Suppose $M$ is Yamabe positive, and hence so is every subset.
If $R'\in \mathcal{R}_{\le 0}$, then $\{R'=0\}$ is
Yamabe positive and Theorem \ref{thm:PrescribedScalarCurvature}
implies $[g]$ includes a metric with scalar curvature $R'$.

\textbf{2)}\; Suppose $M$ is Yamabe null. Since $M$ is open and connected,
Lemma \ref{InnerApproxNull} implies that if $E\subseteq M$ has
positive measure, then $M\setminus E$ is Yamabe positive.  Hence
for any $R'\in \mathcal{R}_{\le 0}$  with $R'<0$ on a set of positive measure,
$\{R'=0\}$ is Yamabe positive, and Theorem \ref{thm:PrescribedScalarCurvature}
implies we can conformally transform to a metric with scalar curvature $R'$.
But $R'\equiv 0$ is impossible, for otherwise Theorem \ref{thm:PrescribedScalarCurvature}
would imply $M$ is Yamabe positive.

\textbf{3)}\; Suppose $M$ is Yamabe negative.  Since $M$ is open, Lemma \ref{InnerApproxNeg}
shows that there is a nonempty open set $W\subseteq M$ such that $M\setminus W$ is
also Yamabe negative.  Suppose $R'\in L^{p}_{\alpha-2}$ is non-positive and supported
in $W$. Then $\{R'=0\}$ contains $M\setminus W$ and is hence Yamabe negative.
But then Theorem \ref{thm:PrescribedScalarCurvature} shows that
we cannot conformally transform to
a metric with scalar curvature $R'$. In particular, $R'\equiv 0$ is one of the
unattainable scalar curvatures.
\end{proof}

While Theorem \ref{YamabeClassification} completely describes
the set of allowable scalar curvatures in cases 1) and 2),
it does not in case 3).  Of course, we already have demonstrated a necessary and sufficient
criterion for being able to make the conformal change:
the zero set of $R'$ must be Yamabe positive. Nevertheless, it would be desirable
to describe this situation more explicitly, and there are a
few things that can be said. First, by Lemma \ref{lem:UniformOuterApproximation},
if $R' \in \mathcal{R}_{\le 0}$ and the weighted volume of $\{R'=0\}$ is sufficiently small,
then $\{R'=0\}$ is Yamabe positive, and
thus $g$ is conformally equivalent to a metric with scalar curvature $R'$.
In particular, if $R'<0$ everywhere, then it is attainable.  Conversely,
by Lemma \ref{InnerApproxNeg},
for any sequence $\{R'_k\} \subset \mathcal{R}_{\le 0}$ such that
$\{R'_k<0\} \subset B_{1/k}(x_0)$ for some fixed $x_0 \in M$, then for $k$ large enough,
$\{R'_k=0\}$ is Yamabe negative, and thus
$g$ is not conformally equivalent to a metric with scalar curvature $R'_k$.
That is, the strictly negative part of $R'$ cannot be constrained to
a small ball. Similarly, an argument analogous to the proof of Lemma \ref{InnerApproxNeg} shows that
the complement of a sufficiently ``small'' neighborhood of infinity
is Yamabe negative, and hence the strictly negative part of $R'$ cannot be
constrained to a small neighborhood of infinity.

Our second characterization of the Yamabe class of an AE manifold involves
its compactification.
An AE manifold can be compactified using a conformal factor that decays suitably
at infinity, and a compact manifold can be transformed into an AE manifold using
a conformal factor with a suitably singularity. We would like to show that
the sign of the Yamabe invariant is preserved under these operations,
and we begin by laying out the details of the compactification/decompactification
procedure.  In particular, there is a precise relationship between the decay of the
metric at infinity and its smoothness at the point of compactification.

\begin{lem}\label{lem:compactification}
Let $p>n/2$ and let $\alpha=\frac{n}{p}-2$, so $-2<\alpha<0$.
Suppose $(M,g)$ is a $W^{2,p}_\alpha$
AE manifold.
There is a smooth conformal factor $\phi$ that decays to infinity
at the rate $\rho^{2-n}$
such that $\bar g= \phi^{N-2} g$ extends to a $W^{2,p}$ metric on
the compactification $\Mbar$ .

Conversely, suppose
$(\Mbar,\gbar)$ is a compact $W^{2,p}$ manifold, with $p>n/2$ and $p\neq n$.
Given a finite set $\mathcal{P}$
of points in $\Mbar$ there is conformal factor $\phibar$
that is smooth on $M=\Mbar\setminus\mathcal{P}$, has a
singularity of order $|x|^{2-n}$ at each point of $\mathcal{P}$,
and such that $g=\phibar^{N-2}\gbar$ is a $W^{2,p}_{\alpha}$
AE manifold with $\alpha = \frac n p -2$.
\end{lem}
\begin{proof} For simplicity we treat the case of only one end.

Let $(M,g)$ be a $W^{2,p}_{\alpha}$ AE manifold and let $z^i$ be the Euclidean
end coordinates on $M$, so
\begin{equation}\label{eq:gDecomposition}
g_{ij} = e_{ij} +k_{ij},
\end{equation} with $k\in W^{2,p}_{\alpha}$. Let $x^i$ be coordinates given
by the Kelvin transform $x^i = z^i/|z|^{2}$, so $z^i = x^i/|x|^{2}$ as well.

We define a conformal factor $\phi = |z|^{2-n}$ near infinity, and extend
it to be smooth on the rest of $M$.
Let $\gbar = \phi^{N-2} g$ and let $\Mbar$ be the one-point compactification of $M$,
with $P$ being the point at infinity.
We wish to show that $\gbar$ extends to a $W^{2,p}(\Mbar)$ metric.

Near $P$, $\phi^{N-2} = |z|^{-4}$ and
\begin{align}\label{eq:gbarDecomposition}
\gbar_{ij} &= e_{ij} + \kbar_{ij}
\end{align}
where
\begin{equation}
\kbar_{ij} :=  k_{ij} - \frac{4}{|x|^2} x^a k_{a(i} x_{j)}
 + \frac{4}{|x|^4} x^a x^b k_{ab} x_i x_j = O(k).
\end{equation}
and $x_a = e_{ab} x^b$. Since $\kbar_{ij}\ra 0$ at $P$,
we set $\gbar_{ij}(P)=e_{ij}$ to obtain a continuous metric, and we
need to show that $\kbar\in W^{2,p}(\Mbar)$.
Since $\kbar \in W^{2,p}_\loc(M)$, and since a point is a
removable set, we need only show that
the second derivatives of $\kbar$ belong to $L^p(B)$ for
some coordinate ball $B$ containing $P$.

Let $\pbar$ represent the derivatives in $x^i$ coordinates. Since
$\frac{\p z}{\p x} = O(|x|^{-2})$, we calculate
\begin{equation} \label{eq:ChainRule} \begin{aligned}
\pbar \kbar &= O(\p k)O(|z|^2) + O(k)O(|z|)\\
\pbar^2 \kbar &= O(\p^2k) O(|z|^4) + O(\p k) O(|z|^3) + O(k)O(|z|^{2}).
\end{aligned}\end{equation}
In order to show $\pbar^2 \kbar \in L^p(B)$, it is sufficient to show that each
of the three terms in equation \eqref{eq:ChainRule} is in $L^p(B)$.

Note that near infinity
\begin{equation}
d\Vbar  = \phi^N dV = |z|^{-2n} dV.
\end{equation} Hence the $L^p$ norm of the
$O(k)O(|z|^2)$ term of equation \eqref{eq:ChainRule} is
controlled by
\begin{equation}\begin{aligned}
\int \left(O(k)O(|z|^2)\right)^p |z|^{-2n}dV
&= \int O\left(|k|^p\right) O\left(|z|^{2p-2n}\right) dV\\
&\leq C\|k\|_{W^{2,p}_{\alpha}}^p,
\end{aligned}\end{equation}
where we have used the equality
\begin{equation}
2p-2n = -n- \alpha p
\end{equation}
and expression \eqref{eq:SobolevNorm} defining the weighted norm.
Hence the $O(k)O(|z|^2)$ term of equation \eqref{eq:ChainRule}
belongs to $L^p(B)$. The two remaining terms have the same asymptotics and
similar calculations show that they belong to $L^p(B)$ as well.

For the converse, consider a $W^{2,p}$ compact manifold $(\Mbar, \gbar)$
with $p>n/2$ and $p\neq n$. Let $P$ be a point to remove to obtain
$M=\Mbar\setminus\{P\}$.  Since $\gbar$ is continuous we can find smooth
coordinates $x^i$ near $P$ such that $\gbar = e+\kbar$ for some $\kbar \in W^{2,p}$
which vanishes at $P$.
Moreover, if $p>n$ then $\gbar$ has H\"older continuous derivatives and the proof
of Proposition 1.25 in \cite{Aubin98} shows we can additionally assume
these are normal coordinates (i.e., the first derivatives of $\kbar$ vanish at $P$).
Finally, since the result we seek only involves properties of $\kbar$ local to $P$, we can
assume that $\kbar=0$ except in a small coordinate ball $B$ near $P$.

We claim there is a constant $C$ such that
\begin{align}
\label{eq:FiniteKbar}
\int_{B} \frac{|\kbar|^p}{|x|^{2p}} &\le C\int_{B} |\pbar^2 \kbar|^p d\Vbar\qquad\text{and}\\
\label{eq:FiniteKbar2}
\int_{B}  \frac{|\pbar \kbar|^p}{|x|^p}d\Vbar &\leq C \int_{B} |\pbar^2 \kbar|^p d\Vbar.
\end{align}
Assuming for the moment that this claim is true,
let $z_i = x_i/|x|^2$. Let $\phibar = |x|^{2-n}$ near $P$ and extend
$\phibar$ as a positive smooth function on the remainder of $M$.
Let $g= \phibar^{N-2} \gbar$.
Near $P$, $\phibar^{N-2} = |x|^{-4}$ and so $g= e+k$ near infinity, where
\begin{equation}
k_{ij} := \kbar_{ij} - \frac{4}{|z|^2} z^a \kbar_{a(i} z_{j)}
+ \frac{4}{|z|^4} z^a z^b \kbar_{ab} z_i z_j = O(\kbar).
\end{equation}
Since $k\in W^{2,p}_{\loc}$, we need only establish the desired
asymptotics at infinity.

A computation similar to the one leading to equation \eqref{eq:ChainRule} shows
\begin{equation}\begin{aligned}
\p k &= O(\pbar \kbar)O(|x|^2) + O(\kbar)O(|x|)\\
\p^2 k &= O(\pbar^2 \kbar) O(|x|^4) + O(\pbar \kbar) O(|x|^3) + O(\kbar)O(|x|^{2}).
\end{aligned}\end{equation}
Also, $d\Vbar = |z|^{-2n} dV$ near $P$. Hence
\begin{align}
\int |\p^2 k|^p |z|^{4p-2n} dV &= \int |\p^2 k|^p |x|^{-4p}|x|^{2n} dV \\
&=\int \left(O(\pbar^2 \kbar)\right)^p +\left(O(\pbar\kbar)O(|x|^{-1})\right)^p
+ \left(O(\kbar)O(|x|^{-2})\right)^p d\Vbar.\label{eq:Decompactification1}
\end{align} From inequalities \eqref{eq:FiniteKbar} and \eqref{eq:FiniteKbar2},
quantity \eqref{eq:Decompactification1} is finite.  Noting
 \begin{equation}
4p-2n = -n -\alpha p + 2p
 \end{equation}
we conclude $|\p^2 k|\in L^{p}_{\alpha-2}$, as desired. A similar calculation shows that
$|\p k| \in L^{p}_{\alpha-1}$ and $|k| \in L^p_{\alpha}$.  This concludes the proof,
up to establishing inequalities \eqref{eq:FiniteKbar} and \eqref{eq:FiniteKbar2}.

Theorem 1.3 of \cite{Bartnik86}
implies that
\begin{align}\label{eq:HardysInequality}
\int_{B} \frac{|f|^p}{|x|^{2p}} d\Vbar \leq c \int_{B} \frac{|\pbar f|^p}{|x|^p}d\Vbar
        \leq C \int_{B} |\pbar^2 f|^p d\Vbar <\infty
\end{align} for smooth functions $f$ that are compactly supported in $B$
and vanish in a neighborhood of $P$.
This inequality relies on the fact that $p\neq n$, which corresponds to the
condition $\delta=0$ in \cite{Bartnik86} Theorem 1.3.

Let $f_n$ be a sequence of smooth functions vanishing near $P$ that
converges to $\kbar$ in $W^{2,p}$; such a sequence exists since $\kbar=0$
at $P$, since $\partial \kbar =0$ at $P$ if $p>n$, and since we have assumed
that $\kbar$ vanishes outside of $B$.
By reduction to a subsequence we may assume that the values and first derivatives
of sequence converge pointwise a.e., and using Fatou's Lemma we find
\begin{equation}
\begin{aligned}
\int_{B} \frac{|\kbar|^p}{|x|^{2p}} &\leq \liminf_{n\to \infty} \int_{B} \frac{|f_n|^p}{|x|^{2p}}\\
      &\leq C \lim_{n\to\infty} \int_{B} |\pbar^2 f_n|^p d\Vbar \\
      &= C\int_{B} |\pbar^2 \kbar|^p d\Vbar <\infty.
\end{aligned}
\end{equation}
This is inequality \eqref{eq:FiniteKbar}, and a similar argument shows that
inequality \eqref{eq:FiniteKbar2} holds as well.
\end{proof}

The threshold $\alpha=-2$ in Lemma \ref{lem:compactification}
arises because there is a connection between the rate
of decay of the AE metric and the rate of convergence of the metric at the
point of compactification in a chosen coordinate system: roughly speaking,
decay of order $\rho^{\alpha}$
corresponds to convergence at a rate of $r^{-\alpha}$.  For a generic smooth metric we can
use normal coordinates to obtain convergence at a rate of $r^2$, but we cannot
expect to do better generally.  Hence the decompactification of a smooth metric
typically does not decay faster than $\rho^{-2}$.  Looking at the
proof of Lemma \ref{lem:compactification}, we note that it can be
readily extended to $s>2$ to show that
a $W^{s,p}_\alpha$ AE metric with $s\geq 2$, $p>n/s$
and $\alpha=(n/p)-s$
can be compactified to a $W^{s,p}$ metric. But the decay condition
$\alpha=(n/p)-s$ is quite restrictive for $s>2$:
smooth metrics decompactify generally to metrics with decay $O(\rho^{-2})$,
but compactification of a $W^{s,p}_{-2}$ metric would not be known to be
$C^3$, regardless of how high $s$ and $p$ are.  A more refined analysis
for $s>2$ would need to take into account asymptotics of the Weyl or Cotton-York
tensor, and we point to Herzlich \cite{Herzlich97} for related results
in the $C^s$ setting.

\begin{prop}\label{prop:yAE=yCpct}
Let $(M,g)$ and $(\Mbar,\gbar)$ be a pair of manifolds as in Lemma
\ref{lem:compactification}, related by $g = \phibar^{N-2} \gbar$.
Then $y_g(M)=y_{\gbar}(\Mbar)$.
\end{prop}
\begin{proof}
For simplicity we assume that $M$ has one end.  Let $P\in \Mbar$
be the singular point of $\phibar$.
Note that $W^{1,2}_c(M)$ is dense in $W^{1,2}_{\delta^*}(M)$
and that
\begin{equation}
S_P:=W^{1,2}(\Mbar)\cap\{u: u|_{B_r(P)}=0 \textrm{ for some } r>0\}
\end{equation} is dense in
$W^{1,2}(\Mbar)$ since $2<n$.
From upper semicontinuity of the Yamabe quotient, the Yamabe
invariants of $g$ and $\gbar$ can be computed by minimizing
the Yamabe quotient over $W^{1,2}_c$ and $S_P$ respectively.
Note that $u \mapsto \phibar u$ is a bijection between $W^{1,2}_c(M)$
and $S_p$. The proof of Lemma \ref{ConformalInvariance}
shows that for $u\in W^{1,2}_c$,
\begin{equation}
Q_g^y(u) = Q_{\gbar}^y(\phibar u)
\end{equation}
and hence $y_g(M) = y_{\gbar}(\Mbar)$.
\end{proof}

Combining Lemma \ref{lem:compactification} and Proposition \ref{prop:yAE=yCpct}
we obtain our second classification.
\begin{prop}\label{prop:classifybycompactify}
Let $(M,g)$ be a $W^{2,p}_{\alpha}$ AE manifold with $\alpha \le \frac{n}{p}-2$.
Then $(M,g)$ is Yamabe positive/negative/null if and only if
some conformal compactification, as described in Lemma \ref{lem:compactification},
has the same Yamabe type.
\end{prop}

Consequently, Yamabe classification on AE manifolds has the same topological flavor
as in the compact setting. For instance, since the
torus does not allow a Yamabe positive metric, the decompactified torus,
which is diffeomorphic to $\R^n$ with a handle, does not allow a metric with nonnegative
scalar curvature.

We mention an application of Proposition \ref{prop:classifybycompactify}
to general relativity. Recall the Einstein constraint equations, \eqref{eq:Constraints},
which initial data $(M, g, K, T_{nn}, T_{ni})$ in general relativity must satisfy.
It is natural to suppose that the energy density $T_{nn}$ is everywhere nonnegative,
which is known as the weak energy condition.
If the initial data is maximal, i.e., if the mean curvature $\tr\, K$ is zero, then
the weak energy condition implies $R\geq 0$. Thus, if
the compactification of an AE manifold has a topology that does not
admit a Yamabe positive metric, then
the original AE manifold does not allow maximal initial data satisfying
the weak energy condition. We mention that the results in \cite{IMP02} show that every
AE manifold does admit \emph{some} solution of the constraints. 

\chapter{Solutions of the Conformal Constraint Equations} \label{chap:FarFromCMC}
%%%%%%%%%%%%%%%%%%%%%%%%%%%%%%%%%%%%%%%%%%%%%%%%%%%%%%%%%%%%%%%%%%%%%%%%%%%%%%%%%%%%%%%%% 

In this chapter, we find solutions of the conformal constraint equations
\eqref{eq:ConfConst} following the setup introduced in \cite{HNT09}
and \cite{Maxwell09}. This method, based on a fixed point theorem, is the first to allow
solutions with arbitrary mean curvatures. In Sections \ref{sec:FixedPoint} 
and \ref{sec:Supersolutions}, we prove the
appropriate analogues on asymptotically Euclidean manifolds. This is mostly work 
with Isenberg, Mazzeo, and Meier from
\cite{DIMM14}.

Nguyen \cite{Nguyen14} later showed via a scaling argument that these solutions 
with arbitrary mean curvatures could alternatively be interpreted 
as rescalings of perturbations of CMC results. We discuss this argument. Also,
Nguyen presented a new method for finding solutions to the conformal constraint
equations, using half-continuity and a fixed point theorem. In Section 
\ref{sec:LocalSup}, we
present a simpler proof of his result, and marginally strengthen it.

\section{The Fixed Point Approach}\label{sec:FixedPoint}

A standard method of solving differential (and other) equations is the fixed point
method. In this method, one first finds a functional whose fixed points are 
solutions of the desired equation. One then uses a priori estimates and other 
properties of the functional to fulfill the conditions of one of the many fixed point
theorems, such as the Schauder fixed point theorem. These theorems guarantee fixed
points of functionals under very general circumstances. A good introduction to these
techniques is found in \cite{Brown04}.

One of the first and most important fixed point theorems is the Brouwer fixed point 
theorem. In its basic form, it says that any continuous function from a ball in $\R^n$
to the same ball has a fixed point. Many other fixed point theorems, such as the ones
we use in this chapter, the Schauder fixed point theorem and the Leray-Schauder
alternative, are based on the Brouwer fixed point theorem. We now state the Schauder 
fixed point theorem.

\begin{thm}[Schauder Fixed Point Theorem]\label{thm:SchauderFP}
Let $S$ be a closed convex subset of a normed linear space $X$ and let $F:S \to S$ be a
compact map. Then $F$ has a fixed point. 
\end{thm}

We introduce the map $F$ that we use in this chapter. In essence, the map $F$
is an iteration map, taking a function $\phi$, solving the vector equation for some $W$
using that $\phi$, and
then solving the Lichnerowicz equation using that $W$. In this way, we iterate the coupling,
which allows us to find solutions at each step. This is similar to the proof of the 
sub and supersolution theorem \ref{thm:SubSupersolutionTheorem}, where we set up an
iteration scheme, using the solution of the previous step in order to fix the 
nonlinearity of the equation $-a\Delta u = f(x,u)$.

To be more precise, for any positive function
$v \in L^\infty$, let $W(v)$ be the solution in $W^{2,p}_\delta$ of the vector equation
\eqref{eq:OrigVect} with $v$ replacing $\phi$.
Let $G(W)$ to be the solution of the Lichnerowicz equation
using $W$ such that $G(W)-\mru \in W^{2,p}_\delta$, where $\mru$ is the desired asymptotic
function. Let 
$E:W^{2,p}_\delta \to L^\infty$ be the compact Sobolev embedding map give by
Proposition \ref{prop:SobolevEmbeddings}. We then define
$F(v):= (E\circ G\circ W)(v)$. Clearly, if $F(\phi) = \phi$, then $(\phi,W(\phi))$ is a
solution to the conformal constraint equations \eqref{eq:ConfConst}.

In order for $F$ to be well defined, there are two requirements. First, for the
vector equation to have a solution, we must require that $g$ has no conformal Killing
fields. Recall that if $p>n$ and $g$ is a $W^{2,p}_\delta$ or $C^{2,\alpha}_\delta$
AE manifold, then this is known to be true (cf. Theorem \ref{prop:Isomorphism}).
Next, for the Lichnerowicz equation to have a solution, we must assume that the 
seed data is admissible (cf. Definition \ref{def:Admissible}); i.e., that there
is a conformal factor $\psi$ that transforms the metric to one with scalar 
curvature $-\kappa \tau^2$.

In order to define the domain set
$S$, recall the definition of global sub and supersolutions.

\begin{defn}
Functions $\phi_\pm$ are ``global sub and supersolutions" of the Lichnerowicz equation
\eqref{eq:OrigLich} if
for any $\phi\leq \phi_+$, $\phi_+$ is a supersolution and $\phi_-$ is a subsolution 
of the Lichnerowicz equation with $W=W(\phi)$. 
\end{defn}

 Suppose we could find global 
sub and supersolutions $\phi_\pm$ with properties as in the sub and supersolution theorem
\ref{thm:SubSupersolutionTheorem}. Then 
$S= \{\phi \in L^\infty: \phi_- \leq \phi \leq \phi_+\}$ is clearly closed and convex,
and by the sub and supersolution theorem, $F:S\to S$. 

In order to show that $F$ is compact, first recall that the composition of compact maps
and continuous maps is compact. The solution map $W$ is continuous by the continuity
of $\phi^q$ and of $\left(\di \frac{1}{2N} L\right)^{-1}$. The map
$E$ is compact. Thus we only need to show that $G$ is continuous. We show that it is in
fact $C^1$ in the following lemma.

\begin{lem}
Given appropriately regular and admissible data, the solution map 
$G:W^{2,p}_{\delta} \to W^{2,p}_\delta$ is continuously G\^{a}teaux
differentiable. A similar statement holds for $C^{2,\alpha}_\delta$ data.
\end{lem}
\begin{proof}
This lemma follows from the implicit function theorem. The proof we give here is 
essentially the same as is used in proving \cite[Prop 13]{Maxwell09}, which is the
equivalent result for compact manifolds. Since the proof is identical for 
$C^{2,\alpha}_\delta$ data, we only prove the first case. For brevity, we
 equivalently prove that the solution map $G$ is continuous on the quantity
$\beta = \sigma+\frac{1}{2N} LW$.

Let $\beta_0 \in W^{1,p}_{\delta-1}$ and set $\psi_0 = G(\beta_0)$. 
We then define $\Gtil(\beta):= \psi_0^{-1} G(\psi_0^2 \beta)$. Thus, for 
$\gtil = \psi_0^{q-2} g$ and $\rtil= \psi_0^{-\frac32 q +1}$, conformal
covariance \ref{prop:ConformalCovariance} implies that
$\phi=\Gtil(\beta)$ is the solution of
\begin{equation}
-a\Delta_{\gtil} \phi
 + R_{\gtil}\phi +\kappa\tau^2 \phi^{q-1} - |\beta|^2_{\gtil} \phi^{-q-1}
   - \rtil \phi^{-\frac{q}{2}}=0
\end{equation} such that $\phi-1 \in W^{2,p}_\delta$.

Thus to show that $G$ is continuous in a neighborhood of $\beta_0$, we only need to
show that $\Gtil$ is continuous near $\psi_0^{-2}\beta_0$.  We remark that
$\Gtil(\psi_0^{-2} \beta_0) \equiv 1$.

We define the map $\Phi: W^{2,p}_\delta \times W^{1,p}_{\delta-1} \to L^{p}_{\delta-2}$
by
\begin{equation}
\Phi(\phi,\beta) = -a\Delta_{\gtil} \phi + R_{\gtil} \phi +
                   \kappa\tau^2 \phi^{q-1} - |\beta|_{\gtil}^2 \phi^{-q-1}
                   - \rtil \phi^{-\frac{q}{2}}.
\end{equation}
Note that $\Phi(\Gtil(\beta),\beta) = 0$. The G\^{a}teaux derivative of $\Phi$ is
given by
\begin{multline}
D\Phi_{\phi,\beta}(h,k) = -a\Delta_{\gtil} h + R_{\gtil}h +\kappa(q-1)\tau^2 \phi^{q-2}h \\
\qquad + (q+1)|\beta|^2 \phi^{-q-2} h -  2\phi^{-q-1} \langle \beta, k\rangle 
     +\frac{q}{2} r \phi^{-\frac{q}{2}-1} h.
\end{multline}

Thus
\begin{equation}
D\Phi_{1,\beta_0}(h,0) = -a\Delta_{\gtil} h + R_{\gtil} h +\kappa(q-1)\tau^2 h - (q+1)|\beta|^2 h 
   - \frac{q}{2} r h.
\end{equation}

However, since $G(\psi_0^{-2}\beta_0) \equiv 1$, we have that
\begin{equation}
R_{\gtil} = -\kappa\tau^2 + |\beta_0|^2 + r,
\end{equation}
and so
\begin{equation}
D\Phi_{1,\beta_0}(h,0) = -a\Delta_{\gtil} h  +\left[\kappa(q-2)\tau^2  + (q+2)|\beta|^2 
    + \frac{q+2}{2} r\right] h.
\end{equation}

Since the coefficient of $h$ is positive and is contained in $L^p_{\delta-2}$ by
assumption, we see from Proposition \ref{prop:Isomorphism} that
 $D\Phi_{1,\beta_0}: W^{2,p}_{\delta} \to L^p_{\delta-2}$ is an
isomorphism. The implicit function theorem on 
Banach spaces then implies that $G$ is $C^1$ in a neighborhood of $\beta_0$.
\end{proof}

On compact manifolds, Maxwell \cite{Maxwell09} showed that, in fact, no global 
subsolution is needed to construct a solution of the conformal constraint equations.
For Yamabe nonnegative metrics, he proves this using a Green's function argument to
find a uniform lower bound on solutions. In the Yamabe negative case, he uses
a conformal factor transforming the metric to one with scalar curvature $-\kappa\tau^2$
as a global subsolution. In the asymptotically Euclidean case, this conformal factor 
argument works for all Yamabe classes, which makes the argument simpler. 
We obtain the following existence theorem.

\begin{thm}\label{thm:CombinedExistence}
Assume that the admissible seed data $(g, \tau, N, \sigma,r,J)$ has the regularity specified in 
\eqref{eq:DataRegularity}. Assume that $r\geq 0$ and that $g$ admits no conformal Killing
fields. Suppose there exists a positive global supersolution $\phi_+$, satisfying the 
hypotheses of the sub and supersolution theorem \ref{thm:SubSupersolutionTheorem}.
Then for any asymptotic
function $\mru$ with asymptotics less than that of $\phi_+$, there exist
$(\phi,W)$ solving the conformal constraint equations \eqref{eq:ConfConst}
such that $\phi$ is positive and $\phi-\mru$ and $W$ are in $W^{2,p}_\delta$.

A similar statement holds for $C^{2,\alpha}_\delta$ data.
\end{thm}
\begin{proof}
We first find a global subsolution. Since the data is admissible, let $\psi$ be the
positive conformal factor transforming the metric to one with scalar curvature
$-\kappa \tau^2$. Let $\phi_- = \alpha \psi$. As in the proof of Theorem
\ref{thm:LichIff}, $\phi_-$ is a subsolution of the Lichnerowicz equation
\eqref{eq:OrigLich} for any $\alpha \leq 1$, regardless
of what $W$ is. Thus $\phi_-$ is a global subsolution. We then choose $\alpha$
small enough such that $\phi_-\leq \phi_+$ and such that the asymptotics of $\mru$
are greater than those of $\phi_-$.

The sub and supersolution theorem \ref{thm:SubSupersolutionTheorem} and the uniqueness
theorem \ref{thm:LichUniqueness} now guarantee that the solution map $G$
is well defined. Since $\phi_\pm$ are global sub and supersolutions, $F$ maps 
$S:= \{\phi \in L^\infty: \phi_-\leq\phi\leq\phi_+\}$ into itself. As discussed above,
$F$ is also a compact map. Thus the Schauder fixed point theorem \ref{thm:SchauderFP}
shows there is a fixed point $\phi$ of $F$. Thus $(\phi, W(\phi))$ is a solution to
the conformal constraint equations.
\end{proof}

Theorem \ref{thm:CombinedExistence} shows that, for any fixed choice of the seed data 
$(g, \tau, N, \sigma, r, J)$ and supersolution $\phi_+$ satisfying the hypotheses
of the theorem, there is at least a $k$-dimensional family of solutions (where $k$ is the
number of ends), parameterized by the product of the intervals $(0, \phi_{+,i}]$, where
$\phi_+ \to \phi_{+,i}$ on the end $E_i$. This nonuniqueness leads one to enquire about
the full extent of these families of solutions: for what asymptotic functions $\mru$ are
there solutions of the conformal constraint equations? Unfortunately, neither the
necessary analysis of the linearizations of the operators in 
\eqref{eq:ConfConst}, nor the a 
priori estimates for the solutions, is clear at this time, so we do not
yet have more definitive results on the full family of solutions.

%%%%%%%%%%%%%%%%%%%%%%%%%%%%%%%%%%%%%%%%%%%%%%%%%%%%%%%%%%%%%%%%%%%%%%%%%%%%%%%%%%%%%%%%%
\section{Global Supersolutions} \label{sec:Supersolutions}
%%%%%%%%%%%%%%%%%%%%%%%%%%%%%%%%%%%%%%%%%%%%%%%%%%%%%%%%%%%%%%%%%%%%%%%%%%%%%%%%%%%%%%%%

We have reduced the problem finding solutions of the conformal constraint equations 
\eqref{eq:ConfConst} to that of finding an appropriate global 
supersolution. We now present two lemmas which prove useful in finding such
supersolutions.

\begin{lem}\label{lem:Falloff}
Assume that $g$ is a $W^{2,p}_\delta$ AE metric with vanishing scalar curvature.
There is a unique solution $w$ to $-a\Delta w = \rho^{\gamma-2}$ with 
$w = c_\gamma\, \rho^\gamma + \hat{w}$, $c_\gamma = (\gamma^2 + (n-2)\gamma)^{-1}$, and
$\hat{w} \in W^{2,p}_{\gamma'}$ where $\gamma' = 2\gamma$ if this
number is greater than $2-n$ (or else $\gamma' \in (2-n, \gamma)$).

Similarly, if $g$ is a $C^{2,\alpha}_\delta$ AE metric then this unique solution $w$
decomposes as $c_\gamma \rho^{\gamma} +\hat{w}$, with $\hat{w} \in C^{2,\alpha}_{\gamma'}$.
\end{lem}
\begin{proof}
Write $w = c_\gamma\, \rho^\gamma + \hat{w}$ and let $\gbar$ be a $W^{2,p}$ metric which
agrees with $g$ away from the ends but is exactly Euclidean on each $E_j$. Then we
must solve
\begin{multline}
(-a\Delta + R)\hat{w} = \left( \rho^{\gamma-2} 
   - c_\gamma (-a\Delta_{\gbar} + R_{\gbar}) \rho^{\gamma}\right)
      - c_\gamma \left( (-a\Delta + R) - (-a\Delta_{\gbar} + R_{\gbar})\right) \rho^\gamma.
\end{multline}
The first term on the right is $L^p$ with compact support, while the second term lies
in $L^p_{2\gamma - 2}$, so the entire right hand side lies in 
$L^p_{2\gamma-2} \subset L^p_{\gamma'-2}$. 
The result follows from Theorem \ref{prop:Isomorphism}.

The proof in the H\"older setting is the same.
\end{proof}

The second lemma is a slight weakening of the elliptic estimate 
\eqref{eq:PEstimate2}, which is adequate for our purposes.

\begin{lem}\label{lem:boundLW}
If $(M,g)$ is AE and has no conformal Killing fields, and if $f \in L^p_{\delta-2}$
with $p>n$, then the unique solution $W\in W^{2,p}_\delta$
to $\di \frac{1}{2N} LW = f$ satisfies
\begin{equation}
\label{eq:LWest}
\|LW\|_{\infty} \le C_1 \rho^{\delta-1} \|f\|_{p,{\delta-2}}.
\end{equation}
\end{lem}
\begin{proof}
Combining \eqref{eq:PEstimate2} and Sobolev embedding \ref{prop:SobolevEmbeddings}, we get
\begin{equation}
\rho^{1-\delta}|L W| \leq \|LW\|_{C^0_{\delta-1}} \leq C_1'\|L W\|_{W^{1,p}_{\delta-1}}
\leq C_1' \|W\|_{W^{2,p}_{\delta}} \leq C_1\|f\|_{p,{\delta-2}},
\end{equation}
which implies \eqref{eq:LWest}.
\end{proof}

In the first main result of this section we construct global supersolutions, 
allowing the mean curvature to be arbitrary but requiring
that the other data (except the metric) be quite small.

\begin{thm}[Far-from-CMC Global Supersolution] \label{thm:FarCMCSupsolution}
Suppose that $(M,g)$ is a $W^{2,p}_{\gamma}$ Yamabe positive AE manifold,
with $p>n$ and $\gamma \in (2-n,0)$, and set $\delta = \gamma/2$. Fix 
$\tau \in W^{1,p}_{\delta-1}$ and $N\in W^{2,p}_\delta$. Suppose $\sigma \in
L_{\delta-1}^{\infty}$, nonnegative $r \in L^{\infty}_{2\delta-2}$ and 
$J \in L^{p}_{\delta-2}$ are sufficiently small
(depending on $\tau$, $g$ and $n$). Then, for any $\mru$, there exists a global
supersolution $\phi_+>0$ with $\phi_+ - \eta\mru \in W^{2,p}_{\gamma'}$ for some
constant $\eta > 0$ and any $\gamma'>\gamma$.

Similarly, if $(M,g)$ is a $C^{2,\alpha}_\gamma$ Yamabe positive AE manifold,
and if the corresponding H\"older norms of $\sigma$, $J$ and $\rho$ are sufficiently
small, then there exists a global supersolution $\phi_+$
with $\phi_+ - \eta\mru \in C^{2,\alpha}_\gamma$ for some $\eta > 0$.
\end{thm}
The main ideas used in this proof are similar to those used  in the compact case,
but there are new issues arising in the
construction of the supersolution on each end.  
Because the proofs in the Sobolev and H\"older settings are identical,
we present only the former.
\begin{proof}
Since $g$ is Yamabe positive, by conformal covariance \ref{prop:ConformalCovariance},
we may assume without loss of generality that $R\equiv 0$. By Lemma \ref{lem:Falloff},
there exists a (unique) $\Psi = \mru + c_\gamma \rho^\gamma +
\hat{\Psi}$, with $\hat{\Psi} \in W^{2,p}_{2\gamma}$, such that
\begin{equation}
\label{Psi}
-a\Delta\Psi = \rho^{\gamma-2},
\end{equation}
or equivalently
\begin{equation}
-a\Delta (\Psi-\mru) = \rho^{\gamma-2}.
\end{equation}
Note that, by the maximum principles \ref{prop:MaxPrinciple} and
\ref{prop:StrongMaxPrinciple}, $\Psi > 0$.

Now set $\phi_+ = \eta\Psi$, where the constant $\eta>0$ is to be chosen below.
We claim that, for appropriate $\eta$, $\phi_+$ is a global supersolution. To verify
this, we first note that from \eqref{eq:LWest}, with $f =\kappa\phi^q d\tau + J$, we have
\begin{equation}
\label{eq:DWestim}
\|LW\|_{\infty} \le C \rho^{\delta-1}\left(\|d\tau\|_{p,{\delta-2}}\|\phi\|^{q}_{\infty}
    +\|J\|_{p,{\delta-2}} \right),
\end{equation}
and hence
\begin{equation}
\left|\sigma + \frac1{2N}LW\right|^2 \le C \rho^{2\delta -2}(\|d \tau\|^2_{p,{\delta-2}} \|\phi\|^{2q}_{\infty} +
\|\sigma\|_{\infty,{\delta-1}}^2 + \|J\|^2_{p,{\delta-2}}).
\end{equation}

Since $\Psi$ decays at the precise rate $\rho^\gamma$ (and is strictly positive), then
deleting subscripts denoting the norms for simplicity, we calculate
\begin{multline}
-a\Delta\phi_+ + \kappa \tau^2 \phi_+^{q-1} - \left|\sigma + \frac1{2N}LW\right|^2\phi_+^{-q-1}
  - r \phi_+^{-q/2} \geq \\
   \eta \, \rho^{\gamma-2} - \rho^{2\delta-2}\left( C_1 \eta^{q-1} 
  + C_2 \eta^{-q-1} (\|\sigma\|^2 + \|J\|^2) + C_3 \eta^{-q/2} \|r\|\right).
\end{multline}
The constants $C_1$, $C_2$ and $C_3$ depend only on $\rho$, $N$, and the dimension $n$.
Since $2\delta-2 = \gamma -2 < 0$ and $q-1>1$, we first choose $\eta$ sufficiently small so that
\begin{equation}
\frac12  \eta \, \rho^{\gamma-2}- C_1 \eta^{q-1} \rho^{2\delta-2} > 0,
\end{equation}
and then choose $\|\sigma\|$, $\|J\|$ and $\|\rho\|$ sufficiently small 
(depending on $C_1$, $F$, $n$ and $\eta$), so that
\begin{equation}
\frac12 \eta \, \rho^{\gamma-2} - \rho^{2\delta-2} \left( C_2 \eta^{-q-1} ( \|\sigma\|^2 + \|J\|^2) 
     + C_3 \eta^{-q/2} \|\rho\|\right) > 0
\end{equation}
as well.  This proves that $\phi_+$ is a global supersolution.
\end{proof}

The second main result is the existence of a global supersolution for near-CMC data,
i.e., where $d\tau$ is sufficiently small as compared to $\tau$.

\begin{thm}[Near-CMC Global Supersolution]\label{thm:NearCMCSup}
Suppose that $(M,g)$ is a $C^{2,\alpha}_{\gamma}$, Yamabe positive AE manifold,
where $\gamma \in (2-n,0)$, and set $\delta = \gamma/2$. Fix data 
$\tau \in C^{1,\alpha}_{\delta-1}$, $N\in C^{2,\alpha}_\delta$, $\sigma \in
C^{0,\alpha}_{\delta-1}$, nonnegative $r \in C^{0,\alpha}_{\delta-2}$ and 
$J \in C^{0,\alpha}_{\delta-2}$. Then, there exists a constant $B>0$, depending on
the seed data, but not on $\tau$, such that if $\tau$ satisfies
$\tau^2 - B\|d\tau\|^2_{C^{0,\alpha}_{\delta-2}}\rho^{2\delta-2}\geq 0$, then
there exists a global
supersolution $\phi_+>0$ with $\phi_+ - \eta \in W^{2,p}_{\gamma}$ for any
constant $\eta > 0$ sufficiently large.
\end{thm}
\begin{remark}\label{rmk:HolderVsSobolev}
The hypothesis $\tau^2 - B\|d\tau\|^2_{C^{0,\alpha}_{\delta-2}}\rho^{2\delta-2}\geq0$
is precisely where the use of H\"older rather than Sobolev data is important for 
asymptotically Euclidean data.  Indeed, if $\tau$ satisfies this inequality, then in
particular, $\tau \geq C \rho^{\delta-1}$, so the norm of $\tau$ in $L^p_{\delta-1}$ is
necessarily infinite.  

The condition $\tau^2-B\|d\tau\|^2_{C^{0,\alpha}_{\delta-2}}\rho^{2\delta-2}\geq 0$,
which in particular imposes a lower bound
on the decay of $\tau$ and requires that $\tau$ never vanishes, may not be fulfilled
by any functions $\tau$. Indeed, as we show in Chapter 
\ref{chap:LimitEquation}, similar near-CMC conditions are not always fulfilled.
\end{remark}
\begin{proof}
Since $\tau$ never vanishes, the data is admissible 
(cf. Definition \ref{def:Admissible}), and so by conformal
covariance \ref{prop:ConformalCovariance},
we may assume without loss of generality that $R=-\kappa \tau^2$.
By Theorem \ref{prop:Isomorphism}, there exists a solution $u$ to
\begin{equation}
-a\Delta u = r + |\sigma|^2
\end{equation} with $u-1 \in C^{2,\alpha}_\delta$. 
Indeed, this is equivalent to
\begin{equation}
-a\Delta (u-1) = r+|\sigma|^2 \geq 0,
\end{equation} and so the maximum principle \ref{prop:MaxPrinciple} shows that
$u\geq 1$. By the estimate \eqref{eq:PEstimate3}, $\sup u$ is bounded and depends
only on $r, |\sigma|$ and $g$.

Now set $\phi_+ = \eta u$, where $\eta$ is chosen below, and using estimate
\eqref{eq:DWestim} and the inequality $\phi\leq \phi_+$, we have
\begin{multline}
|LW|^2 \leq C^2 \rho^{2\delta-2} ( (\sup \phi)^{q}\|d\tau\|_{C^{0,\alpha}_{\delta-2}}
     + \|J\|_{C^{0,\alpha}_{\delta-2}})^2  \\
\leq 2C^2 \rho^{2\delta-2} ((\sup \eta u)^{2q}\|d\tau\|_{C^{0,\alpha}_{\delta-2}}^2
     +\|J\|^2_{C^{0,\alpha}_{\delta-2}}) \\
\leq 2C^2 \rho^{2\delta-2} \left((\sup u)^{2q}\|d \tau\|_{C^{0,\alpha}_{\delta-2}}^2(\eta u)^{2q}
     +\|J\|^2_{C^{0,\alpha}_{\delta-2}}\right);
\end{multline}
the constant $C$ is the same one appearing in \eqref{eq:DWestim}.

Dropping the subscripts on the norms, and using the fact that 
$\tau^2 \geq C \rho^{2\delta-2}$ for some $C>0$, we calculate 
\begin{multline}
-a\Delta \phi_+ - \kappa \tau^2 \phi_+ + \kappa \tau^2 \phi_+^{q-1} 
   - \left|\sigma + \frac{1}{2N} LW\right|^2 \phi_+^{-q-1} - r\phi_+^{-q/2}\\
\geq \kappa \tau^2 \left[(\eta u)^{q-1} - \eta u\right] 
  + r\left[\eta - (\eta u)^{-q/2}\right] - \left|\sigma + \frac{1}{2N} LW\right|^2 \phi_+^{-q-1}
    +|\sigma|\eta\\
\geq \frac12 \kappa \tau^2 \left[(\eta u)^{q-1} - \eta u - C \|J\|\|\tau\|^{-2} (\eta u)^{-q-1}\right]
   + (\eta u)^{q-1}\left[ \frac12 \kappa \tau^2 - C \rho^{2\delta-2} (\sup u)^{2q} \|d\tau\|^2\right]\\
   \geq 0,
\end{multline} for all $\eta$ large enough, as long as 
\begin{equation}
C \rho^{2\delta-2}(\sup u)^{2q}  \|d\tau\|^2 \leq \tau^2,
\end{equation} which has been assumed. This shows that $\phi_+$ is a global supersolution. 
\end{proof}

We now present a scaling argument of Nguyen \cite{Nguyen14}, which
shows that the far-from-CMC result \ref{thm:FarCMCSupsolution} is simply
a rescaling of a near-CMC result.

\begin{prop}\label{prop:NguyenScaling}
$(\phi,W)$ is a solution of the conformal constraint equation 
\eqref{eq:ConfConst} for the seed data $(g, \tau, N, \sigma,r,J)$
if and only if $(C^{-1}\phi, C^{-q/2 -1}W)$ is a solution of the conformal
constraint equations for the data 
$(g, C^{q/2-1}\tau, N, C^{-q/2-1}\sigma, C^{-q/2-1}r, C^{-q/2-1}J)$.
\end{prop}
\begin{proof}
Plugging the scaled data and solutions into conformal constraint equations gives
this immediately.
\end{proof}

\begin{cor} \label{cor:FarIsNear}
Given a far-from-CMC solution to the conformal constraint equations, as given by
Theorem \ref{thm:FarCMCSupsolution}, there is a solution of the conformal constraint
equations equivalent to it, in the sense of Proposition \ref{prop:NguyenScaling},
which is a perturbation of the CMC case, $\tau \equiv 0$.
\end{cor}

It is well known (see \cite{CBIY00}) that perturbations of the CMC case
always lead to solutions of the conformal constraint equations. This means that
the far-from-CMC result is simply the rescaling of a previously known near-CMC result.
This means, unfortunately, that virtually nothing is known about the far-from-CMC case
in the general sense, i.e., with $\sigma, r$ and $J$ arbitrary.
We do note that most existing near-CMC results require
that $d\tau$ be sufficiently small as compared to $\inf \tau^2$, or similar, while the
near-CMC condition for the perturbation of $\tau \equiv 0$ is that the
 $W^{1,p}_{\delta-1}$ norm is sufficiently small. In particular,
$\tau$ can have zeroes.

It is also interesting to consider the range of allowed asymptotic functions. The 
far-from-CMC result allows only very small asymptotic functions, while this rescaled
near-CMC result allows arbitrarily large asymptotic functions, as long as $\tau$ is 
sufficiently small.
Since all asymptotic functions are allowed if the seed data $(\tau, \sigma, r, J)$
vanishes (since we are then just solving $\Delta \phi =0$), this leads one to wonder
whether large asymptotic functions are only ever attainable if $\tau$ is small.

%%%%%%%%%%%%%%%%%%%%%%%%%%%%%%%%%%%%%%%%%%%%%%%%%%%%%%%%%%%%%%%%%%%%%%%%%%%%%%%%%%%%%%%%%
\section{Local Supersolutions}\label{sec:LocalSup}
%%%%%%%%%%%%%%%%%%%%%%%%%%%%%%%%%%%%%%%%%%%%%%%%%%%%%%%%%%%%%%%%%%%%%%%%%%%%%%%%%%%%%%%%%

The search for global supersolutions for more general cases than those considered in the
previous section has been very difficult, partly since any such construction seems to
require an estimate on $LW$ like \eqref{eq:DWestim}, which in turn seems to allow only
either $\tau$ to be small or the rest of the data to be small. However, requiring a
global supersolution in order for solutions of the conformal constraints to exist is
likely stricter than necessary. In order to give new tools for finding solutions to
the conformal constraint equations, Nguyen introduced the idea of ``local
supersolutions."

\begin{defn} \label{def:LocalSup}
A function $\phi_+ \in L^\infty$ is a ``local supersolution" of the conformal constraint
equations if for every positive function $\phi \leq \phi_+$ with $\phi= \phi_+$ 
somewhere, there exists $p \in M$ such that $F(\phi)(p) \leq \phi(p)$. (For the 
definition of $F$, see Section \ref{sec:FixedPoint}.)
\end{defn}

Local supersolutions are more general than global supersolutions.

\begin{prop}
A global supersolution is a local supersolution.
\end{prop}
\begin{proof}
Let $\phi_+$ be a global supersolution, and suppose $\phi$ is a positive function such
that $\phi\leq \phi_+$ and $\phi(p) = \phi_+(p)$ at $p\in M$. Then, by the definition
of global supersolution, $F(\phi)(p)\leq \phi_+(p) = \phi(p)$, and thus $\phi_+$ is a
local supersolution. 
\end{proof}

Nguyen originally used the idea of half-continuity along with a fixed point theorem 
allowing for half-continuous maps to show that the existence of a local supersolution
leads to a solution of the conformal constraint equations. The following result is 
Theorem 4.12 in \cite{Nguyen14}.

\begin{thm} \label{thm:NguyenExistence}
Suppose the seed data $(g,\sigma,1/2,\tau,0,0)$, $M$ compact, $p>n$, 
are such that
the zero set of $\tau$ has zero measure and $g$ allows no conformal Killing fields.
If $g$ is Yamabe nonnegative, assume $\sigma \not\equiv 0$. If $g$ is Yamabe negative,
assume there is a conformal factor changing the metric to one with scalar curvature 
$-\kappa \tau^2$.
Let $F:L^\infty \to L^\infty$ be the iteration map for the constraint equations,
as described in Section \ref{sec:FixedPoint}. If there exists a local supersolution
$\psi$, 
then there exists a fixed point $\phi$ for $F$ with $\phi \leq \psi$. In particular,
the conformal constraint equations have a solution.
\end{thm}

The seeming advantage of this result over using a global supersolution is that
one needs only to check that the solution is smaller at a single point, rather
than on the entire manifold. We present a simpler proof of a slightly stronger
result. The new proof is based on the Leray-Schauder alternative.

\begin{thm}[Leray-Schauder Alternative]\label{thm:LeraySchauder}
Let $F:X\to X$ be a compact, continuous map of a normed linear space. Let $\Omega$ be
a bounded star-shaped open subset of $X$ containing 0, and suppose that $x\in \p \Omega$
implies that $F(x) \neq \lambda x$ for any $\lambda>1$.  Then $F$ has a fixed point on
$\overline{\Omega}$. 
\end{thm}

\begin{thm}\label{thm:MyNguyen}
Let $(g,\tau, N, \sigma, r, J)$ be seed data. Suppose $g$ has no conformal Killing 
fields. If $M$ is AE, suppose the data is 
admissible. If $M$ compact and $g$ is Yamabe nonnegative, assume 
$\sigma \not\equiv 0$. If $M$ is compact and $g$ is Yamabe negative,
assume there is a conformal factor changing the metric to one with scalar curvature 
$-\kappa \tau^2$. Suppose there exists $\psi\in L^\infty$ such that for any 
$0<\phi\leq \psi$, with $\inf |\phi-\psi| = 0$, that $F(\phi) \neq \lambda\phi$
for all $\lambda>1$. 
Then there exists a fixed point $\phi$ for $F$ with $\phi \leq \psi$. In particular,
the conformal constraint equations have a solution.
\end{thm}
\begin{proof}
Define $F'(\phi) = F(|\phi|)$. 
The map $F:L^\infty \to L^\infty$ is a compact continuous map by the proof of Theorem
\ref{thm:CombinedExistence}, and so $F'$ is as well. Thus
we only need to find an appropriate set $\Omega$, as in Theorem \ref{thm:LeraySchauder}.

Let $\Omega= \{\phi \in L^\infty: -\psi< \phi< \psi\}$. Clearly $\Omega$ is bounded, open,
star-shaped,
and contains zero. Suppose there were some $\phi \in \partial\Omega$ and $\lambda>1$ such
that $F'(\phi) = \lambda\phi$. 
In particular, this means that $\inf|\phi-\psi| = 0$.
Since $F'$ outputs only positive solutions, $\phi>0$.
 By assumption, there are no
such $\phi$. Thus $\Omega$ fulfills the conditions of Theorem \ref{thm:LeraySchauder},
and so $F'$ has a fixed point $\phi$ with $\phi>0$. Since $\phi>0$, $F'(\phi) = F(\phi)$,
and so $F$ has a fixed point as well.
\end{proof}

Note that for $M$ compact, $\inf |\phi-\psi|=0$ means $\phi=\psi$ somewhere. For 
AE manifolds, though, $\phi$ may not equal $\psi$ anywhere.

Other than applying to AE manifolds, Theorem \ref{thm:MyNguyen} has a few advantages
over Theorem \ref{thm:NguyenExistence}. The main advantage is that one only needs to
show that $F(\phi) \neq \lambda\phi$ instead of showing that $F(\phi)(p) \leq \phi(p)$
for appropriate $\phi$, a slightly more general condition. This allows one to assume
that the solution for the Lichnerowicz equation is a multiple of the function one began
with, rather than an arbitrary solution.
Unfortunately, no local supersolution that is not also a global supersolution has yet 
been found.

\chapter{The Limit Equation Criterion} \label{chap:LimitEquation}
%%%%%%%%%%%%%%%%%%%%%%%%%%%%%%%%%%%%%%%%%%%%%%%%%%%%%%%%%%%%%%%%%%%%%%%%%%%%%%%%%%%%%%%%% 

Another method of finding solutions to the conformal constraint equations
\eqref{eq:ConfConst} is the limit equation criterion, originally
introduced in \cite{DGH11}. This result says that if a particular equation, called the
limit equation, has \emph{no} solutions, then the conformal constraint equations
have a solution. To be precise, the main result of \cite{DGH11} says the following:

\begin{thm}\label{thm:OrigLimit}
Suppose the seed data $(g,\sigma,1/2,\tau,0,0)$, $M$ compact, $p>n$, 
are such that $\tau>0$ and $g$ allows no conformal Killing fields.
If $g$ is Yamabe nonnegative, assume $\sigma \not\equiv 0$. If $g$ is Yamabe negative,
assume there is a conformal factor changing the metric to one with scalar curvature 
$-\kappa \tau^2$.
Then at least one of he following holds:
\begin{itemize}
\item The conformal constraint equations \eqref{eq:ConfConst}
admit a solution $(\phi,W)$ with $\phi>0$. Furthermore, the set of solutions
$(\phi,W) \in W^{2,p}\times W^{2,p}$ is compact.

\item There exists a nontrivial solution $W\in W^{2,p}$ of the limit equation
   \begin{equation}\label{eq:DGHLimit}
     \di LW = \alpha_0 \sqrt{1/\kappa} |LW| \frac{d\tau}{\tau} 
   \end{equation} for some $\alpha_0 \in (0,1]$.
\end{itemize}
\end{thm}

The name ``limit equation" comes from the original method of proof. In \cite{DGH11},
they first show that a ``subcritical" version of the conformal constraint equations
always have a solution. They make the equations subcritical by changing the exponent
of $\phi$ in the vector equation \eqref{eq:OrigVect} from $q$ to $q-\epsilon$. This
allows global supersolutions to be easily found. Dahl, Gicquaud, and Humbert then
show that if a sequence of these solutions with $\epsilon\to 0$ are bounded, then
they must converge to a solution of the original equation. If they are unbounded,
then they must converge to a solution of the limit equation \eqref{eq:DGHLimit}.
Since then, another, simpler method has been found for setting up the sequence 
(cf. \cite{Nguyen14}), though the argument for convergence is essentially the same.

The two conclusions of Theorem \ref{thm:OrigLimit} are not a dichotomy. Nguyen in
\cite{Nguyen14} showed that there is seed data on the sphere that allows solutions to
both the conformal constraint equations and the limit equation. Thus, unfortunately,
the use of the limit equation is limited to the case mentioned earlier. If one can
show, for particular seed data, that the limit equation has no solutions, then the
conformal constraint equations must have a solution.

Though the limit equation criterion itself is not limited to this case, the criterion
has only been successfully applied for near-CMC seed data. In particular, in 
\cite{DGH11}, they find that the limit equation \eqref{eq:DGHLimit} has no nontrivial
solutions if either the $C^0$ or $L^n$ norm of $d\tau/\tau$ is sufficiently small.
The reason it is difficult to prove stronger results is that the usual method of
proving nonexistence is to find an estimate for the right hand side of 
\eqref{eq:DGHLimit} that provides a contradiction. In order to do this, seemingly the
only tool that one has is to make $d\tau/\tau$ small, since $LW$ appears on both sides
of the limit equation.

In this chapter, we prove most of the limit equation criterion in the AE setting,
 except for the vital 
condition that the solution $W$ of the limit equation must be nontrivial. 
The proof breaks down only at this point. It can be repaired if 
the data is near-CMC. Since the compact case has only been applied to near-CMC data,
this seems like a reasonable assumption. However, since we must also require that 
$\tau \to 0$ at infinity and that $\tau>0$, arbitrarily 
near-CMC data does not exist, and thus the near-CMC condition is very difficult to check.
This chapter is based on an unpublished collaboration with 
Jim Isenberg and Romain Gicquaud.

%%%%%%%%%%%%%%%%%%%%%%%%%%%%%%%%%%%%%%%%%%%%%%%%%%%%%%%%%%%%%%%%%%%%%%%%%%%%%%%%%%%%%%%%%
\section{Setup on Asymptotically Euclidean Manifolds} \label{sec:LimitIntro}
%%%%%%%%%%%%%%%%%%%%%%%%%%%%%%%%%%%%%%%%%%%%%%%%%%%%%%%%%%%%%%%%%%%%%%%%%%%%%%%%%%%%%%%%%

The main difficulty in translating the limit equation criterion to asymptotically 
Euclidean manifolds is the fall off rate of $\tau$. In the original proof, the assumption that
$\tau>\epsilon>0$ is vital; this is most easily seen by noticing that we divide by 
$\tau$ in the limit equation \eqref{eq:DGHLimit}. However, for the seed data to lead to
asymptotically Euclidean initial data, $\tau$ must converge to zero. These competing 
conditions lead to most of the additional difficulty in the AE case.

The proof is simpler when using H\"older norms for essentially the same reasons as 
discussed in Remark \ref{rmk:HolderVsSobolev}. Put together, our assumptions
are as follows.
\begin{assumpt} [Seed Data Assumption]
\label{AssumptLimit}
The seed data $(g,\tau, N, \sigma, r, J)$  satisfy the conditions
\begin{itemize}
\item $(M,g)$ is a $C^{2,\alpha}_{\delta}$ AE manifold, with  $\delta\in (2-n,0)$,
      allowing no conformal Killing fields.
\item $\tau \in C^{1,\alpha}_{\delta-1}$, and $|\tau| \geq C \rho^{\delta-1}>0$.
       Without loss of generality, assume that $\tau>0$.
\item $N-1 \in C^{2,\alpha}_{\delta}$.
\item $\sigma \in C^{0,\alpha}_{\delta-1}$. Thus $|\sigma|\leq C \tau$.
\item $0\leq r$, and $r\in C^{0,\alpha}_{\delta-2}$. Thus $r \leq C\tau$.
\item $J \in C^{0,\alpha}_{\delta-2}$.
\end{itemize}
\end{assumpt}

We can now state the main result of this chapter. 
\begin{thm}\label{thm:AELimit}
Suppose the seed data satisfies Assumption \ref{AssumptLimit}. 
Then at least one of the following holds:
\begin{itemize}
\item For any asymptotic function $\mru$, the conformal constraint equations
 \eqref{eq:ConfConst}
 admit a solution $(\phi,W)$ with $\phi>0$ and $\phi-\mru$ and $W$ in 
 $C^{2,\alpha}_{\delta}$.

\item There exists a (perhaps trivial) solution $W\in W^{2,p}_\delta$ of the limit
 equation
   \begin{equation}\label{eq:AELimit}
     \di \frac{1}{2N} LW = \alpha_0 \sqrt{1/\kappa} |LW| \frac{d\tau}{2N\tau} 
   \end{equation} for some $\alpha_0 \in [0,1]$. Furthermore, 
   $|W|\leq C \rho^{\delta}$ and $|LW|\leq C \rho^{\delta-1}$ for some $C$ dependent
   only on $g$ and $\|d\tau\|_{C^{0}_{\delta-2}}$. 
   If 
   \begin{equation}\label{eq:NearCMCCondition}
   \kappa \tau^2 - \frac{1}{4N^2} |LW_0|^2 \geq c \tau^2
   \end{equation} for some $c>0$ and for all solutions $W_0$ of the vector equation
    \eqref{eq:OrigVect} with $J \equiv 0$ and $\phi\leq 1$, then the solution
    $W$ of the limit equation is nontrivial, and $\alpha_0\neq 0$.
\end{itemize}
\end{thm}

The near-CMC condition \eqref{eq:NearCMCCondition} is phrased oddly because it is
easier to work with later. Similar to estimate \eqref{eq:LWest},
\begin{align}\notag
\|LW_0\|_{C^{0}_{\delta-1}} &\leq C \| \phi^{q} |d\tau| + J\|_{C^0_{\delta-2}}\\
   &\leq C\|d\tau\|_{C^0_{\delta-2}}.\label{eq:W0Bound}
\end{align}
Note that the constant $C$ only depends on $g$. Thus a sufficient condition for 
\eqref{eq:NearCMCCondition} to hold is that
there exists a $C>0$ depending only on $g$ and $N$, such
that $\tau^2 - C \|d\tau\|^2_{C^0_{\delta-2}} \rho^{2\delta-2} \geq c\tau^2$. This is
more clearly a near-CMC condition.

The simpler proof of the limit equation criterion presented in \cite{Nguyen14}
is based on the Schaefer fixed point theorem.

\begin{thm}[Schaefer Fixed Point Theorem]\label{thm:Schaefer}
 Assume that $F:X\to X$ is a compact map on a Banach space, and assume that the set
 \begin{equation}
   K= \{x\in X: \exists t \in [0,1] \textrm{ such that } x= t F(x)\}
 \end{equation} is bounded. Then $F$ has a fixed point.
\end{thm}

Let $F$ be the iteration map, as in Chapter \ref{chap:FarFromCMC}, giving a solution
of the Lichnerowicz equation with asymptotic function $\mru$. Recall that $F$ is
a compact map on $L^\infty$. We use the Schaefer fixed point theorem by setting 
\begin{equation}\label{eq:KDef}
K:= \{ \phi\in L^\infty:  \exists t \in [0,1] \textrm{ such that } \phi= t F(\phi)\}.
\end{equation} Note that the definition of $K$ does not directly mention the 
asymptotic function of $\phi$. However, since $\phi = t F(\phi)$, we know that 
$\phi - t\mru \in W^{2,p}_\delta$, for some $t\in (0,1]$.

%%%%%%%%%%%%%%%%%%%%%%%%%%%%%%%%%%%%%%%%%%%%%%%%%%%%%%%%%%%%%%%%%%%%%%%%%%%%%%%%
\section{Convergence of Solutions}
\label{sec:Conv}
%%%%%%%%%%%%%%%%%%%%%%%%%%%%%%%%%%%%%%%%%%%%%%%%%%%%%%%%%%%%%%%%%%%%%%%%%%%%%%%%

In this section, we show that if $K$ is unbounded, then the limit equation has
a solution. By conformal covariance, and since $\tau^2>0$, it follows from the
prescribed scalar curvature theorem \ref{thm:PrescribedScalarCurvature} 
that we can assume without loss
of generality that $R= -\kappa \tau^2$. The definition of $K$ shows that $K$ being
unbounded is equivalent to the
existence of an unbounded sequence $(\phi_i,t_i)$ with $t_i\in(0,1]$ solving
\begin{subequations}\label{eq:SeqConst}\begin{align}
-a \Delta \phi_i -\kappa \tau^2 \phi_i + \kappa \tau^2 \phi_i^{q-1} 
    - \left|\sigma + \frac{1}{2N} LW_i\right|^2 \phi_i^{-q-1} - r \phi_i^{-q/2}=0
     \label{eq:SeqLich}\\
\di \frac{1}{2N} LW = \kappa (t_i \phi_i)^{q} d\tau + J. \label{eq:SeqVect}
\end{align} \end{subequations}

\begin{lem}\label{lem:LimitConvergence}
Suppose such a sequence $(\phi_i, t_i)$ exists. Then there exists a (perhaps trivial)
solution $W\in W^{2,p}_\delta$ of the limit equation
   \begin{equation}\label{eq:AELimit}
     \di \frac{1}{2N} LW = \alpha_0 \sqrt{1/\kappa} |LW| \frac{d\tau}{2N\tau} 
   \end{equation} for some $\alpha_0 \in [0,1]$. Furthermore, 
   $|W|\leq C \rho^{\delta}$ and $|LW|\leq C \rho^{\delta-1}$ for some $C$ dependent
   only on $g$ and $\|d\tau\|_{C^{0}_{\delta-2}}$. 
   If 
   \begin{equation}\label{eq:NearCMCCondition2}
   \kappa \tau^2 - \frac{1}{4N^2} |LW_0|^2 \geq c \tau^2
   \end{equation} for some $c>0$ and for all solutions $W_0$ of the vector equation
    \eqref{eq:OrigVect} with $J \equiv 0$ and $\phi\leq 1$, then the solution
    $W$ of the limit equation is nontrivial, and $\alpha_0\neq 0$.
\end{lem}
\begin{proof}
We claim the sequence $W_i:= W(\phi_i)$ 
(sub)converges to a solution of the limit equation, up to rescaling. In the original
proof of the limit equation criterion, the sequence of subcritical solutions was 
renormalized by an energy based on $|LW_i|:= |LW(t_i\phi_i)|$. While this energy has
some advantages which
we discuss below, its use requires proving an estimate of the form
$\sup \phi_i^q \leq \sup\{1, \int |LW|^2\}$. Despite some effort, we were unable to
prove an analogous estimate for AE data. The problem is that $|LW_i|$
falls off as $\tau$, but the convergence of a renormalized $|LW_i|$ is only in a 
weaker space. It is essentially for this same reason that we are only able to prove
nontriviality of the solution of the limit equation for near-CMC data.

Instead of an energy based on $|LW_i|$, we bound $\sup \phi_i$ directly. By
assumption $\sup \phi_i \to \infty$. Let $\Gamma_i = \sup \phi_i$.
We start by rescaling the seed data and solutions by certain powers of the energy. In
particular, we rescale $\sigma, r, J, \phi_i$, and $W_i$ (we do not rescale the metric,
$N$, or $\tau$) as follows
\begin{equation}\label{eq:rescaling}
      \phitil_i   := \Gamma_i^{-1} \phi_i,
\quad \Wtil_i    := \Gamma_i^{-q} W_i,
\quad \sigmatil_i := \Gamma_i^{-q} \sigma,
\quad \rtil_i    := \Gamma_i^{\frac{-3q}{2}+1}r,
\quad \Jtil_i     := \Gamma_i^{-q} J.
\end{equation}
Then, dividing the conformal constraint equations \eqref{eq:SeqConst}
by certain powers of the energy, and substituting in these rescaled quantities, we obtain
\begin{subequations}\label{eq:LimitConst}
\begin{align}\label{eq:LimitLich}
\frac{1}{\Gamma_i^{q-2}} \left(-a\Delta \phitil_i -\kappa \tau^2 \phitil_i\right)
  + \kappa\tau^2 \phitil_i^{q-1}
    -\left|\sigmatil_i+\frac{1}{2N}L\Wtil_i\right|^2 \phitil_i^{-q-1}
     - \rtil_i \phitil_i^{-q/2}=0,\\
\label{eq:LimitVect}
\di \frac{1}{2N}L\Wtil_i
   = \kappa(t_i \phitil_i)^{q} d\tau + \Jtil_i.
\end{align}
\end{subequations}

Proceeding, we substantially follow the original proof from \cite{DGH11}. 
Similar to the estimate \eqref{eq:W0Bound},
\begin{align}\notag
\|W_i\|_{C^{1,\beta}_{\delta}} &\leq C \| t_i^q\phi_i^{q} |d\tau| + J\|_{C^0_{\delta-2}}\\
   &\leq C \Gamma^{q}_i \left( \|d\tau\|_{C^0_{\delta-2}} 
      + \|J/\Gamma_i^q\|_{C^0_{\delta -2}} \right).\label{NewLWptwise}
\end{align}
Note that the constant $C$ only depends on $g$, $N$, and our choice of $\beta>\alpha$.
Consequently, $\Wtil_i$ is uniformly bounded in $C^{1,\beta}_{\delta}$ and $L\Wtil_i$ is
uniformly bounded in $C^{0,\beta}_{\delta-1}$. Using the compact Sobolev embedding 
\ref{prop:HolderEmbeddings}, a subsequence of $\Wtil_i$ converges
 in $C^{1,\alpha}_{\delta'}$ for any $\delta'>\delta$ to some
$\Wtil_\infty \in C^{1,\beta}_{\delta}$. Thus, in particular, 
\begin{equation}
|\Wtil_\infty| \leq C \rho^{\delta} \quad\text{and}\quad |L\Wtil_\infty| \leq C \rho^{\delta-1},
\end{equation} for $C$ dependent only on $g$, $N$, and $\beta$. We cannot, however, 
be certain that $\Wtil_\infty\not\equiv 0$.

Heuristically, as $i\to \infty$ we would expect all the terms in the rescaled
Lichnerowicz equation \eqref{eq:LimitLich} except the $\tau^2$ and $L\Wtil_i$ terms
to go to zero. Thus, we define the function $\phitil_\infty$ by
\begin{equation}
\kappa\tau^2 \phitil_\infty^{q-1} := \frac{1}{4N^2}|L \Wtil_\infty|^2 \phitil_\infty^{-q-1},
\end{equation} which reduces to
\begin{equation}
\label{eq:phitilinf}
\phitil^q_\infty = \frac{|L\Wtil_\infty|}{2N\sqrt{\kappa}\tau}.
\end{equation}
Comparing expression \eqref{eq:phitilinf} with the rescaled vector constraint equation
\eqref{eq:LimitVect}, we see that if $\phitil_i \to \phitil_{\infty}$
in an appropriate space, and if $\Wtil_{\infty}$ is a solution
(in an appropriate sense) to the limit of equation \eqref{eq:LimitVect} as 
$i \to \infty$, then it follows that $\Wtil_\infty $ is a solution to the limit equation
\eqref{eq:AELimit}. So long as $\phitil_\infty$ is not
identically zero, this solution is non-trivial. Therefore, we focus on verifying these
limits.

For any $\epsilon >0$, we claim there is an $i_0$ such that if 
$i\geq i_0$ that
\begin{equation}
\label{philimit}
|\phitil_i - \phitil_\infty| < \epsilon \rho^\epsilon.
\end{equation}
If \eqref{philimit} holds (for small $\epsilon$), it then follows that
$\phitil_i \to \phitil_\infty$ in $C^0_{\epsilon'}$ for any small $\epsilon'>0$.
Recalling the definition of $\phitil_\infty$ from \eqref{eq:phitilinf}, we let 
$\phitil_+$ be any $C^2$ function for which the inequality
\begin{equation}
\phitil_\infty+\frac{\epsilon}{2}\rho^\epsilon \leq \phitil_+ \leq \phitil_\infty+ \epsilon\rho^\epsilon.
\end{equation}
holds. We claim  that $\phitil_+\geq \phitil_i$ everywhere for $i$ large enough. Since
$\phitil_i\leq 1$, it is immediately clear that this is true except on a compact set 
depending only on $\epsilon$. On that compact set, we claim that $\phitil_+$ is a
(pointwise) supersolution of the rescaled Lichnerowicz equation \eqref{eq:LimitLich} for
all $i$ is large enough.

Multiplying the rescaled Lichnerowicz equation \eqref{eq:LimitLich} by 
$\phitil^{q+1}_+$, we need to verify the inequality
\begin{equation}
\label{phitilineq}
\frac{\phitil_+^{q+1}}{\Gamma^{q-2}} \left(-a\Delta \phitil_+
  - \kappa\tau^2 \phitil_+\right) + \kappa\tau^2 \phitil^{2q}_+
  \geq \left|\sigmatil + \frac{1}{2N}L \Wtil_i\right|^2 + \rtil \phitil^{q/2+1}.
\end{equation}
Since, by definition,
\begin{equation}
\phitil_+^{2q}
  \geq \left(\phitil_\infty + \frac{\epsilon}{2}\rho^\epsilon\right)^{2q}
   \geq \phitil^{2q}_\infty + \left(\frac{\epsilon}{2}\rho^\epsilon\right)^{2q},
\end{equation}
inequality \eqref{phitilineq} is satisfied provided that
\begin{multline}
\label{hmm}
\frac{\phitil_+^{q+1}}{\Gamma_i^{q-2}} \left(-a\Delta \phitil_+
  - \kappa\tau^2\phitil_+\right) + \kappa\tau^2 \left(\frac{\epsilon}{2} \rho^\epsilon\right)^{2q}
  \\ \geq \left|\sigmatil_i + \frac{1}{2N}L \Wtil_i\right|^2 
     -\frac{1}{2N}|L\Wtil_\infty|^2+ \rtil_i \phitil^{q/2+1}.
\end{multline}

Recalling that $L\Wtil_i \to L\Wtil_\infty$, we readily verify that all of the terms
in equation \eqref{hmm} go to zero pointwise
as $i \to \infty$ except for the $\epsilon$ term. Thus for any fixed compact set,
there exists an $i_0$ such that for all $i\geq i_0$, $\phitil_+$ is a pointwise
supersolution on that compact set.

We want to use the sub and supersolution theorem to prove that $\phitil_+ \geq \phitil_i$
for large $i$. Since $\phitil_+$ is only a supersolution on a compact set, we cannot use
the sub and supersolution theorem on AE manifolds \ref{thm:SubSupersolutionTheorem}.
However, in the complement of the compact set $\{\frac\epsilon2\rho^\epsilon\leq 1\}$,
we know 
$\phitil_+\geq \phitil_i$. Thus we can use the sub and supersolution theorem on compact 
manifolds with boundary (cf. \cite{Dilts13b}) to find a solution of \eqref{eq:LimitLich}
less than $\phitil_+$ on that compact set. (For convenience, we can use 
$\alpha \Gamma^{-1}_i$  for the subsolution. It is a global subsolution for any
$\alpha\leq 1$, independent of  $\epsilon$.) Since such solutions are unique, it must
be $\phitil_i$. Thus we obtain the pointwise inequality
\begin{equation}
\label{ptwiseineq}
\phitil_i \leq \phitil_+ \leq \phitil_\infty + \epsilon \rho^\epsilon.
\end{equation}
everywhere on $M$.

We prove a similar (subsolution type) result for any $C^2$ function $\phitil_-$ satisfying
\begin{equation}
\phitil_\infty - \epsilon\rho^\epsilon \leq \phitil_- 
    \leq \phitil_\infty - \frac{\epsilon}{2} \rho^\epsilon.
\end{equation}
Since $\phitil_i >0$, we only need to show that $\phitil_-$ is a subsolution of the 
Lichnerowicz equation where $\phitil_-$ is positive. The proof follows similarly, 
using $\phitil_+$ as our supersolution. We thus have
\begin{equation}\label{ptwiseineq2}
\phitil_i \geq \phitil_\infty-\epsilon\rho^\epsilon
\end{equation} for large enough $i$.

Using the two inequalities \eqref{ptwiseineq} and \eqref{ptwiseineq2}, we conclude
that $\phitil_i \to \phitil_\infty$ in $C^0_{\epsilon'}$ for any $\epsilon'>0$. 
This implies that $\phitil_i^q d\tau \to \phitil_\infty^q d\tau$ in 
$L^p_{\delta-2+\epsilon' q}$. Also, a subsequence of $t_i^q$ converges to a number
$\alpha_0 \in [0,1]$. Thus $\Wtil_i \to \Wtil_\infty$
in $W^{2,p}_{\delta+\epsilon'q}$ for $\epsilon'<-\delta/q$. The $\Wtil_\infty$
weakly satisfies the limit equation \eqref{eq:AELimit}. Since the right hand side 
of the limit equation is in $C^{0,\alpha}_{\delta-2}$ (since 
$L\Wtil_\infty \in C^{0,\beta}_{\delta-1}$), we can conclude by Proposition 
\ref{prop:AELinearExistence} that $\Wtil_\infty \in C^{2,\alpha}_{\delta}$.

Our only remaining task is to verify that $W_\infty \not\equiv 0$ in the near-CMC
case. In this case, if $\alpha_0=0$, any solution of the limit equation is a 
conformal Killing field, which implies that $W_\infty \equiv 0$, a contradiction (cf. 
Proposition \ref{prop:Isomorphism}). 
To show that $W_\infty$ is not trivial, we repeat the previous argument with a few changes. 
This implies that $\Wtil_\infty \not\equiv 0$. 

Assume that $L\Wtil_\infty \equiv 0$, which by definition implies that 
$\phitil_\infty \equiv 0$. Let $\phitil_+$ be a constant less than 1. 
We can derive the equivalent of equation \eqref{hmm} for this $\phitil_+$, namely
\begin{equation}
\label{hmm2}
-\frac{\phitil_+^{q+2}}{\Gamma_i^{q-2}} \kappa \tau^2\phitil_+ 
   + \kappa\tau^2 \phitil_+^{2q}
  \geq \left|\sigmatil_i + \frac{1}{2N}L \Wtil_i\right|^2 + \rtil_i \phitil^{q/2+1}.
\end{equation}
We note that as $i\to \infty$, the $J$ term in equation \eqref{NewLWptwise} goes to zero.
Thus, by our near-CMC assumption,
\begin{equation}\label{eq:Near-CMC_Assumption}
\kappa\tau^2 \phitil_+^{2q} - \left|\sigmatil_i+\frac{1}{2N}L\Wtil_i\right|^2\geq c \tau^2
\end{equation} for large enough $i$, $\phitil_+$ sufficiently close to $1$, and some 
small $c>0$. 
The other terms in \eqref{hmm2} go to zero in $C^{0}_{2\delta-2}$, and so for $i$ large
enough, \eqref{hmm2} holds on all of $M$.

Using the sub and supersolution theorem on AE manifolds, we can deduce
\begin{equation}
\phitil_i \leq \phitil_+<1.
\end{equation} Since $\phitil_i=1$ somewhere, this is a contradiction.
\end{proof}

The main result is now easily proved.

\begin{proof}[Proof of Theorem \ref{thm:AELimit}]
Let $F$ be the iteration map, as in Chapter \ref{chap:FarFromCMC}, giving a solution
of the Lichnerowicz equation with asymptotic function $\mru$. Let 
\begin{equation}\label{eq:KDef}
K:= \{ \phi\in L^\infty:  \exists t \in [0,1] \textrm{ such that } \phi= t F(\phi)\}
\end{equation} 

If $K$ is bounded, then the Schaefer fixed point theorem \ref{thm:Schaefer} gives a 
solution to the conformal constraint equations. If $K$ is unbounded, Lemma
\ref{lem:LimitConvergence} shows the limit equation has a solution.
\end{proof}

We now make a number of remarks on the proof of Lemma \ref{lem:LimitConvergence}, which
is the heart of the limit equation criterion. On compact manifolds, it is easy to show 
that $\phitil_\infty$ is nontrivial. By definition, $\phitil_i =1$ at some point $p_i$, and since
$M$ is compact, the $p_i$ converge to some point $p_\infty$. Since the $\phitil_i$
converge in $L^\infty$, we have $\phitil_\infty(p_\infty) = 1$, and so $\phitil_\infty$
and $\Wtil_\infty$ are nontrivial. On AE manifolds, this argument breaks down at two
points. First, since $M$ is noncompact, the points $p_i$ may wander off to infinity.
Indeed, if they were contained on a compact set, then the $\phitil_i$ would converge
to a nontrivial $\phitil_\infty$ in a similar fashion. Second, we are only able to show
that $\phitil_i$ converges to $\phitil_\infty$ in $C^0_{\epsilon'}$ for $\epsilon'>0$,
and thus $\Wtil_i$ only converges to $\Wtil_\infty$ in 
$C^{1,\beta}_{\delta + \epsilon'q}$. In particular, if $\Wtil_\infty$ were trivial,
$|L\Wtil_i|^2$ converges to zero only in $C^0_{2\delta-2+2\epsilon'q}$, and so we can
only show the inequality \eqref{hmm2} on compact sets, unless we make the near-CMC
assumption that we did. Fixing either of these two points would show that the
solution to the limit equation is nontrivial.

Another idea, as discussed above, is to use some relative of $\|LW_i\|_2$ as the energy.
If $L\Wtil_i$ converges in that norm, then $\|L\Wtil_\infty\|_2$
would be 1, and thus nontrivial. This type of argument works in the compact case, and is
in fact what Dahl, Gicquaud, and Humbert originally did in \cite{DGH11}. Unfortunately,
proving convergence of the $\phitil_i$ requires an estimate related to 
$\sup \phi_i^q \leq C\max\{\|LW_i\|_2, 1\}$. We have tried to prove such an estimate for
AE manifolds, but were unable to do so. The problem is, again, that $|LW_i|$ and $\tau$ have
the same falloff at infinity.

One possible reason why the proof of the limit equation criterion for AE manifolds has
proven difficult is its relationship with asymptotic functions. The limit equation
criterion \ref{thm:AELimit}, as long as the limit equation has no solutions, gives a
solution for \emph{any} asymptotic constant. In chapter \ref{chap:FarFromCMC}, however,
except in the near-CMC case, our proof strongly depends on the asymptotic function 
being sufficiently small. Using the rescaling of Proposition \ref{prop:NguyenScaling},
the asymptotic function is allowed to be larger, as long as $\tau$ scaled towards zero,
a kind of near-CMC condition. Thus, perhaps, the limit equation criterion for more 
arbitrary $\tau$ may be easier to prove as long as the asymptotic
function is sufficiently small. Unfortunately, we were also unable to leverage this idea.

We can use a modification of the argument in Lemma \ref{lem:LimitConvergence} to
show that if $\tau$ is sufficiently near-CMC, then the conformal constraint 
equations have a solution.

\begin{prop} \label{prop:NearCMCSoln}
If 
\begin{equation}\label{eq:NearCMCConditionFail}
\kappa \tau^2\left(\frac12\right)^{2q} - \frac{1}{4N^2} |LW_0|^2 \geq c \tau^2
\end{equation} for some $c>0$ and all solutions $W_0$ of the vector equation
\eqref{eq:OrigVect} with $J \equiv 0$ and $\phi\leq 1$, then the conformal constraint 
equations have a solution as in Theorem \ref{thm:AELimit}.
\end{prop}
\begin{proof}
Suppose $\tau$ satisfies \eqref{eq:NearCMCConditionFail}. 
Let $K$ be as in \eqref{eq:KDef}. Suppose $K$ is unbounded.
Then we can proceed as in the proof of Lemma \ref{lem:LimitConvergence}. 

Let $\Wtil_\infty$ and $\phitil_\infty$ be constructed as in the proof of Lemma
\ref{lem:LimitConvergence}. Then, using, $\phitil_\infty$ to define $W_0$,
condition \eqref{eq:NearCMCConditionFail} reduces to
\begin{equation}
\phitil_\infty^{2q} := \frac{|L\Wtil_\infty|^2}{4 N^2 \kappa\tau^2} 
             \leq \left(\frac12\right)^{2q}-\epsilon
\end{equation} for some small $\epsilon>0$. In particular, $\phitil_\infty$ is bounded
above by a number strictly less than $1/2$. 

Let $\phitil_+$ be a constant function such that
$\phitil_\infty + \frac12 < \phitil_+ < 1 \leq \phitil_\infty +1$.

Using the near-CMC condition \eqref{eq:NearCMCConditionFail}, $\phitil_+\geq \phitil_i$
everywhere, as is shown in the proof of Lemma \ref{lem:LimitConvergence}. But $\phitil_i =1$ 
somewhere and $\phitil_+ < 1$. This is a contradiction, so for such $\tau$, $K$ must
be bounded. Thus the Schaefer fixed point theorem \ref{thm:Schaefer} gives a solution
to the conformal constraint equations.
\end{proof}

As discussed in the introduction of this chapter, since the limit equation criterion has
only been applied in the near-CMC case, it might seem reasonable
to assume that $\kappa \tau^2 - \frac{1}{4N^2} |LW_0|^2 \geq c \tau^2$. Unfortunately, 
we can show that $\tau$ cannot be arbitrarily near-CMC, in the sense that
there are constants $C>0$ sufficiently large such that no $\tau>\rho^{\delta-1}$
satisfies the related
inequality $\tau^2 - C \|d\tau\|^2_{C^0_{\delta-2}} \rho^{2\delta-2} \geq c\tau^2$. 
This makes it very difficult to verify that the near-CMC condition in the limit 
equation criterion \ref{thm:AELimit} or Proposition \ref{prop:NearCMCSoln}
is fulfilled.

\begin{prop}\label{prop:NoNearCMC}
There is a constant $C>0$ sufficiently large such that for any $C'>0$,
no $\tau$ satisfying $\tau>C' \rho^{\delta-1}$ also satisfies
\begin{equation}
\tau^2 - C \|d\tau\|^2_{C^0_{\delta-2}} \rho^{2\delta-2} \geq 0.
\end{equation}
\end{prop}
\begin{proof}
Suppose to the contrary that there are $\tau_i$ such that, after scaling, 
$\tau_i \geq \rho^{\delta-1}$ and $\|d\tau_i\|_{C^0_{\delta-2}}< 1/i$. Using 
the Poincar\'{e} inequality \ref{lem:poincare}, for $p\in (n/2,n)$, we have
\begin{equation}
\|\tau_i\|_{p,{1-n/p}} \leq C \|d\tau_i\|_{p,-n/p} 
\leq C\|d\tau_i\|_{C^0_{\delta-2}} \leq C\frac1i.
\end{equation} This contradicts the fact that $\tau_i\geq \rho^{\delta-1}$.
\end{proof} 

\chapter{ADM Mass and the Asymptotic Function} \label{chap:Mass}
%%%%%%%%%%%%%%%%%%%%%%%%%%%%%%%%%%%%%%%%%%%%%%%%%%%%%%%%%%%%%%%%%%%%%%%%%%%%%%%%%%%%%%%%% 

\'{O} Murchadha
\cite{OM05} showed that, in the compact, CMC case, the volume of the solution to
the conformal constraint equations \eqref{eq:ConfConst}
is monotonically related to the constant curvature $\tau$. In particular, instead of
specifying a constant $\tau$, one could equivalently specify the volume of the
solution manifold.

In the asymptotically Euclidean case, $\tau$ must vanish at infinity, and so there
is only one choice for constant $\tau$, $\tau \equiv 0$. However, there is a new
``constant" that one may specify, the asymptotic function. An AE manifold
does not have (finite) volume, but it seems logical to ask what, if anything, the choice of
asymptotic function controls. In analog to volume for the compact case, the ADM
mass is a natural candidate.

In general relativity, it is very difficult to define the mass of a non-isolated object, such
as a star. However, there are good definitions for the mass of an entire system.
One such definition is the ADM mass, a metric invariant for AE manifolds
first described by Arnowitt, Deser, and Misner in \cite{ADM61}. 
This mass describes the total mass of all matter in the AE manifold, as measured by
the mass's effects on the asymptotics of the metric. The usual definition
for the ADM mass is
\begin{equation}\label{eq:ADMMass}
M_{ADM}(g) := \frac{1}{16\pi} \lim_{t\to\infty} \sum_i 
    \int_{S_t} (g_{ij,i}-g_{ii,j}) \nu^j_e dS_e,
\end{equation} where $S_t$ is the Euclidean coordinate sphere of radius $t$ on an end,
$\nu^j_e$ is the Euclidean unit normal to $S_t$ and $dS_e$ is the Euclidean spherical
volume element. Bartnik \cite{Bartnik86} showed that the ADM mass is independent of 
the choice of Euclidean coordinates. 
As expected, for Euclidean space, the ADM mass is zero. For the 
Schwarzschild family of solutions, the mass is exactly the standard mass parameter $M$. 

If $g$ does not fall off to the Euclidean metric $e$ fast enough, the mass
may not exist. For this reason, we assume that $\delta<1-n/2$ for this chapter.
If $g$ is a $W^{2,p}_{\delta}$ AE metric, then $M_{ADM}(g)$ exists. Also, since the
mass is dependent on only one end of the AE manifold, we ignore the other ends, and
use an asymptotic constant $c$ instead of the more general asymptotic function $\mru$.
This makes no difference in our calculations below.

Suppose $\phi-c:= f\in W^{2,p}_{\delta}$. Let $\gtil = \phi^{q-2} g$, as usual.
In order to calculate the mass, we have to use coordinates such that 
$\gtil \to e$ in those coordinates. If $c\neq 1$, this means we must scale the 
coordinates. Let
$x^i$ be the Euclidean coordinates for $g$. Let $\xbar^i = c^{q/2 -1}x^i$. Then, denoting
$\gtil$ in $\xbar$ coordinates as $\gbar$, we have
\begin{equation}
\gbar_{ij} d\xbar^i d\xbar^j = \gbar_{ij} c^{2-q} dx^i dx^j 
  = \gtil_{ij} dx^i dx^j = \phi^{q-2} g_{ij} dx^i dx^j.
\end{equation} 
Thus, as functions on $M$ (not as tensors), $\gbar_{ij} = (\phi/c)^{q-2} g_{ij}$.

Let $\ebar$ be the Euclidean metric in the $\xbar$ coordinates. Then
$dS_{\ebar} = c^{(q/2-1)(n-1)} dS_e$ and $\frac{\p f}{\p \xbar} 
= \frac{\p f}{\p x} c^{1-q/2}$. Finally, since we are integrating over spheres,
tracing with $\nu_{\ebar}^{\overline{j}}$ picks the radial component.
Since the radial direction is the same for both 
metrics, the scaling of the derivatives/metric takes care of this term.

Combining these facts, we have
\begin{align}
M_{ADM}(\gbar) &= \frac{1}{16\pi} \lim_{t\to\infty} \sum_i \int_{S_t} 
   (\gbar_{ij,\overline{i}} - \gbar_{ii,\overline{j}})\nu_{\ebar}^{\overline{j}} dS_{\ebar}\\
&= \frac{c^{(q/2-1)(n-2)}}{16\pi} \lim_{t\to\infty} \sum_i \int_{S_t} 
   \left[ \left( (\phi/c)^{q-2} g_{ij}\right)_{,i} - \left( (\phi/c)^{q-2} g_{ii}\right)_{,j}\right]
   v_e^j dS_e.
\end{align} Using $(q/2-1)(n-2) =2$ and $(\phi/c)^{q-2} = 1+ (q-2)f/c + L.O.T.$, 
and ignoring the lower order terms, which vanish in the limit, we find
\begin{align}
M_{ADM}(\gbar) &= c^2M_{ADM}(g) + \frac{c(q-2)}{16\pi} \lim_{t\to\infty} \sum_i \int_{S_t} 
   \left[  g_{ij}f_{,i} - g_{ii}f_{,j}\right] v_e^j dS_e \\
&= c^2 M_{ADM}(g) + \frac{c(q-2)(1-n)}{16\pi} \lim_{t\to\infty} \int_{S_t} \p_r f dS_e,
\end{align} where $\p_r$ is the Euclidean radial derivative. 

Starting with $W^{2,p}_\delta$ seed data, the Lichnerowicz equation \eqref{eq:OrigLich}
implies that $\Delta \phi \in L^1$. Applying integration by parts,
\begin{equation}
\int_{B_t} \Delta \phi dV_g = \int_{S_t} \nabla_i \phi \nu_g^i dS_g.
\end{equation} If $t\to\infty$,
we can drop lower order terms, and thus
\begin{equation}
\int_M \Delta \phi dV_g = \lim_{t\to\infty} \int_{S_t} \p_r \phi dS_e.
\end{equation}
Thus
\begin{align}\label{eq:PhiMass}
M_{ADM}(\gbar) 
&= c^2 M_{ADM}(g) + \frac{c(q-2)(n-1)}{16\pi a} \int_M -a\Delta\phi dV_g.
\end{align}

\section{Model Problem}

We now discuss a (relatively) simple model problem. Assume we have seed data
with regularity as in Assumption \ref{eq:DataRegularity}, with $p>n$.
Let $\tau \equiv 0$, i.e., the CMC case. In this case, the metric $g$ must be 
Yamabe positive in order for the Lichnerowicz equation \eqref{eq:OrigLich}
to have a solution, and so, 
without loss of generality, we may assume that $R\equiv 0$. For simplicity, we also
assume that $r\equiv 0$ and $J\equiv 0$. Also, since $d\tau\equiv 0$,
we know $LW\equiv 0$. We ignore the degenerate case $\sigma \equiv 0$. 
Thus the conformal constraint equations reduce to a single equation,
\begin{equation}\label{eq:ReducedLich}
-a\Delta \phi - |\sigma|^2 \phi^{-q-1}.
\end{equation}

Let $\phi_c$ be the solution to \eqref{eq:ReducedLich} such that 
$\phi_c-c \in W^{2,p}_\delta$. Such solutions exist for all $c$ by Theorem
\ref{thm:LichIff}. We first prove a lemma.

\begin{lem}\label{lem:PhiUpperBound}
Let $\eta = 2\delta -4+n$.There exists a positive solution $u\in W^{2,p}_{\eta+2}$ of
\begin{equation}\label{eq:uDef}
-a\Delta u - \rho^{\eta}=0
\end{equation} such that $C_0 \rho^{2-n} \leq u \leq C_1 \rho^{\eta+2}$.
\end{lem}
\begin{proof}
Note that $\eta \in (-n,-2)$ since $\delta \in (2-n, 1-n/2)$.
Thus Proposition \ref{prop:AELinearExistence} shows that 
$\Delta:W^{2,p}_{\eta+2} \to L^p_{\eta}$ is an isomorphism, and
so there exists a solution $u\in W^{2,p}_{\eta+2}$ of
\begin{equation}\label{eq:uDef}
-a\Delta u - \rho^{\eta}=0
\end{equation} such that $u \leq C \rho^{\eta+2}$.
The maximum principle \ref{prop:MaxPrinciple} shows that $u$ is positive. 

We claim that $u\geq C\rho^{2-n}$ for some $C>0$. A straightforward extension of 
\cite{Bartnik86} shows that $\Delta:u \mapsto (\Delta u, u|_{\p E_R})$
is an isomorphism between the 
spaces $W^{2,p}_\delta(E_R) \to L^{p}_\delta(E_R) \times W^{2-1/p,p}(\p E_R)$, where $E_R$ is the
exterior of the ball of radius $R$ in the end. Using this result, we can
also extend the sub and supersolution theorem \ref{thm:SubSupersolutionTheorem}
to allow (smooth) internal boundaries with Dirichlet boundary data.

Let $v \in W^{2,p}_\delta$ be the unique solution of 
\begin{equation}\label{eq:PhiDecayPDE}
\Delta v =0 \textrm{ on } E_R \,\,\,\,\, \textrm{ and } \,\,\,\,\, v = 1 \textrm{ on } \p E_R.
\end{equation} 
Clearly, for small enough $\alpha>0$, $\alpha v$ is a subsolution to
\eqref{eq:uDef} on $E_R$. Since $u$ is a solution, and thus a supersolution,
 the extension of the sub and supersolution theorem shows that $u\geq \alpha v$. 

We claim that $v> C\rho^{2-n}$ for some $C>0$, which completes the proof. Consider
the functions $v_\pm = \alpha_\pm(C_\pm \rho^{2-n} \mp \rho^{2-n+\delta/2})$ 
for appropriately
chosen constants $\alpha_+,$, $\alpha_-$, $C_+$, and $C_-$.
We claim that the functions $v_\pm$ are respectively super
and subsolutions of the boundary value problem \eqref{eq:PhiDecayPDE}.
Indeed, using the decomposition
\begin{equation}
\Delta = \Delta_e + k^{ij} \p_i \p_j + g^{ij} \Gamma_{ij}^k \p_k,
\end{equation} and the fact that $W^{2,p}_\delta \subset C^{1}_\delta$, it is clear that 
$\Delta \rho^{2-n}=O(\rho^{-n+\delta})$ since $\Delta_e \rho^{2-n} = 0$.
(Recall that $\rho$ is the radial coordinate sufficiently far out on each end.) 
We similarly get that $\Delta \rho^{2-n+\delta/2}=O(\rho^{-n + \delta/2})$.
This is the highest order term that remains. Then, because of this term's sign, it
eventually dominates, making $v_\pm$ a super or subsolution. Finally, $\alpha_\pm$ can be made
large or small to ensure the boundary condition falls between $v_\pm$. Again using the
extension of the sub and supersolution theorem, along with the fact that $\Delta$ is an
isomorphism, $v_-\leq v \leq v_\pm$, and so $v> C\rho^{2-n}$, as claimed.
\end{proof}

We now list a number of properties of the solutions $\phi_c$ of the reduced Lichnerowicz
equation \eqref{eq:ReducedLich} and their derivatives, 
$\delta \phi_c := \frac{\p}{\p c} \phi_c$.
\begin{prop}[Properties of $\phi_c$]
The function $\phi_c$ satisfies
\begin{equation}\label{eq:PhiCBound}
c< \phi_c \leq c^{-q/2}C_2 u +c,
\end{equation} for some $C_2>0$ independent of $c$. Also, 
\begin{equation}
0<\phi_1 -1\leq \phi_c
\end{equation} for all $c>0$. Finally, 
\begin{equation}\label{eq:DeltaPhiCBound}
\sup\left\{0,\frac{q+1}{c} (c-\phi_c) + 1\right\} \leq \delta\phi_c \leq 1.
\end{equation} In particular, as $c \to \infty$, $\phi_c \to c$ 
and $\delta \phi_c \to 1$.
\end{prop}
\begin{proof}

First, note that $c$ is a subsolution for \eqref{eq:ReducedLich}. We claim that
$c^{-q/2} C_2 u+c$ is a supersolution for \eqref{eq:ReducedLich} for some $C_2>0$
independent of $c$. Note that $c<c^{-q/2}Cu +c$ since $u>0$, and so the sub and 
supersolution theorem \ref{thm:SubSupersolutionTheorem} combined with the uniqueness
of solutions of the Lichnerowicz equation from Theorem \ref{thm:LichUniqueness}
shows that $c \leq \phi_c \leq c^{-q/2} C_2 u+c$. 

For $c^{-q/2} C_2 u +c$ to be a supersolution to the reduced Lichnerowicz equation 
\eqref{eq:ReducedLich}, we must show
\begin{equation}
-a\Delta (c^{-q/2} C u +c) - |\sigma|^2 (c^{-q/2} C u +c)^{-q-1}\geq 0.
\end{equation} Using $|\sigma|^2\leq S \rho^{2\delta-2}$ and 
$C_0 \rho^{2-n} \leq u \leq C_1 \rho^{\eta+2}$, where $\eta = 2\delta -4+n$,
this reduces to
\begin{align}
c{-q/2} C_2 \rho^{\eta} &\geq|\sigma|^2(c^{-q/2}C_2 u + c)^{-q-1} \\
C_2 C_0 c^{-q/2 -1} \rho^{2-n} +1 &\geq S^{1/(q+1)} C_2^{-1/(q+1)} c^{\frac{q}{2(q+1)} -1} \rho^{\frac{2-n}{q+1}}.
\end{align} For some $C_2$ large enough, independent of $c$, this is true, and
so $c^{-q/2} C_2 u+c$ is a supersolution. The
calculation to show this is long and unenlightening, so we do not include it.
Essentially, the $1$ bounds the right hand side if $\rho$ and/or $c$ are large, while
the other term bounds the right hand side if $\rho$ and/or $c$ are small.

Note that $\phi_c-c \in W^{2,p}_\delta$ and that $-a\Delta(\phi_c-c)\geq 0$. Since 
$\phi_c\geq c$ and $\phi_c\not\equiv c$, by the strong maximum principle 
\ref{prop:StrongMaxPrinciple}, $\phi_c-c>0$, and so $\phi_c>c$.

A quick calculation shows that $\phi_1-1+\epsilon$ is a subsolution to 
\eqref{eq:ReducedLich} for any $0<\epsilon<1$. Using $\phi_c$ as a supersolution,
the sub and supersolution theorem \ref{thm:SubSupersolutionTheorem} combined with 
uniqueness from Theorem \ref{thm:LichUniqueness} show that 
$\phi_c\geq\phi_1-1+\epsilon$. Letting $\epsilon \to 0$, we have $\phi_c \geq\phi_1-1>0$.

Taking the variation of \eqref{eq:ReducedLich}, we find that $\delta \phi_c$ satisfies
\begin{equation}\label{eq:DerivPhi}
(-a\Delta + (q+1)|\sigma|^2 \phi_c^{-q-2}) \delta \phi_c = 0.
\end{equation} Since $\phi_c$ changes at a rate of one near infinity, 
we require $\delta \phi_c \to 1$ at infinity. By Proposition \ref{prop:Isomorphism},
$\delta\phi_c -1 \in W^{2,p}_\delta$. Then,
since $(-a\Delta + (q+1)|\sigma|^2 \phi_c^{-q-2}) (\delta \phi_c-1) \leq 0$,
the maximum principle \ref{prop:MaxPrinciple} shows that $\delta\phi_c-1\leq 0$,
and so $\delta\phi_c\leq 1$.

To show that $\delta\phi_c\geq 0$, note that for $c'>c$, $\phi_{c'}$ is a supersolution
for $\phi_c$, i.e., for the reduced Lichnerowicz equation \eqref{eq:ReducedLich}.
Thus $\phi_c$ is nondecreasing, and so $\delta\phi_c\geq 0$.

We claim that $\frac{q+1}{c} (c-\phi_c) +1$ is a subsolution of \eqref{eq:DerivPhi},
which then implies that it is a lower bound for $\delta\phi_c$.
Indeed,
\begin{equation}
\begin{aligned}
\left(-a \Delta+ (q+1)|\sigma|^2 \phi_c^{-q-2}\right)\left[\frac{q+1}{c} (c-\phi_c) +1\right]
   &=|\sigma|^2 \phi_c^{-q-1} \left[ -\frac{q+1}{c}(q+2) + \frac{(q+1)(q+2)}{\phi_c}\right]\\
   &\leq 0
\end{aligned}
\end{equation} since $\phi_c\geq c$. 
\end{proof}

With those properties of $\phi_c$, we can now understand how the integral term
in \eqref{eq:PhiMass} behaves as $c\to 0$ or $c\to \infty$. In this model problem,
the integral term, modulo a constant, becomes $c\int_M |\sigma|^2 \phi_c^{-q-1}$.

\begin{prop}\label{prop:MassLimits}
For all $c$ large enough, the integral term in \eqref{eq:PhiMass} strictly decreases
and approaches $0$ as $c\to \infty$.

If $\sigma$ has compact support, the integral term goes to zero as $c\to 0$.

If $|\sigma|^2 \geq C \rho^{\alpha}$ for some 
$\alpha > \frac{2n}{n-1}\delta - \frac{n}{n-1}$ and $C>0$, the integral term becomes
unbounded as $c\to 0$.
\end{prop}
\begin{remark}
The lower bound on $\sigma$ need not hold on all of $M$. Indeed, it is sufficient for
$\sigma$ to be bounded below only on some wedge of positive angle, perhaps far out on the
end. Also, note that, for $\alpha = 2\delta-2$, the inequality for $\alpha$ reduces
to $\delta <1-n/2$, which was already assumed.
Finally, note that the $\delta$ from the lower bound on $\alpha$ is the 
$\delta\in (2-n,0)$ from
the inequality $|\sigma|^2 \leq C \rho^{2\delta-2}$. In particular, the result can
hold even if $\sigma$ falls off faster than the metric.
\end{remark}

\begin{proof}
For all $c$,
\begin{equation}
0 < c \int_M |\sigma|^2 \phi_c^{-q-1} \leq c^{-q} \int_M |\sigma|^2,
\end{equation} and so the integral term approaches zero as $c\to \infty$.

The derivative of the mass as a function of $c$ is
\begin{equation} \label{eq:MassDeriv}
\frac{\p}{\p c} M_{ADM}(\gbar) = 2c M_{ADM}(g) 
  +C_0\int |\sigma|^2 \phi_c^{-q-1} \left[ 1- (q+1)\frac{c}{\phi_c} \delta \phi_c\right].
\end{equation} Since $c/\phi_c \to 1$ and $\delta \phi_c \to 1$ uniformly in $c$,
the integrand in \eqref{eq:MassDeriv} is negative for large $c$. Thus the integral term
decreases monotonically for all $c$ large enough.

Suppose $\sigma$ has compact support. Then
\begin{equation}
c\int_M |\sigma|^2 \phi_c^{-q-1} \leq c \int_M |\sigma|^2 (\phi_1-1)^{-q-1}.
\end{equation} Since $\phi_1-1>0$ and $\sigma$ has compact support, the integral term is
finite. Thus, as $c\to 0$, the integral term of the mass \eqref{eq:PhiMass} vanishes.

Suppose that $|\sigma|^2 \geq C \rho^{\alpha}$ for some 
$\alpha > \frac{2n}{n-1}\delta - \frac{n}{n-1}$ and $C>0$. 
Recall that by Equation \eqref{eq:PhiCBound} 
and Lemma \ref{lem:PhiUpperBound}, $\phi_c \leq c^{-q/2}C\rho^{\eta+2} + c$,
where $\eta := 2\delta -4 +n$.

Dropping all constants not depending on $c$, one has on an end $E$,
\begin{align}
\int_E |\sigma|^2 \phi_c^{-q-1} dV &\geq \int_E \frac{\rho^{\alpha+n-1}}{(c^{-q/2}\rho^{\eta+2}+c)^{q+1}} d\rho d\sigma\\
&= c^{-q-1 + \frac{(1+q/2)(\alpha +n)}{\eta+2}} \int^\infty_{C c^{\frac{1+q/2}{-\eta-2}}} \frac{r^{\alpha+n-1}}{(1+r^{\eta+2})^{q+1}} dr
\end{align} The first line is true for some integration form $d\sigma$. For the second line,
we pulled out $c$'s and used the substitution $\rho = r c^{\frac{1+q/2}{\eta+2}}$, and integrated
out the $d\sigma$. Since the final integral converges as $c\to 0$, we may bound it
by a constant. Thus,
\begin{equation}
\int_E |\sigma|^2 \phi_c^{-q-1} dV \geq C_0  c^{-q-1 + \frac{(1+q/2)(\alpha +n)}{\eta+2}}.
\end{equation} If $ -q-1 + \frac{(1+q/2)(\alpha +n)}{\eta+2}<-1$, then the integral
term of the mass \eqref{eq:PhiMass} blows up as $c\to 0$. Using $\eta = 2\delta -4+n$, this 
condition reduces to $\alpha > \frac{2n}{n-1}\delta - \frac{n}{n-1}$. 
%\mnote{JD: This is page 49 of my notes}
This establishes the final claim of the proposition.
\end{proof}

While the original hope was that mass and asymptotic constants were in one to one
correspondence, Proposition \ref{prop:MassLimits} unfortunately shows that this is not
the case. If $M_{ADM}(g) =0$, and there are $\sigma$ which have compact support, then
for those $\sigma$, the mass is not a monotonic function of $c$. For any positive mass,
since $c^2 M_{ADM}(g)$ is zero as $c\to 0$ and unbounded as $c\to \infty$, for 
$\sigma$ which do not fall off very quickly, the mass is again not monotonic as a
function of $c$. 
%%%APPENDICES

%%BIBILIOGRAPHY
\bibliographystyle{alpha}

\bibliography{ThesisBib}
\end{document}